\documentclass[11pt]{article}
\usepackage[utf8]{inputenc}
\usepackage[margin=1in,footskip=0.25in]{geometry}
\usepackage{amsmath}
\usepackage{mdframed}
\usepackage{mathtools}
\usepackage{bbm}
\usepackage{titling}
\usepackage{parskip}
\usepackage{amsthm}
\usepackage{amssymb,amscd,amsfonts,amsbsy}
\usepackage{hyperref}
\usepackage{setspace}
\usepackage[
  backend=biber,
  style=apa,
  natbib=true,
  maxcitenames=2,
  mincitenames=1,
  maxbibnames=99,
  giveninits=true,
  uniquename=false,
  doi=false,
  url=false,
  isbn=false,
  eprint=false
]{biblatex}
\usepackage{graphicx}
\usepackage{caption}
\usepackage{subcaption}
\usepackage{float} 
\usepackage{setspace}
\usepackage{makecell}
\usepackage{array}
\usepackage{tikz}
\usetikzlibrary{positioning,arrows.meta}
\usepackage{authblk}
\usepackage{float}
\usepackage{titletoc}
\usepackage{booktabs}
\usepackage{makecell}
\usepackage{array}
\usepackage{listings}
\usepackage{xcolor}
\usepackage{tabularx}
\usepackage{booktabs}
\usepackage{longtable}
\usepackage{needspace}

\AtEveryBibitem{\clearfield{pages}}

\definecolor{codebackground}{gray}{0.95}

\tikzset{
  var/.style = {circle, draw, thick, inner sep=1.6pt, minimum size=18pt},
  >=Latex, node distance=14mm and 20mm
}

\DeclareNameAlias{author}{family-given}
\DeclareNameAlias{editor}{family-given}
\DeclareNameAlias{translator}{family-given}

\hypersetup{
    colorlinks=true,
    linkcolor=blue,
    filecolor=magenta,      
    urlcolor=cyan,
    pdftitle={Overleaf Example},
    pdfpagemode=FullScreen,
    }

\newcommand{\subtitle}[1]{%
  \posttitle{%
    \par\end{center}
    \begin{center}\large#1\end{center}
    \vskip0.5em}%
}

\usepackage{graphicx}
\graphicspath{ {images/} }
\newtheorem{remark}{Remark}

\newtheorem{proposition}{Proposition}

\title{Two-sample inference for sequence-derived feature distributions from deep sequencing data}

\author[1]{James Peng\thanks{Email: \href{mailto:james.s.peng@berkeley.edu}{james.s.peng@berkeley.edu}}}\author[2]{Florian Stijven}
\author[3]{Michal Juraska}
\author[3]{Allan C. deCamp}
\author[3]{Elena E. Giorgi}
\author[4]{Linbo Wang}
\author[1,4]{Peter B. Gilbert}

\affil[1]{Department of Biostatistics, University of Washington, Seattle, WA, USA}
\affil[2]{I-BioStat, KU Leuven, Leuven, Belgium}
\affil[3]{Vaccine and Infectious Disease Division, Fred Hutchinson Cancer Center, Seattle, WA, USA}
\affil[4]{Department of Statistical Sciences, University of Toronto, Toronto, Canada}

\begin{document}

\maketitle

\begin{abstract}
\begin{singlespace}
Many biological contexts involve heterogeneous populations of related molecular variants within a single individual. For example, some viral infections such as HIV-1 exhibit substantial diversity within a single host. Advances in sequencing technology now make it possible to characterize such populations by generating up to thousands of genetic sequences from a biological sample. In many settings, each sequence can be mapped to a numeric feature, such as Hamming distance to a reference sequence. This collection of feature values defines a one-dimensional distribution that characterizes individual-level viral diversity. When comparing these features across independent groups, past analyses have typically reduced each individual’s data to a single summary value, such as the mean or mode. We propose a framework that retains the full individual-level distribution as the object of analysis. Drawing on existing tools from distributional data analysis, our framework uses Wasserstein barycenters to summarize the observed  individual-level distributions within groups into a single representative distribution. These group barycenters are then compared using differences in quantiles and tested for equality. A practical challenge is that the number of sequences obtained, known as sequencing depth, can vary across biological samples. Through simulation studies, we examine when low sequencing depth can distort inference and present practical strategies for addressing these issues. We then illustrate the framework using data from the HIV-1 Antibody Mediated Prevention efficacy trials and a human B cell receptor sequencing study of vaccine-elicited immune responses. By using the entire collection of sequences for each individual, our approach can reveal group differences that may be obscured by single-number summary analyses.
\end{singlespace}
\end{abstract}

\section{Introduction}

Many biological studies seek to characterize heterogeneous populations of related molecular variants within a single individual. For example,  rapidly evolving RNA viruses, such as HIV-1 and hepatitis C, can have a diverse population within each host, referred to as a quasispecies. This diversity is generated by within-host evolutionary processes such as mutation, selection, genetic drift, recombination, and transmission bottlenecks. It can influence clinically important outcomes, including immune escape, drug resistance, viral persistence, disease progression, and transmission \citep{domingo_viral_2012, sim_tracking_2015, nelson_within-host_2015, dinis_deep_2016}. Another example occurs in the human adaptive immune repertoire, which contains a vast array of B cell receptors. Following infection or vaccination, an individual's B cell lineages can undergo clonal expansion and somatic hypermutation, generating diverse collections of related receptors \citep{nouri_somatic_2020}.

To study this diversity, researchers collect biological samples, such as blood or tissue, from individuals and sequence the molecular variants present in those samples. Historically, such diversity has been difficult to characterize with high resolution because older approaches, such as single genome amplification followed by Sanger sequencing, are resource-intensive and typically yield up to 20-30 sequences per biological sample. In recent years, \textit{deep sequencing} technologies have transformed this setting by generating large collections of sequences per sample, enabling more detailed study of heterogeneous populations within individuals \citep{nelson_within-host_2015, boyd_deep_2016, lauring_within-host_2020}.

Direct analysis of raw sequence data is challenging because the number of possible sequences grows rapidly with sequence length, and observed sequences often vary substantially across individuals. To make these analyses more tractable, sequences are often summarized using lower-dimensional features. These can include discrete indicators for specific amino acid residues or motifs, classification into broad sequence types or lineages, or quantitative numeric scores. Our proposal focuses specifically on the third type: numeric sequence features. Examples include the Hamming distance to a reference sequence or sequence-specific phenotypic properties (either experimentally derived or predicted from sequence). Once individual sequences are mapped to numeric values, a common analytical task is to compare these sequence-derived features across treatment conditions or demographic factors such as geographic region. In the HIV-1 sequencing context, such analyses have included quantifying HIV-1 Env V1V2 loop length and predicted N-linked glycosylation sites across stages of disease progression \citep{curlin_hiv-1_2010} and comparing predicted transmissibility across modes of sexual transmission \citep{tully_differences_2016}. In B cell receptor studies, analyses have compared levels of somatic hypermutation and distance from inferred germline sequences between treatment arms \citep{willis_vaccination_2025}.


With deep sequencing data, the collection of sequences from a single individual yields a distribution of numeric feature values. Past analyses have often collapsed these data into a single representative sequence or summarized the feature distribution by a single scalar value, such as the mean or mode. This reduction simplifies group comparisons by allowing for classical two-sample procedures but, as we will argue, may fail to detect differences between groups not captured by the chosen summary. In contrast, we propose retaining each individual's complete collection of feature values. Thus, the data unit for each individual is the distribution itself. 

Because distributional observations cannot be averaged, differenced, or tested in the same way as real-valued observations, analyzing them requires additional statistical tools. A growing body of literature has developed tools for data in which the observational unit is a distribution \citep{agueh_barycenters_2011, chen_wasserstein_2023, lin_causal_2023}. A central object in this work is the \textit{Wasserstein barycenter}, which provides a notion of an average for a set of distributions. Our aim is to translate ideas from this literature into a practical framework for biological sequencing studies. We describe how to compare two independent groups by treating each individual’s sequence-derived features as a distribution, summarizing each group with a Wasserstein barycenter, and evaluating group differences through pointwise differences in their barycenter quantile functions. Furthermore, we develop hypothesis tests for the equality of group barycenters, designed to detect different classes of alternatives that may arise in practice. The proposed methods are implemented in a \textsf{R} package, \texttt{mcbarycenter}.

A complication in this approach is that individual-level feature distributions are observed at varying resolutions. The number of sequences obtained from a biological sample---known as \textit{sequencing depth}---can differ substantially across individuals, depending on the laboratory sequencing tool used and potentially reflecting differences in the abundance of the target molecules in the sample. While deep sequencing can theoretically yield thousands of sequences per sample, the realized sequence count often varies substantially and may be low for certain samples \citep{sims_sequencing_2014}. With low sequencing depth, the observed feature values provide a limited view of the distribution of the feature within the underlying biological population. This becomes especially problematic in comparisons when sequencing depth varies systematically between groups, which could happen if the amount of the biological target systematically differs between the groups being compared.

HIV-1 prevention trials provide a motivating example in which these sequencing depth issues arise. In these studies, participants are randomized to receive either an active intervention or a placebo and are followed for incident HIV-1 infection. For participants who acquire HIV-1, blood samples are collected to generate viral sequences. These sequences are then compared between study arms to evaluate whether the intervention is associated with differences in the viral variants observed among individuals who acquire the virus \parencite{gilbert_sieve_2001}. Sequencing depth in these datasets can vary depending on the sequencing technology and processing pipeline used, and may also reflect the amount of viral genetic material available in a sample. Thus, if an intervention affects post-infection viral load, it may induce systematic differences in sequencing depth between study arms. 

We highlight three distinct examples of trials with different sequencing depths: the Phambili (HVTN 503), AMP (HVTN 703/704), and Imbokodo (HVTN 705) trials (Figure \ref{fig:study_seq_depth}; Supplemental Table \ref{tab:hiv_studies}). The Phambili study, which used older Sanger sequencing technology, has very low sequencing depths that are distributed similarly across arms
(Vaccine: median 5 [range 2--12]; Placebo: 6 [5--12]). In contrast, the AMP study, which used PacBio deep sequencing, exhibits a wide range of sequencing depths with an overall median of 141.5 sequences per person (range: 1--988), relatively balanced across arms (Treatment: 134 [1--672]; Placebo: 157 [4--988]). Finally, the Imbokodo trial, which also used deep sequencing, demonstrates a scenario with systematic differences between groups. While the overall median is 121 sequences (range: 1--589), vaccine recipients tended to have lower sequencing depths than placebo recipients (Vaccine: 96 [2--493]; Placebo: 146 [1--589]). These three examples illustrate that sequencing depth is not only highly variable across studies, but can also be affected by the intervention itself. We will study the impact of varying sequencing depth on our framework through simulation studies and provide practical guidance for mitigating this potential issue.

\begin{figure}[!htbp]
    \centering
    \includegraphics[width=1\textwidth]{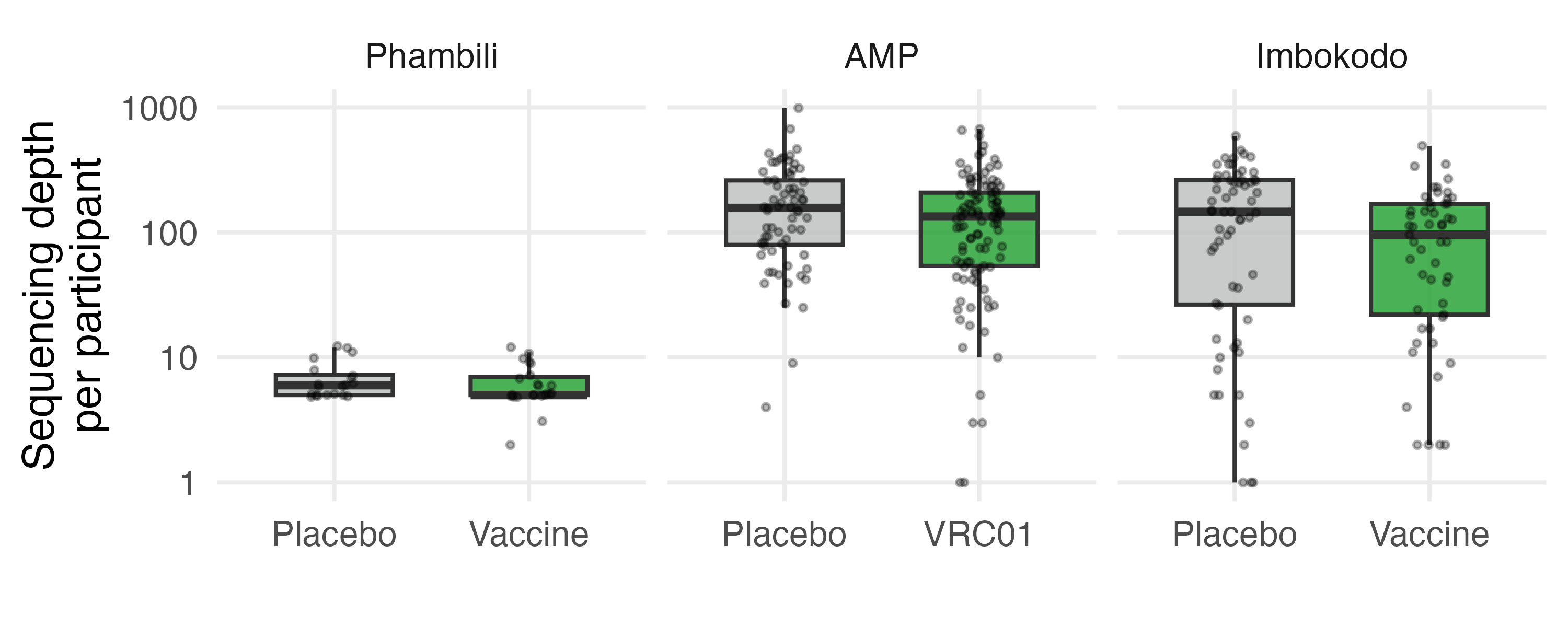}
    \caption{HIV-1 Envelope sequencing depth of participants acquiring HIV-1 infection for three HIV-1 prevention trials---Phambili (HVTN 503), AMP (HVTN 703/704), and Imbokodo (HVTN 705) ---stratified by study arm. The y-axis is shown on the log10 scale.}
    \label{fig:study_seq_depth}
\end{figure}

To our knowledge, Wasserstein barycenter-based distributional methods have not previously been applied to the analysis of deep sequencing data. A related paper applies a similar framework in genetics research of single-cell differential expression \citep{zhang_bsde_2022}. Their method uses Wasserstein barycenter techniques to analyze single-cell sequencing data in a two independent sample problem, comparing cases versus controls. While relying on similar ideas, our proposal differs in a few ways: first, we explicitly consider the potential for low or differential sequencing depth to produce biased estimates. Second, we consider an estimand---the difference in quantiles---that contrasts the  Wasserstein barycenters of the two groups, whereas they consider the Wasserstein distance between the two barycenters as a single summary measure. Third, we approach hypothesis testing differently. Their permutation testing procedure evaluates the null of exchangeability between two groups, which may be automatically violated if the distribution of sequencing depths differs systematically between groups. In contrast, our testing approach can be valid even when sequencing depths differ between groups.

The remainder of the paper is organized as follows. Section~\ref{sec:statistical_tasks} provides a high-level introduction to the statistical tasks for our proposed analytic approach and presents illustrative toy examples to build intuition. Section \ref{sec:methodology} describes the data structure, estimands, and methodology for estimating the group-specific Wasserstein barycenters and testing the hypothesis of equal barycenters between groups. Section \ref{sec:simulations} presents results from simulation studies. Section \ref{sec:application} presents the framework applied to the AMP HIV-1 prevention trials and human B cell receptor sequencing study of vaccine-elicited immune responses. Section \ref{sec:discussion} concludes.

\section{Statistical tasks with distributional data}\label{sec:statistical_tasks}

To build intuition, we first introduce the framework's statistical tasks at a high level and provide illustrative toy examples. The statistical tasks include:

\paragraph{Task 1: Average distribution.} First, we define what constitutes an ``average" distribution across individuals. We use the \textit{Wasserstein barycenter}, which provides a notion of averaging distributions under optimal transport geometry \citep{agueh_barycenters_2011}. For a more comprehensive overview of Wasserstein barycenters, we refer the reader to the review article by \citet{bigot_statistical_2020}.

Compared to pointwise averages of density functions, the Wasserstein barycenter can better preserve the salient features of component distributions. In Figure \ref{fig:barycenter_quantile}(a)-(b), we display the densities of five bimodal distributions of the same shape along with their Wasserstein barycenter; as shown, the barycenter has the same bimodal shape as the five input distributions. 
Additionally, for one-dimensional distributions, the Wasserstein barycenter has a convenient and interpretable representation: its quantile function is equal to the pointwise average of the quantile functions of its component distributions \citep{bigot_characterization_2018} (Figure \ref{fig:barycenter_quantile}(c)-(d)). Given this property, we frame our presentation around quantile functions of distributions rather than their densities. In practice, we summarize the within-host feature distributions for each group using their mean quantile function or equivalently, the quantile function of their Wasserstein barycenter.

\begin{figure}[ht!]
    \centering
    \includegraphics[width=\textwidth]{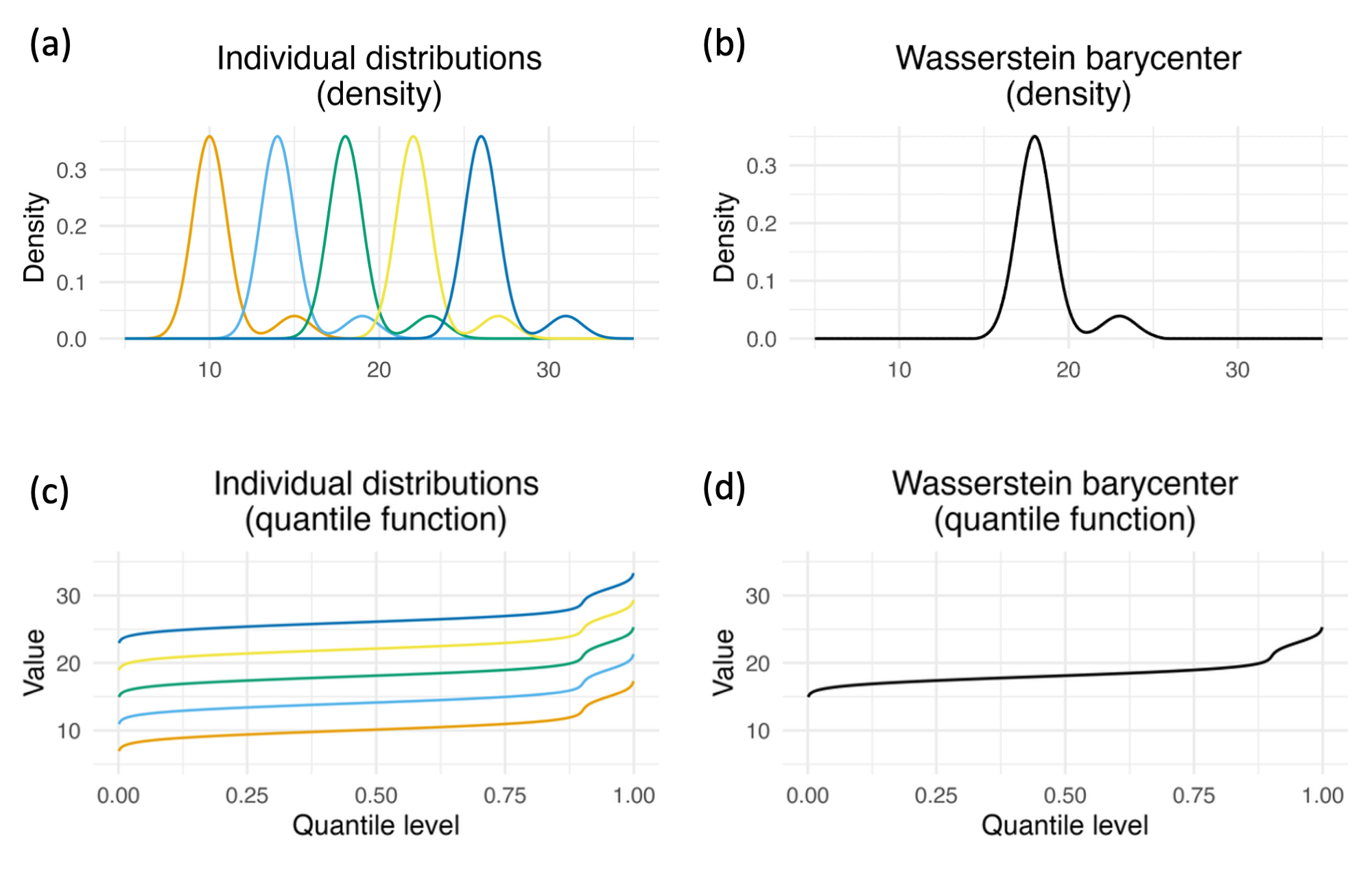}
    \caption{(a) Densities of five bimodal distributions and (b) the density of their Wasserstein barycenter. (c) Quantile functions of the same five bimodal distributions and (d) the quantile function of their Wasserstein barycenter. The Wasserstein barycenter quantile function is obtained by taking the pointwise mean of the input quantile functions.}    \label{fig:barycenter_quantile}
\end{figure}

\paragraph{Task 2: Differences between barycenter distributions.} After estimating the barycenter for each group, we quantify how the barycenter distributions differ between groups. Analyses that reduce a distribution to a single scalar such as the mean or mode often miss important differences in shape or tail behavior. Instead, we employ a \textit{quantile difference map}, which compares two distributions by evaluating the differences between them across quantile levels in $(0, 1)$. Comparing quantiles is a well-established method for contrasting distribution functions \citep{doksum_empirical_1974}, interpreting quantile regression \citep{koenker_regression_1978}, and studying causal treatment effects for both real-valued and distributional outcomes \citep{firpo_efficient_2007, lin_causal_2023, katta_interpretable_2024}. Such a comparison allows us to more thoroughly characterize differences between barycenter distributions; a toy example is shown in Figure \ref{fig:quantile_diff}.

\begin{figure}[ht!]
    \centering
    \includegraphics[width=\textwidth]{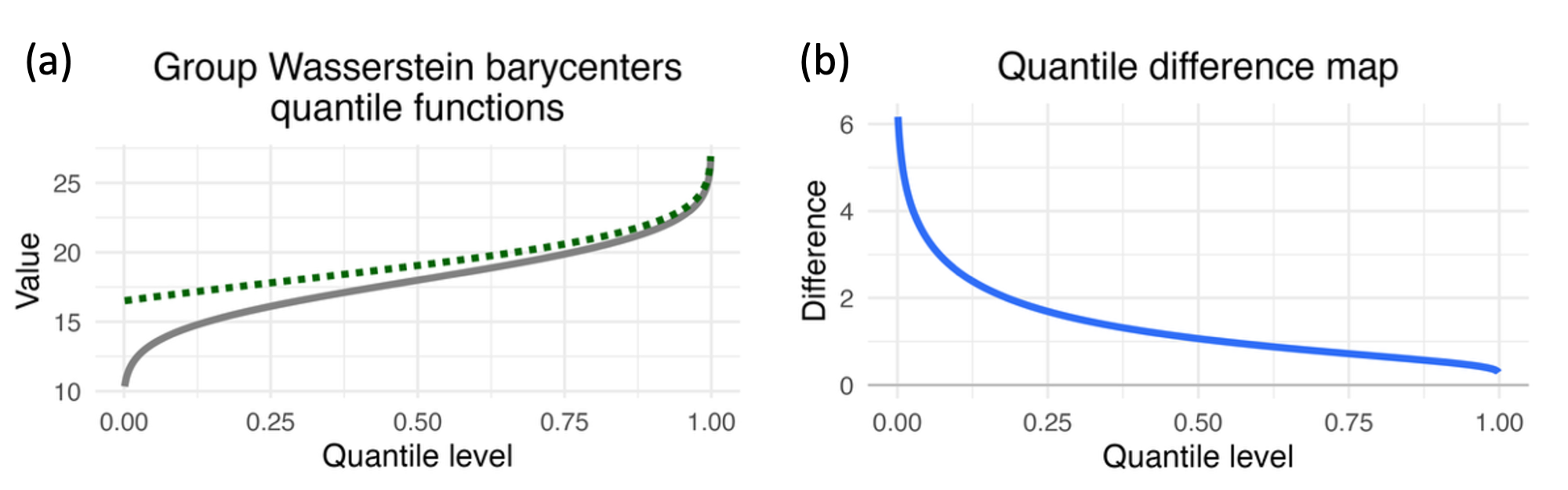}
    \caption{(a) Two group Wasserstein barycenter quantile functions (Group 1: solid gray; Group 2: dotted green), and (b) their associated quantile difference map (comparing Group 2 versus Group 1).  In this example, the 0.01 quantile of the Group 2 barycenter, representing a near-minimum sequence feature value, is approximately six units higher than that of Group 1. In contrast, the median feature value in Group 2 is only approximately one unit higher than in Group 1. In the upper tail, the two barycenters converge, with their maximum feature values being approximately equal.}    \label{fig:quantile_diff}
\end{figure}

\paragraph{Task 3: Hypothesis testing for equal Wasserstein barycenters.} Finally, we require a hypothesis test to determine whether the Wasserstein barycenters differ between groups. Specifically, we test the null hypothesis that the Wasserstein barycenters of the within-host viral feature distributions are identical across groups. Equivalently, this null asserts that the difference in barycenter quantiles between the two groups is zero across all quantile levels in $(0, 1)$.

By representing the barycenter distribution via its quantile function, this problem can be viewed through the lens of functional data analysis, where tests evaluating equality of functionals means between groups are well-studied \citep{horvath_inference_2012}. For distributional data specifically, two-sample tests have been explored by a handful of past works. \citet{dubey_frechet_2019} propose an ANOVA-type test evaluating the equality of barycenters and Fréchet variances; \citet{zhang_bsde_2022} use  a permutation test for the strong null of exchangeability between groups; and \citet{jeong_wasserstein_2026} derive the asymptotic limiting distribution of a scaled, squared Wasserstein distance test statistic. Additionally, \citet{petersen_wasserstein_2021} examine testing in a Wasserstein regression context, which, in its simplest form, can be reduced to a two-sample barycenter problem. Our approach aligns most closely with \citet{jeong_wasserstein_2026}. However, we propose more flexible test statistics designed to target various alternative data scenarios that could arise in practice.

\subsection{Illustrative toy examples}

To build intuition, Figure \ref{fig:toy_scenarios} illustrates four toy scenarios alongside the corresponding outputs from our framework, which we refer to as a \textit{distributional analysis}. In each scenario, the left panel shows the individual-level feature distributions and their group-specific Wasserstein barycenters, the middle panel shows the corresponding quantile functions, and the right panel shows the resulting difference in barycenter quantiles between groups. We contrast this approach with a \textit{mean-based analysis}, which relies on a classical two-sample mean comparison after computing the average feature value for each individual.

\begin{enumerate}
\item[(a)] \textit{Similar distributions (Figure \ref{fig:toy_scenarios}(a)).} Individuals in Group 1 exhibit sequence feature distributions similar to those in Group 2, resulting in similar barycenters. The distributional analysis estimates negligible differences between the two group barycenters, producing a quantile difference curve nearly flat at 0 across all quantile levels. Similarly, the mean-based analysis estimates a near-zero difference between group means.

\item[(b)] \textit{Uniform differences (Figure \ref{fig:toy_scenarios}(b)).} Group 2 individuals generally exhibit higher sequence feature values than Group 1, resulting in a uniform difference between the two group barycenters across quantile levels. Consequently, the distributional analysis yields an approximately uniform positive difference across all quantiles. The mean-based analysis similarly estimates a positive difference in means of roughly the same magnitude.

\item[(c)] \textit{Upper-tail difference (Figure \ref{fig:toy_scenarios}(c)).}  Group 2 individuals exhibit more extreme feature values in the upper tail, but the distributions otherwise resemble those from Group 1. Therefore, for most quantile levels, the Group 1 and Group 2 barycenters are similar, but they diverge at the upper tail. The distributional analysis captures this by producing near-zero estimates for most quantile levels but a positive difference at the upper tail. In contrast, the mean-based analysis yields an attenuated estimate, which averages this positive tail difference with the near-zero difference in other parts of the distribution.

\item[(d)] \textit{Difference in variability (Figure \ref{fig:toy_scenarios}(d)).} Group 1 individuals tend to have distributions with more variability compared to Group 2. Therefore, the Group 1 barycenter has a higher variance than the Group 2 barycenter, despite sharing the same overall mean. Here, the distributional analysis reveals a positive difference at the lower tail, a negative difference at the upper tail, and an approximately zero difference at the 0.5 quantile. In contrast, the mean-based analysis estimates a null difference between group means.
\end{enumerate}

In scenarios (a) and (b), both the mean-based and distributional  approaches yield similar results and may have comparable power to detect group differences. In scenario (c), however, the distributional approach can be more sensitive to the upper-tail difference and also characterizes a feature of the barycenter difference that the mean-based analysis obscures. Finally, in scenario (d), the mean-based approach will fail to detect the difference in variability, whereas the distributional approach can capture this difference. Together, these examples highlight two advantages of our proposal: it can have power to detect distributional differences that mean-based analyses may miss, and it provides a higher resolution comparison between groups.

\begin{figure}[!htpb]
    \centering
    \includegraphics[width=0.95\textwidth]{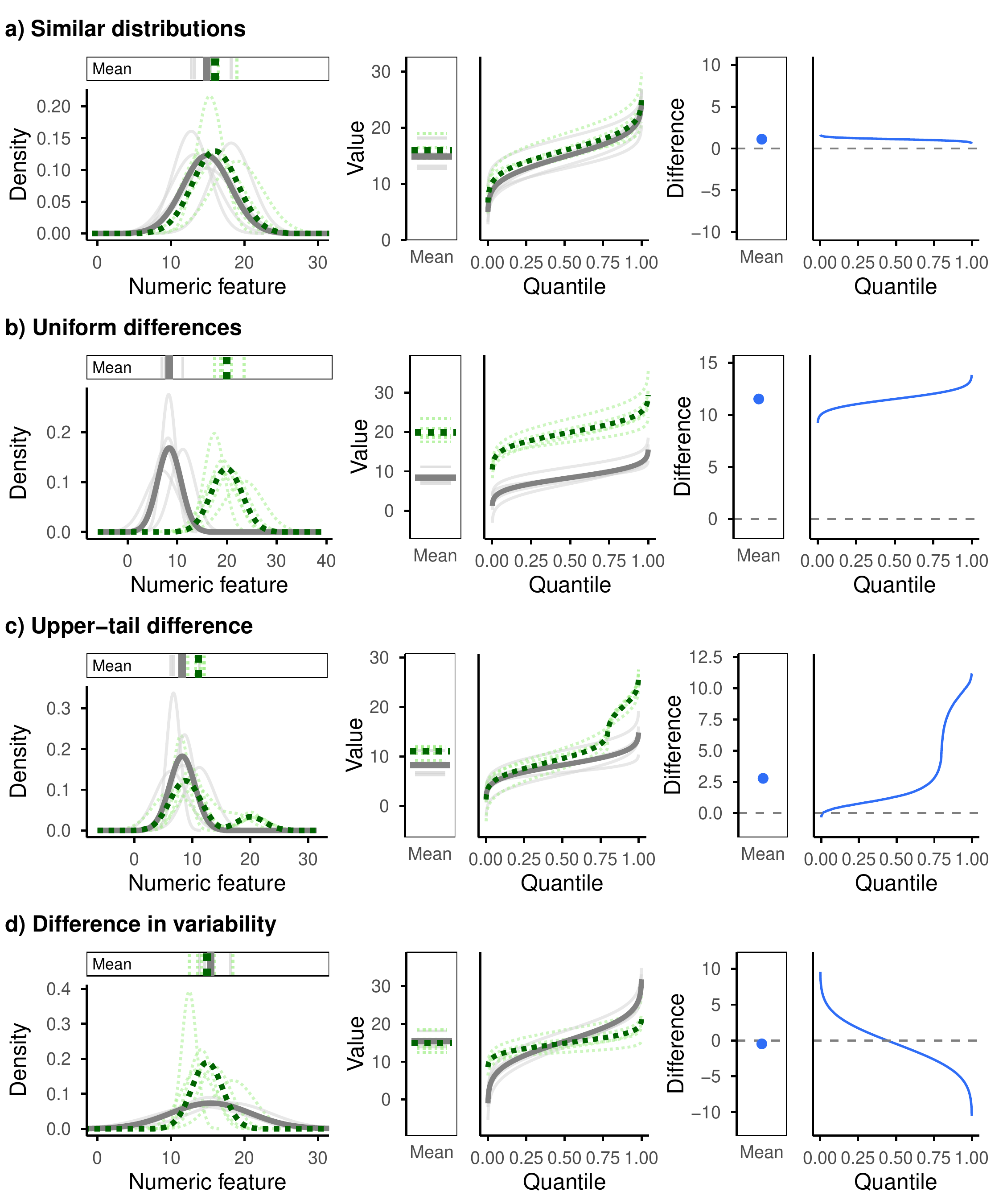}
    \caption{Output of a distributional analysis for four hypothetical toy datasets. In each scenario, there are five participants in each group. In the left panels, light curves show the densities of numeric features for each individual (Group 1, gray solid; Group 2, green dotted), while bold curves represent the group-specific Wasserstein barycenters. The middle panels show the corresponding individual-level and group barycenter quantile functions. The right panels display the difference in quantiles between the Group 2 and Group 1 barycenters. The inner panels labeled ``Mean" show the corresponding mean-based analysis, comparing the average feature values between groups. Panels illustrate the following scenarios: (a) similar distributions, (b) approximately uniform differences, (c) upper-tail difference, and (d) difference in variability.}
\label{fig:toy_scenarios}
\end{figure}

\section{Methodology} \label{sec:methodology} 

\subsection{Data setup}

For each individual $i$, let $G_i \in \{1,2\}$ denote the binary group indicator and $Y_i$ denote the individual's underlying within-host distribution of a numeric sequence feature (continuous or discrete). Note that $Y_i$, which is a realization of a \textit{distribution-valued random object}, is not directly observed; we instead observe $m_i$ measurements $Y_{i1},\ldots,Y_{im_i}$ sampled from this distribution, where $m_i$ denotes the sequencing depth. We assume that, conditional on $Y_i$, $Y_{i1},\ldots,Y_{im_i} \stackrel{\mathrm{i.i.d.}}{\sim} Y_i$. We allow the distribution of sequencing depth $m_i$ to differ in each group, so that $m_i \mid G_i=g \sim \eta_g$ for $g \in \{1,2\}$, where $\eta_g$ denotes the distribution of sequencing depths for Group $g$. In total, we observe $n$ independent individuals in our dataset, with $n_1$ in Group 1 ($G = 1$) and $n_2$ in Group 2 ($G = 2$), where observations are identically distributed within each group.

From the observed measurements for individual $i$, we can construct the empirical distribution $\widehat{Y}_i^{\mathrm{emp}} = \frac{1}{m_i}\sum_{j=1}^{m_i}\delta_{Y_{ij}}$. We also consider a smoothed version of this empirical distribution, denoted $\widehat{Y}_i^{\mathrm{sm}}$, obtained by applying a smoothing procedure to $\widehat{Y}_i^{\mathrm{emp}}$, such as kernel density smoothing or smoothing of the quantile function \citep{bigot_upper_2018}. Both $\widehat{Y}_i^{\mathrm{emp}}$ and $\widehat{Y}_i^{\mathrm{sm}}$ can be viewed as error-prone proxies for $Y_i$, with the magnitude of error depending on the sequencing depth $m_i$. Because the sequencing-depth distributions $\eta_1$ and $\eta_2$ may differ between groups, the amount of measurement error may also differ between groups. As shown in the following sections, this measurement error can bias estimation of the target estimands.

For a generic probability measure $\lambda$ on the real line, let $F_{\lambda}(y)$ denote its cumulative distribution function for $y \in \mathbb{R}$, and let $F_{\lambda}^{-}(t)$ denote its quantile function for $t \in (0,1)$. Because the individual-level outcome is distribution-valued, we measure distances between distributions using the $2$-Wasserstein metric. In optimal transport, the Wasserstein distance quantifies the minimum cost of transporting probability mass from one distribution to another; in the case of the 2-Wasserstein distance on Euclidean space, this cost is based on squared Euclidean distance. For one-dimensional distributions, the $2$-Wasserstein distance between $\lambda_1$ and $\lambda_2$, denoted by $W_2(\lambda_1, \lambda_2)$, has a convenient representation in terms of their quantile functions:
\begin{equation}
W_2(\lambda_1,\lambda_2)
:=
\left(
\int_0^1
\left|
F_{\lambda_1}^{-}(t)-F_{\lambda_2}^{-}(t)
\right|^2
\,dt
\right)^{1/2}.
 \label{eq:w2_distance}
\end{equation}

Using this metric, we define the Wasserstein barycenter of the distribution-valued random object $Y$ as the probability measure $\mathbb{B}Y$ that minimizes the expected squared Wasserstein distance between a random realization of $Y$ and a candidate probability measure $\nu$:
\begin{equation*}
\mathbb{B}Y
:=
\arg\min_{\nu}
\mathbb{E}
\left[
W_2^2(Y,\nu)
\right],
\end{equation*}
where the minimum is taken over probability measures $\nu$, and the expectation is taken with respect to the randomness in $Y$. We denote the corresponding quantile function by  $F_{\mathbb{B}Y}^{-}(t)$. More generally, for a generic random variable $X$, we define the conditional Wasserstein barycenter of $Y$ given $X$ by
\begin{equation*}
\mathbb{B}[Y \mid X]
:=
\arg\min_{\nu}
\mathbb{E}
\left[
W_2^2(Y,\nu)\mid X
\right].
\end{equation*}

By the one-dimensional representation of the Wasserstein distance in \eqref{eq:w2_distance}, the quantile function of the Wasserstein barycenter is given by the pointwise mean of the quantile functions of the component distributions:
\begin{equation}
F_{\mathbb{B}Y}^{-}(t)
=
\mathbb{E}\left[F_Y^{-}(t)\right],
\qquad
t \in (0,1).
\label{eq:barycenter_quantile_translation}
\end{equation}

Because of this representation, we frame our methodology around estimating the Wasserstein barycenter's quantile function. In Task 1, we wish to estimate the group-specific Wasserstein barycenter quantile function
\begin{equation}
\mu_g(t)
:=
F^{-}_{\mathbb{B}[Y \mid G=g]}(t),
\qquad
t \in (0,1), \quad g \in \{1,2\}.
\end{equation}

In Task 2, we estimate the quantile difference map, which compares the quantiles of the group-specific Wasserstein barycenters between the two groups:
\begin{equation}
\Delta(t)
:=
\mu_2(t)-\mu_1(t),
\qquad
t \in (0,1).
\label{eq:statistical_estimand}
\end{equation}

Finally, in Task 3, our hypothesis test evaluates the global null hypothesis that the two group-specific barycenter quantile functions are equal at every quantile level in a compact set $\mathcal{T} \subset (0,1)$:
\begin{equation}
H_0:\Delta(t)=0 \text{ for all } t \in \mathcal{T}.
\end{equation}

\begin{remark}[Simple random sample assumption] \label{remark:simple_random_sample}
    The assumption $Y_{i1},\ldots,Y_{im_i} \stackrel{\mathrm{i.i.d.}}{\sim} Y_i$ may not hold depending on the sequencing technology and processing procedure used. For example, PCR resampling, amplification bias, sequencing error, and bioinformatic filtering can cause the observed sequences to deviate from independent draws from the biological viral population. In such settings, we assume that $Y_{i1},\ldots,Y_{im_i} \stackrel{\mathrm{i.i.d.}}{\sim} \tilde{Y}_i$, where $\tilde{Y}_i$ can be interpreted as the distribution of sequence-derived features after the sequencing and processing procedure, rather than the underlying biological viral population. Molecular barcoding approaches, such as Primer ID and related unique molecular identifier-based methods, can reduce PCR resampling artifacts and sequencing errors, improving the extent to which observed sequences reflect the viral template population present in the biological sample \citep{zhou_unique_2022}.
\end{remark}

\subsection{Task 1: Estimation and inference for the Wasserstein barycenter}

We aim to estimate the group-specific barycenter quantile functions $\mu_1(t)$ and $\mu_2(t)$. We first consider two estimators based on the empirical distributions $\widehat{Y}^{\mathrm{emp}}_i$ and $\widehat{Y}^{\mathrm{sm}}_i$ introduced above. Let
$$
Q_i^{\mathrm{emp}}(t) := F^{-}_{\widehat{Y}^{\mathrm{emp}}_i}(t),
\qquad
Q_i^{\mathrm{sm}}(t) := F^{-}_{\widehat{Y}^{\mathrm{sm}}_i}(t),
\qquad t \in (0,1),
$$
denote the empirical and pre-smoothed empirical quantile functions for individual $i$, respectively. The first estimator is the \textit{empirical} Wasserstein barycenter, which averages the empirical quantile functions within each group. Specifically, for $g \in \{1,2\}$, we define
\begin{equation}
\widehat{\mu}^{\mathrm{emp}}_g(t)
:=
\frac{1}{n_g}
\sum_{i=1}^n \mathbbm{1}\{G_i=g\}\,Q_i^{\mathrm{emp}}(t),
\qquad t \in (0,1).
\label{eq:empirical_barycenter}
\end{equation}

The second is a \textit{pre-smoothed empirical} Wasserstein barycenter, obtained by replacing each individual's empirical distribution with its smoothed version before averaging the corresponding quantile functions:
\begin{equation}
\widehat{\mu}^{\mathrm{sm}}_g(t)
:=
\frac{1}{n_g}
\sum_{i=1}^n \mathbbm{1}\{G_i=g\}\,Q_i^{\mathrm{sm}}(t),
\qquad t \in (0,1).
\label{eq:presmoothed_barycenter}
\end{equation}

These two empirical estimators perform well when sequencing depths are uniformly high, because each individual's empirical quantile function is an accurate approximation to their true unobserved quantile function. When sequencing depth is low, however, the empirical quantile function for each individual can be biased as an estimator of the corresponding true quantile function, and this can be especially pronounced in the tails \citep{bobkov_one-dimensional_2019}. As a result, averaging these biased quantile functions across individuals produces a biased estimator of the group-specific barycenter quantile function, which does not go away with increasing sample size \citep{bigot_upper_2018, peng_estimating_2026}. Although pre-smoothing can reduce sampling noise, it will not generally remove this bias. 

Motivated by this problem, \citet{peng_estimating_2026} studied Wasserstein barycenter estimation under sparse sampling, where each individual distribution is observed through only a small number of measurements. To address the limitations of individual-level empirical quantile estimation in this setting, the paper developed the marginal-constructed barycenter (MCB) estimator, which uses information pooled across individuals rather than attempting to accurately reconstruct each individual distribution separately. The approach relies on a non-informative sequencing depth assumption, meaning that sequencing depth for an individual is assumed to be independent of that individual's feature distribution, conditional on group. The MCB estimator is based on an alternative representation of the barycenter estimation problem through binomial mixture models, and can be configured with different structural models for these mixtures. While the asymptotic properties require that the binomial mixing distributions be correctly specified, more flexible implementations, such as spline-based mixture models, can be used in practice. (With these flexible implementations, misspecification will still be present; however, its effect on estimation and inference may be minor.) We denote the corresponding group-specific barycenter quantile estimator by $\widehat{\mu}_g^{\mathrm{mcb}}(t)$.

The three estimators are therefore suited to datasets with different sequencing depths. The empirical and pre-smoothed empirical barycenters are most appropriate when sequencing depths are uniformly large enough to support reliable estimation of individual-level quantiles. The MCB estimator is instead designed for sparse sampling settings, where individual-level distributions may be poorly estimated due to small within-subject sample sizes. In this case, the MCB estimator replaces reliance on accurate individual-level distribution estimation with structural assumptions on the marginal mixture models and on non-informative cluster sizes. We evaluate the performance of each estimator across datasets with different sequencing depths in our simulation studies in Section \ref{sec:simulations}.

Each estimator---$\widehat{\mu}^{\mathrm{emp}}_g(t)$, $\widehat{\mu}^{\mathrm{sm}}_g(t)$, and $\widehat{\mu}^{\mathrm{mcb}}_g(t)$---is pointwise asymptotically linear ($n_g \rightarrow \infty$) for each $t \in (0,1)$ under estimator-specific regularity conditions; these conditions are stated explicitly in Appendix \ref{sec:asymptotic_conditions}. Briefly, $\widehat{\mu}^{\mathrm{emp}}_g(t)$ and $\widehat{\mu}^{\mathrm{sm}}_g(t)$ are pointwise asymptotically linear when the estimated quantile function for each individual distribution converges sufficiently quickly to the corresponding true quantile function. These conditions may require sequencing depths $m_i$ to increase sufficiently quickly with $n_g$, so that the average per-individual quantile estimation error is negligible at the $n_g^{-1/2}$ scale. In contrast, the MCB estimator does not require sequencing depths to increase with $n_g$ and is pointwise asymptotically linear when the binomial mixture models used in the estimation procedure are correctly specified. For notational simplicity, let $\widehat{\mu}_g(t)$ denote any one of the estimators $\widehat{\mu}^{\mathrm{emp}}_g(t)$, $\widehat{\mu}^{\mathrm{sm}}_g(t)$, or $\widehat{\mu}^{\mathrm{mcb}}_g(t)$. Then, under the conditions in Appendix \ref{sec:asymptotic_conditions}, we have pointwise asymptotic normality  for each fixed $t \in (0,1)$:
\begin{equation}
\sqrt{n_g}\left\{\widehat{\mu}_g(t)-\mu_g(t)\right\}
\xrightarrow{d}
N\!\left(0,\sigma_g^2(t)\right),
\qquad
\sigma_g^2(t)
=
\operatorname{Var}\!\left\{\phi_g(O_i;t)\mid G_i=g\right\},
\label{eq:mu_pointwise_clt}
\end{equation} 
where $\xrightarrow{d}$ denotes convergence in distribution.
Accordingly, pointwise asymptotically valid confidence intervals for $\mu_g(t)$ can be constructed as
\begin{equation}
\widehat{\mu}_g(t)
\pm
z_{1-\alpha/2}\,
\frac{\widehat{\sigma}_g(t)}{\sqrt{n_g}},
\label{eq:mu_pointwise_ci}
\end{equation}
where $\widehat{\sigma}_g(t)$ is a consistent estimator of $\sigma_g(t)$.

The pointwise asymptotic normality result above provides valid inference for $\mu_g(t)$ at any fixed quantile level $t$. However, pointwise results alone do not characterize the joint behavior of $\widehat{\mu}_g(t)$ across quantile levels and therefore cannot be used directly to construct simultaneous confidence bands or global hypothesis tests. For these purposes, it is useful to view the function $t \mapsto \mu_g(t)$ itself as the estimand and to establish convergence of the estimated function as a stochastic process. Under additional regularity conditions, also stated in Appendix \ref{sec:asymptotic_conditions}, the estimator $\widehat{\mu}_g$ converges weakly as a stochastic process. Specifically, for any compact set $\mathcal{T} \subset (0,1)$,
\begin{equation}
\sqrt{n_g}\left\{\widehat{\mu}_g-\mu_g\right\}
\rightsquigarrow
\mathbb{G}_g
\qquad \text{in } \ell^\infty(\mathcal{T}),
\label{eq:mu_weak_convergence}
\end{equation}
where $\rightsquigarrow$ denotes weak convergence of stochastic processes and $\mathbb{G}_g$ is a tight mean-zero Gaussian process with covariance kernel $\Sigma_g(s,t)$. This result allows us to perform uniform inference over $t \in \mathcal{T}$. In particular, let $\sigma_g(t) = \{\Sigma_g(t,t)\}^{1/2}$ and let $c_{1-\alpha,g}$ denote the $(1-\alpha)$ quantile of $\sup_{t \in \mathcal{T}} |\mathbb{G}_g(t)/\sigma_g(t)|$. Then an asymptotic $(1-\alpha)$ uniform confidence band for $\mu_g$ based on the studentized process is given by
\begin{equation}
\left\{
\widehat{\mu}_g(t)
\pm
\widehat{c}_{1-\alpha,g}
\frac{\widehat{\sigma}_g(t)}{\sqrt{n_g}}
:
t \in \mathcal{T}
\right\},
\label{eq:mu_uniform_band}
\end{equation}
where $\widehat{\sigma}_g(t)$ and $\widehat{c}_{1-\alpha,g}$ are consistent estimators of $\sigma_g(t)$ and $c_{1-\alpha,g}$, respectively.

\subsection{Task 2: Estimation and inference for the difference in quantiles}

Next, we consider estimation and inference for the quantile difference map $\Delta(t)$ defined in \eqref{eq:statistical_estimand}. Pointwise asymptotic normality of $\widehat{\Delta}(t)$, constructed using any of $\widehat{\mu}^{\mathrm{emp}}_g(t)$, $\widehat{\mu}^{\mathrm{sm}}_g(t)$, or $\widehat{\mu}^{\mathrm{mcb}}_g(t)$, follows from the asymptotic normality of the group-specific estimators together with the continuous mapping theorem and independence of both groups. Let $n=n_1+n_2$, and suppose that $n_g/n \to \pi_g \in (0,1)$ for $g \in \{1,2\}$. Then, for each fixed $t \in (0,1)$,
\begin{equation}
\sqrt{n}\left\{\widehat{\Delta}(t)-\Delta(t)\right\}
\xrightarrow{d}
N\!\left(0,\sigma_\Delta^2(t)\right),
\label{eq:delta_pointwise_clt}
\end{equation}
where, under independence of the two groups, $\sigma_\Delta^2(t)
= \frac{\sigma_1^2(t)}{\pi_1} + \frac{\sigma_2^2(t)}{\pi_2}$.

Weak convergence in $\ell^\infty(\mathcal{T})$ of $\widehat{\Delta}$ follows from weak convergence of the group-specific estimators together with the continuous mapping theorem. Specifically,
\begin{equation}
\sqrt{n}\left\{\widehat{\Delta}-\Delta\right\}
\rightsquigarrow
\mathbb{G}_\Delta
\qquad \text{in } \ell^\infty(\mathcal{T}),
\label{eq:delta_fclt}
\end{equation}
where $\mathbb{G}_\Delta$ is a mean-zero Gaussian process with covariance kernel $\Sigma_\Delta(s,t)$.
Under independence of the two groups,
\begin{equation*}
\Sigma_\Delta(s,t)
=
\frac{\Sigma_1(s,t)}{\pi_1} + \frac{\Sigma_2(s,t)}{\pi_2}
,
\qquad s,t \in \mathcal{T}.
\label{eq:delta_cov_kernel_sum}
\end{equation*}
These results can be used to derive both asymptotically valid pointwise confidence intervals and uniform confidence bands, similar to \eqref{eq:mu_pointwise_ci} and \eqref{eq:mu_uniform_band} in the previous section, respectively.

\subsection{Task 3: Hypothesis testing}

Finally, we describe hypothesis testing for the global null hypothesis $H_0: \Delta(t) = 0 \text{ for all } t \in \mathcal{T}$. By the weak convergence result in \eqref{eq:delta_fclt}, the estimated difference process satisfies $\sqrt{n}\left\{\widehat{\Delta}-\Delta\right\}
\rightsquigarrow
\mathbb{G}_\Delta$ in $\ell^\infty(\mathcal{T})$. Under $H_0$, $\Delta(t)=0$ for all $t \in \mathcal{T}$, so the scaled observed process $\sqrt{n}\,\widehat{\Delta}$ converges weakly to $\mathbb{G}_\Delta$.

A global test statistic may be constructed by applying a real-valued functional to the estimated difference process. That is, for a functional $\phi:\ell^\infty(\mathcal{T})\to\mathbb{R}$, we may consider statistics of the form $T = \phi(\sqrt{n}\widehat{\Delta})$. If $\phi$ is continuous in $ \ell^\infty(\mathcal{T})$, then under $H_0$, $T \rightsquigarrow \phi(\mathbb{G}_\Delta)$ by the continuous mapping theorem. Therefore, critical values can be obtained from the limiting Gaussian process, or in practice from a simulation-based procedure that consistently approximates its distribution \citep[Chapter 10]{kosorok_introduction_2008}. 

We consider three global test statistics based on this weak convergence result. The first two are based on the studentized difference process. Let $\widehat{Z}_\Delta(t)=\sqrt{n}\widehat{\Delta}(t)/\{\widehat{\Sigma}_\Delta(t,t)\}^{1/2}$ when $\widehat{\Sigma}_\Delta$ is an estimate of the covariance kernel of the limiting process $\mathbb{G}_\Delta$. We define the supremum test statistic $T_\infty$ and the integrated squared test statistic $T_{\mathrm{int}}$ as
\begin{equation}
T_\infty
:=
\sup_{t \in \mathcal{T}}
\left|
\widehat{Z}_\Delta(t)
\right|
\qquad
T_{\mathrm{int}}
:=
\int_{\mathcal{T}}
\widehat{Z}_\Delta(t)^2\,dt.
\label{eq:test_stats}
\end{equation}

The third statistic $T_M$ is based on the squared functional Mahalanobis semi-distance of the estimated difference process, analogous to Hotelling's $T^2$ test but for functional data \citep{joseph_two-sample_2015}. Let $\widehat{\mathcal{C}}_\Delta$ denote the covariance operator on $L^2(\mathcal{T})$ induced by the estimated covariance kernel $\widehat{\Sigma}_\Delta$, such that $(\widehat{\mathcal{C}}_\Delta f)(s)
:=
\int_{\mathcal{T}}
\widehat{\Sigma}_\Delta(s,t)f(t)\,dt.$
Then,
\begin{equation}
T_M
:=
n
\left\langle
\widehat{\Delta},
\widehat{\mathcal{C}}_{\Delta}^{-1}\widehat{\Delta}
\right\rangle.
\label{eq:test_mahalanobis}
\end{equation}

In practice, we do not use the inverse covariance operator $\widehat{\mathcal{C}}_{\Delta}^{-1}$ directly, since covariance inversion is typically unstable or ill-posed in infinite-dimensional function spaces. We therefore compute a regularized version of $T_M$ using a truncated spectral inverse based on the first $K$ empirical functional principal components of $\widehat{\mathcal{C}}_{\Delta}$. Specifically, if $\widehat{\lambda}_1,\ldots,\widehat{\lambda}_K$ and $\widehat{\psi}_1,\ldots,\widehat{\psi}_K$ denote the first $K$ empirical eigenvalues and eigenfunctions of $\widehat{\mathcal{C}}_{\Delta}$, then we compute
\begin{equation}
T_M^{(K)}
:=
n
\sum_{k=1}^{K}
\frac{
\left\langle
\widehat{\Delta},
\widehat{\psi}_k
\right\rangle^2
}{
\widehat{\lambda}_k
}.
\label{eq:test_mahalanobis_regularized}
\end{equation}

\begin{table}[ht]
\centering
\footnotesize
\setlength{\tabcolsep}{3pt}
\renewcommand{\arraystretch}{1.4}

\begin{tabular}{>{\centering\arraybackslash}m{0.13\textwidth}
                >{\centering\arraybackslash}m{0.22\textwidth}
                >{\centering\arraybackslash}m{0.22\textwidth}
                >{\centering\arraybackslash}m{0.23\textwidth}
                >{\centering\arraybackslash}m{0.16\textwidth}}
\hline
Statistic & Test Statistic & Test Statistic (on grid $\{t_1, ..., t_m\}$) & Sensitive to... & Figure \\
\hline

Supremum $(T_{\infty})$
&
$
\sup_{t \in \mathcal{T}}
\left|
\frac{\widehat{\Delta}(t)}
{\widehat{\mathrm{se}}\{\widehat{\Delta}(t)\}}
\right|
$
&
$
\max_{1 \leq j \leq k}
\left|
\frac{\widehat{\Delta}(t_j)}
{\widehat{\mathrm{se}}\{\widehat{\Delta}(t_j)\}}
\right|
$
&
Localized differences at one or a small number of quantile levels
&
\includegraphics[width=1.3in,height=2in,keepaspectratio]{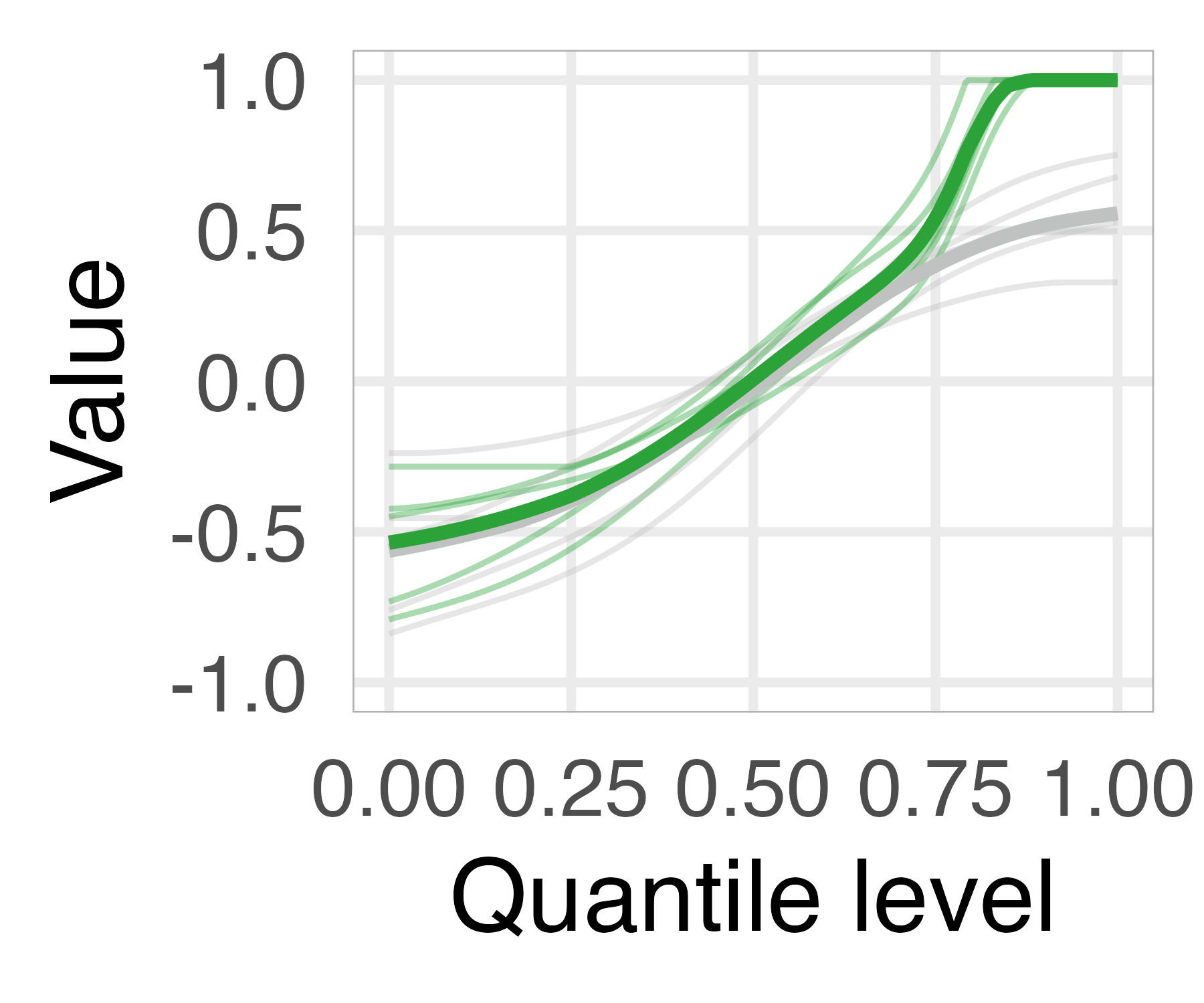}
\\

Integrated squared $(T_{\mathrm{int}})$
&
$
\int_\mathcal{T}
\left[
\frac{\widehat{\Delta}(t)}
{\widehat{\mathrm{se}}\{\widehat{\Delta}(t)\}}
\right]^2
\, dt
$
&
$
\sum_{j=1}^{m} w_j
\left[
\frac{\widehat{\Delta}(t_j)}
{\widehat{\mathrm{se}}\{\widehat{\Delta}(t_j)\}}
\right]^2
$
&
Broad differences spread across the quantile curve
&
\includegraphics[width=1.3in,height=2in,keepaspectratio]{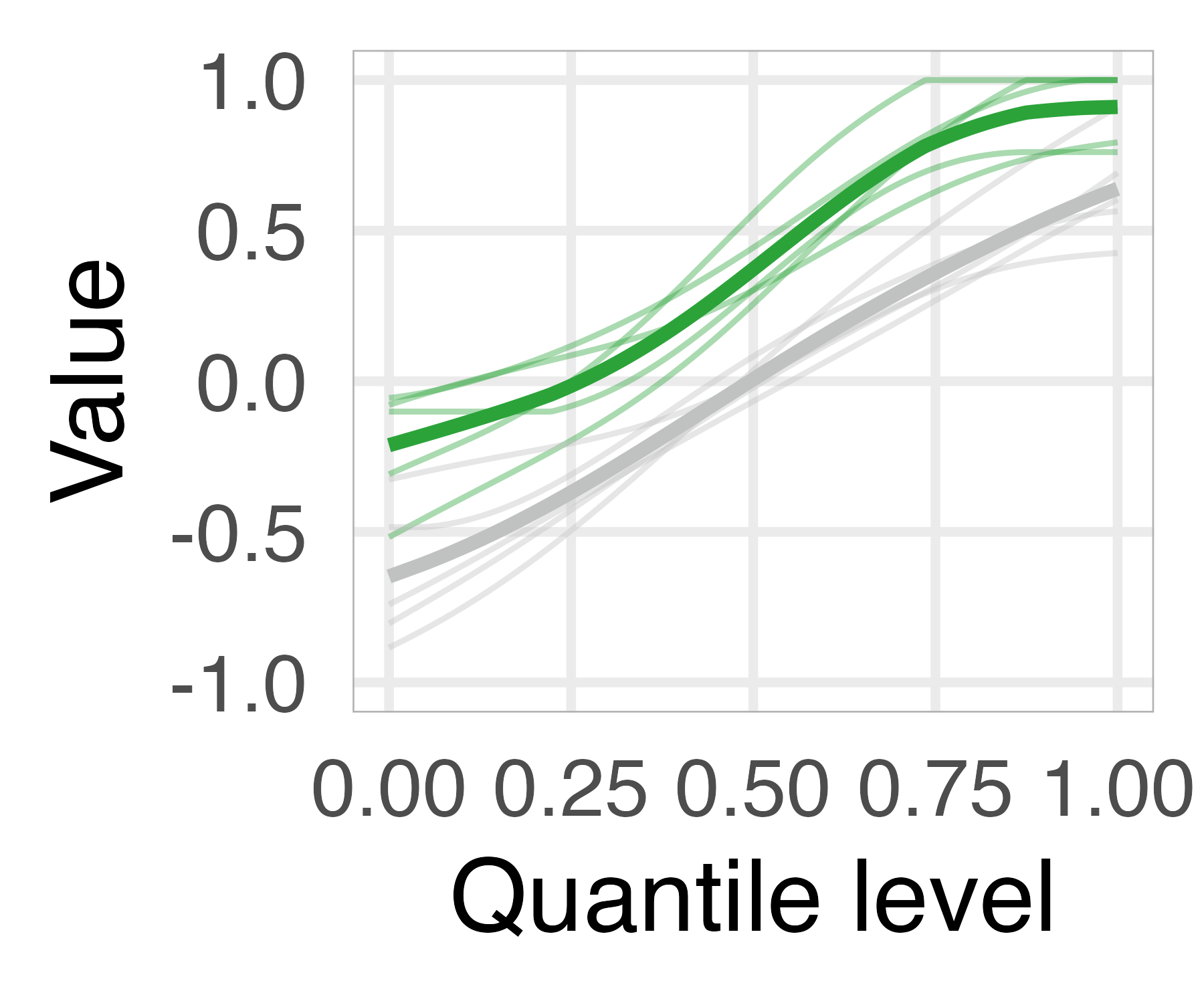}
\\

Hotelling-type $(T_M)$
&
$
n
\left\langle
\widehat{\Delta},
\widehat{\mathcal{C}}_{\Delta}^{-1}\widehat{\Delta}
\right\rangle
$
&
$
n
\widehat{\boldsymbol{\Delta}}^\top
\widehat{\boldsymbol{\Sigma}}_{\Delta,K}^{-1}
\widehat{\boldsymbol{\Delta}}
$
&
Differences after accounting for covariance across quantile levels
&
\includegraphics[width=1.3in,height=2in,keepaspectratio]{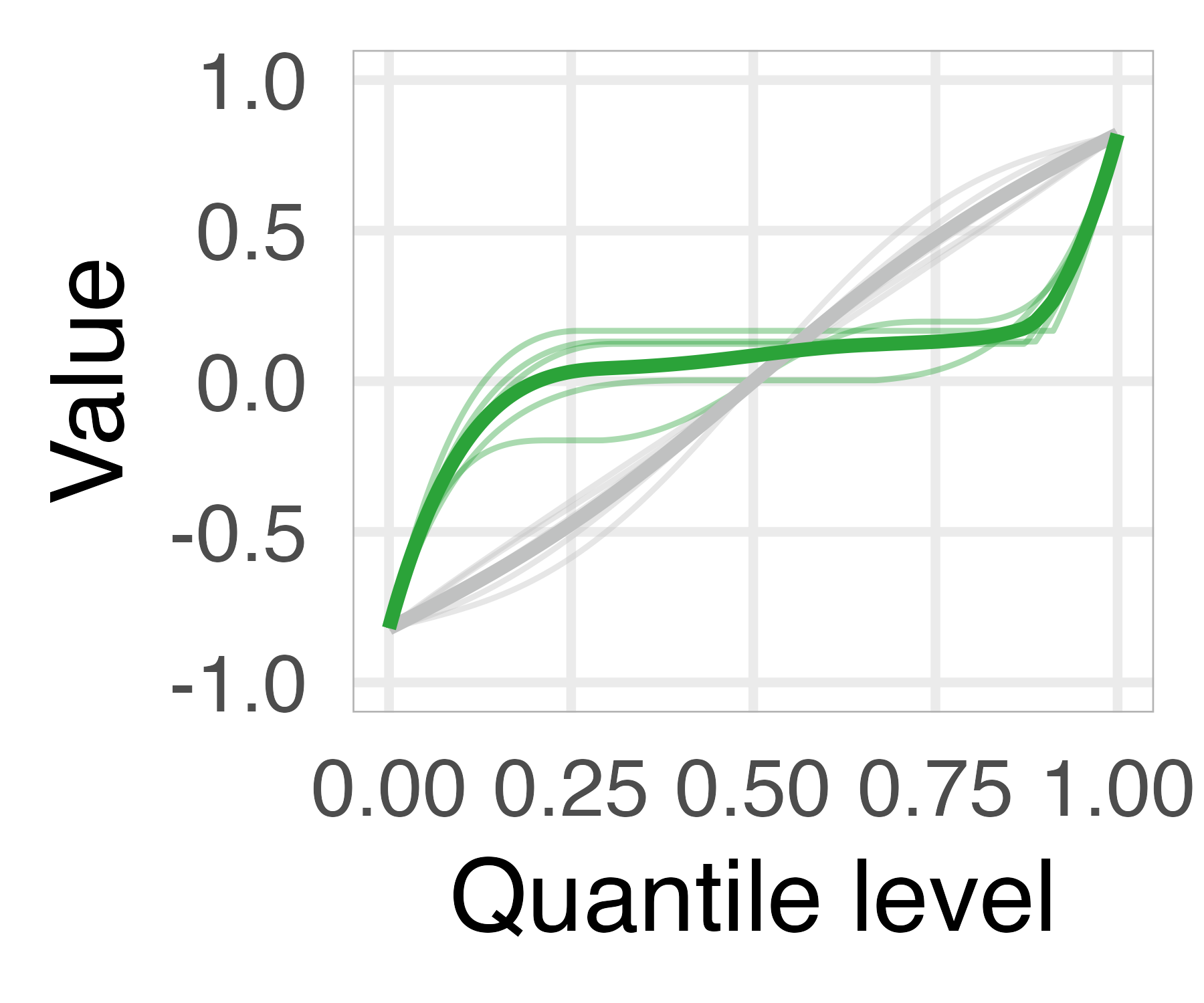}
\\

\hline
\end{tabular}
\caption{
Candidate test statistics for comparing group-specific barycenter quantile functions and schematic alternatives to which each statistic is expected to be sensitive. The right-hand column shows toy examples of where the test statistic should be most powerful, with light lines representing individual-level quantile functions and dark lines representing the corresponding group barycenter quantile functions. Group 1 is shown in gray and Group 2 in green.
}\label{tab:hypothesis_testing}
\end{table}

The three test statistics are sensitive to different alternatives. The supremum statistic is most useful for detecting localized differences concentrated at a small number of quantile levels. The integrated squared statistic aggregates evidence across the quantile curve and is useful for detecting accumulated departures from the null. The Hotelling-type statistic accounts for the covariance structure across quantile levels and may be more powerful when the signal lies along low-variance directions of the estimated process (Table \ref{tab:hypothesis_testing}). We will evaluate the performance of each of these test statistics in our simulation study in Section \ref{sec:simulations}.

\begin{remark}[Other test statistics]
Other test statistics can be constructed from the estimated difference process depending on the alternatives of interest. For example, one-sided supremum statistics may be useful for directional alternatives, while weighted versions of $T_{\mathrm{int}}$ can place greater emphasis on specific regions of the quantile function, such as the tails.
\end{remark}

\subsection{Implementation on a grid}

In practice, the procedure can be implemented using a finite-dimensional representation of the
estimated quantile functions. In our implementation, the functions are evaluated on a common
dense grid in the following steps:
\begin{enumerate}
    \item Choose a grid $t_1,\ldots,t_m \in (0,1)$.

    \item Estimate the Wasserstein barycenter quantile function for each group on the grid using
    $\widehat{\mu}^{\mathrm{emp}}_g(t)$, $\widehat{\mu}^{\mathrm{sm}}_g(t)$, or
    $\widehat{\mu}^{\mathrm{mcb}}_g(t)$. Denote the resulting generic estimator by
    $\widehat{\mu}_g(t)$, and write its grid representation as
    \begin{equation*}
    \widehat{\boldsymbol{\mu}}_g
    =
    \left\{
    \widehat{\mu}_g(t_1),\ldots,\widehat{\mu}_g(t_m)
    \right\}^\top .
    \end{equation*}
    Let $\widehat{\boldsymbol{\Sigma}}_g$ denote the estimated covariance matrix of the limiting process
$\sqrt{n_g}(\widehat{\boldsymbol{\mu}}_g-\boldsymbol{\mu}_g)$. For the empirical and pre-smoothed empirical barycenters, this covariance matrix can be estimated from the sample covariance of the individual-level empirical quantile functions; for the MCB estimator, it can be estimated using a non-parametric bootstrap.

    \item Estimate the difference vector as
    \begin{equation*}
    \widehat{\boldsymbol{\Delta}}
    =
    \widehat{\boldsymbol{\mu}}_2
    -
    \widehat{\boldsymbol{\mu}}_1
    =
    \left\{
    \widehat{\Delta}(t_1),\ldots,\widehat{\Delta}(t_m)
    \right\}^\top .
    \end{equation*}
    Assuming the two groups are independent, estimate the covariance matrix of the limiting process
    $\sqrt{n}(\widehat{\boldsymbol{\Delta}}-\boldsymbol{\Delta})$ by
    \begin{equation*}
    \widehat{\boldsymbol{\Sigma}}_{\Delta}
    =
    \frac{\widehat{\boldsymbol{\Sigma}}_1}{\widehat{\pi}_1} + \frac{\widehat{\boldsymbol{\Sigma}}_2}{\widehat{\pi}_2},
    \end{equation*}
    where $\widehat{\pi}_g=n_g/n$.

    \item Compute the observed test statistics (Table \ref{tab:hypothesis_testing}). First define the pointwise
    standardized differences
    \begin{equation*}
    \widehat{Z}_{\Delta,j}
    =
    \frac{
    \sqrt{n}\widehat{\Delta}(t_j)
    }{
    \sqrt{\widehat{\Sigma}_{\Delta,jj}}
    },
    \qquad j=1,\ldots,m,
    \end{equation*}
    where $\widehat{\Sigma}_{\Delta,jj}$ is the $j$th diagonal element of
    $\widehat{\boldsymbol{\Sigma}}_{\Delta}$.

    The supremum statistic is computed as
    \begin{equation*}
    T_\infty
    =
    \max_{1 \leq j \leq m}
    \left|
    \widehat{Z}_{\Delta,j}
    \right|,
    \end{equation*}
    and the integrated squared statistic is approximated by
    \begin{equation*}
    T_{\mathrm{int}}
    =
    \sum_{j=1}^m
    w_j\widehat{Z}_{\Delta,j}^2,
    \end{equation*}
    where $w_j$ are numerical integration weights, depending on the chosen grid.

    The Hotelling-type statistic is computed using a regularized inverse of
    $\widehat{\boldsymbol{\Sigma}}_\Delta$:
    \begin{equation*}
    T_M
    =
    n
    \widehat{\boldsymbol{\Delta}}^\top
    \widehat{\boldsymbol{\Sigma}}_{\Delta,K}^{-1}
    \widehat{\boldsymbol{\Delta}}.
    \end{equation*}
    Here $\widehat{\boldsymbol{\Sigma}}_{\Delta,K}^{-1}$ is obtained by retaining the first
    $K$ empirical functional principal components of $\widehat{\boldsymbol{\Sigma}}_\Delta$,
    chosen to explain a prespecified proportion of the total variance.

    \item Approximate the null distributions. For $T_\infty$ and $T_{\mathrm{int}}$, the null
    distribution is obtained by simulating mean-zero Gaussian difference processes with covariance
    $\widehat{\boldsymbol{\Sigma}}_\Delta$:
    \begin{equation*}
    \boldsymbol{G}_b
    \sim
    N(0,\widehat{\boldsymbol{\Sigma}}_\Delta),
    \qquad b=1,\ldots,B.
    \end{equation*}
    Each simulated process is then standardized pointwise:
    \begin{equation*}
    Z_{b,j}
    =
    \frac{G_{b,j}}
    {\sqrt{\widehat{\Sigma}_{\Delta,jj}}},
    \qquad j=1,\ldots,m.
    \end{equation*}
    The statistics $T_\infty$ and $T_{\mathrm{int}}$ are recomputed for each
    $\boldsymbol{Z}_b :=(Z_b(t_1),\ldots,Z_b(t_J))^\top$ to estimate the null distribution for each test statistic. The Monte Carlo $p$-value is the proportion of simulated null
    statistics greater than or equal to the observed statistic.

    For the Hotelling-type statistic, the null distribution can be approximated using the simulated process or the asymptotic reference $\chi^2$-distribution \citep{joseph_two-sample_2015}
    \begin{equation*}
    T_M^{(K)} \sim \chi^2_K,
    \end{equation*}
    where $K$ is the number of retained empirical principal components.
    \end{enumerate}

These procedures are implemented in an accompanying \textsf{R} package \texttt{mcbarycenter}, which provides all three barycenter estimators, the corresponding difference-in-quantiles estimators, and the global testing procedures described above. A vignette illustrating the use of these methods is provided in Appendix~\ref{sec:vignette}. 

\begin{remark}[Multiplier bootstrap]
    In step 2 of the procedure above, if estimation of the covariance matrix $\widehat{\boldsymbol{\Sigma}}_g$ for $g \in \{1,2\}$ is infeasible due to the density of the grid or for other reasons, one could use the Gaussian multiplier bootstrap to approximate the null distribution of the test statistics $T_\infty$ and $T_{\mathrm{int}}$ \citep{chernozhukov_gaussian_2013}. 
\end{remark}

\section{Simulation studies}\label{sec:simulations} 

We conduct two simulation studies. The first examines how low sequencing depth, as well as differential sequencing depth between groups, can affect estimation and inference for Wasserstein barycenters. This study is intended to show performance of each estimator---the empirical Wasserstein barycenter, pre-smoothed empirical barycenter, or MCB estimator---under different sequencing depth settings. The second simulation study evaluates the proposed hypothesis testing procedure, focusing on the validity of each test statistic under the null and its power under a range of alternative scenarios.

In both simulation studies, we estimate each barycenter quantile function on a grid of quantile levels $\{0.01, 0.02, \ldots, 0.99\}$. We compare the following estimators:
\begin{enumerate}
    \item Empirical Wasserstein barycenter.
    \item Pre-smoothed empirical Wasserstein barycenter. Each individual distribution is pre-smoothed using kernel density estimation with a Gaussian kernel and a bandwidth selected by Silverman's rule of thumb \citep{silverman_density_1986}.
    \item Correctly specified MCB estimator. A Beta distribution specification is used to fit the binomial mixing distributions, which matches the Dirichlet process specification of the data generating process.
    \item Flexible but misspecified MCB estimator. The binomial mixing distributions are fit using splines \citep{efron_empirical_2016}. The spline basis consists of natural splines with 5 degrees of freedom, augmented with point mass basis functions at 0 and 1, and fit using a regularization parameter of 0.01.
\end{enumerate}
The correctly specified MCB estimator is included to evaluate performance when the structural assumptions required by the MCB estimator are satisfied. The spline-based MCB estimator is included to evaluate performance in a more realistic setting in which the mixing distributions may be misspecified.

\subsection{Simulation Study 1: impact of varying sequencing depth on estimation bias}

\subsubsection{Setup}

In the first simulation study, we examine how sequencing depth affects estimation and pointwise inference for the group-specific barycenter quantile functions and the quantile difference map. We consider sequencing depth settings roughly matching scenarios observed in the HIV-1 prevention studies described in the Introduction (Figure~\ref{fig:study_seq_depth}). The settings considered in the simulation are summarized in Table~\ref{tab:sim_seq_depth}.

\begin{table}[htbp]
\centering
\caption{Sequencing depth ranges for Simulation Study 1}
\label{tab:sim_seq_depth}
\begin{tabular}{lll}
\hline
\textbf{Sequencing depth scenario} & \textbf{Sequencing depth (Group 1)} & \textbf{Sequencing depth (Group 2)} \\
\hline
Low balanced        & 2 to 15  & 2 to 15 \\
Low differential            & 2 to 15  & 2 to 7 \\
High balanced       & 2 to 300 & 2 to 300 \\
High differential           & 2 to 300 & 2 to 150 \\
\hline
\end{tabular}

\vspace{0.5em}
\end{table}

For each sequencing depth setting, we consider two sequence feature scenarios:
\begin{enumerate}
    \item \textit{Null scenario}: The two groups have equal Wasserstein barycenters. In both groups, each individual-level distribution is generated from a Dirichlet process with base measure $\mathrm{Normal}(0,1)$ and concentration parameter 2.

    \item \textit{Alternative scenario}: The two groups have unequal Wasserstein barycenters. In Group 1, each individual-level distribution is generated from a Dirichlet process with base measure $\mathrm{Normal}(0,1)$ and concentration parameter 2. In Group 2, each individual-level distribution is generated from a Dirichlet process with concentration parameter 2 and a bimodal Normal mixture base measure, with 80\% weight on $\mathrm{Normal}(0,1)$ and 20\% weight on $\mathrm{Normal}(6,0.25^2)$.
\end{enumerate}

For each simulation setting, we generate observations for 50 individuals in each group using a two-stage mechanism. (Additional simulation studies with 20 and 100 individuals per group are included in Appendix \ref{sec:additional_sims}.) First, for each individual, we sample an individual-level distribution according to the sequence feature scenario described above. Second, we sample that individual's sequencing depth from a discrete uniform distribution over the corresponding range in Table~\ref{tab:sim_seq_depth}, and then sample the corresponding number of observations from the individual's distribution. Combining the four sequencing depth settings with the two sequence feature scenarios yields eight total simulation settings. The true group-specific Wasserstein barycenter quantile functions, and the resulting quantile difference map, are approximated by computing the empirical barycenter of one large dataset with 100{,}000 individuals per group and 1{,}000 sequences per individual. Performance is assessed in terms of bias, estimated standard error, and pointwise $95\%$ confidence interval coverage. Each simulation setting was repeated over 500 independently generated datasets, and estimator performance was summarized across these replicates. 

\subsubsection{Results}

Summary results for the bias, mean estimated standard error, and confidence interval coverage of the estimators of $\Delta(t)$ are shown in Figure~\ref{fig:sim_null_res} under the null and Figure~\ref{fig:sim_non_null_res} under the alternative. In each figure, results are averaged separately over the central quantiles ($0.06$--$0.94$) and the tail quantiles ($0.01$--$0.05$ and $0.95$--$0.99$).

We first consider the low, balanced sequencing depth setting. Under the null scenario, all estimators perform similarly well, with negligible bias and near-nominal confidence interval coverage despite the low sequencing depths. This is expected because, under the null and with the same sequencing depth distribution in the two groups, the bias in the estimated group-specific barycenters is the same across groups. As a result, when the two estimated barycenter quantile functions are differenced, this bias cancels, producing accurate estimates of $\Delta(t)=0$ across quantile levels.

Under the alternative scenario, however, differences among the estimators become  apparent. Both empirical barycenters exhibit larger bias than the MCB estimators, with the discrepancy most pronounced in the tail quantiles. The pre-smoothed empirical barycenter reduces bias relative to the empirical barycenter, but its coverage remains poor in the tails and it does not perform as well as the MCB estimators. The mean absolute bias for the tail quantiles was 0.454 for the empirical barycenter and 0.345 for the pre-smoothed empirical barycenter, compared with 0.030 and 0.064 for the correctly specified and spline-based MCB estimators, respectively. Confidence interval coverage in the tails was 71\% for the empirical barycenter and 65\% for the pre-smoothed empirical barycenter, compared with 93\% and 91\% for the correctly specified and spline-based MCB estimators, respectively.

Next, we consider the low, differential sequencing depth setting. Under the null scenario, the empirical and pre-smoothed empirical barycenters perform reasonably well for the central quantiles, with low bias and near-nominal confidence interval coverage. However, their performance deteriorates in the tails, where both estimators exhibit larger bias and reduced coverage. This pattern is more pronounced under the alternative scenario: the empirical and pre-smoothed empirical barycenters show increased bias even in the central quantiles, with the largest bias and poorest coverage occurring in the tails. These results indicate that when sequencing depth is low, systematic differences in sequencing depth between groups can distort estimation and inference on $\Delta(t)$, particularly at the tail quantile levels. In contrast, both MCB estimators reduce bias and improve coverage. The correctly specified MCB estimator has the smallest bias and coverage closest to the nominal level. The flexible spline-based MCB estimator also performs well with respect to both bias and coverage, suggesting that, in this simulation setting, the spline-based specification is sufficiently flexible to approximate the relevant mixing distributions used to estimate the barycenter. Under the alternative scenario, the mean absolute bias in the tail quantiles was 0.902 for the empirical barycenter and 0.462 for the pre-smoothed empirical barycenter, compared with 0.038 and 0.152 for the correctly specified and spline-based MCB estimators, respectively; the corresponding confidence interval coverage was 31\%, 71\%, 95\%, and 89\%.

Finally, we consider the high sequencing depth settings. In these settings, the empirical barycenter performs substantially better than in the low sequencing depth settings and is generally similar to the MCB estimators for both central and tail quantiles.  This improvement is expected because, with larger sequencing depths, each individual's empirical quantile function provides a better approximation to the corresponding latent quantile function. Although some individuals still have low sequencing depths in these settings, they do not appear to have enough influence to substantially distort the empirical barycenter, particularly in the central part of the distribution. In the alternative scenario, the MCB estimators show slightly lower tail bias and coverage closer to the nominal level compared to the empirical barycenter. The pre-smoothed empirical barycenter has worse performance in the higher sequencing depth settings, especially in the tails. This shows that the smoothing step may introduce additional bias when the empirical distributions are already well-estimated.

More detailed simulation results are provided in Supplementary Figures~\ref{fig:diff_bias}--\ref{fig:bary_ci}. These figures show bias, mean estimated standard error, and confidence interval coverage for both $\Delta(t)$ and the group-specific barycenter quantile functions, $\mu_g(t)$, across all quantile levels $t \in \{0.01, 0.02, \ldots, 0.99\}$. Additional simulations with 20 and 100 individuals per group showed similar patterns and are presented in Appendix~\ref{sec:additional_sim1}. Notably, the MCB estimators continued to perform well even with 20 individuals per group.

\begin{figure}[!htpb]
    \centering
     \includegraphics[width=\textwidth]{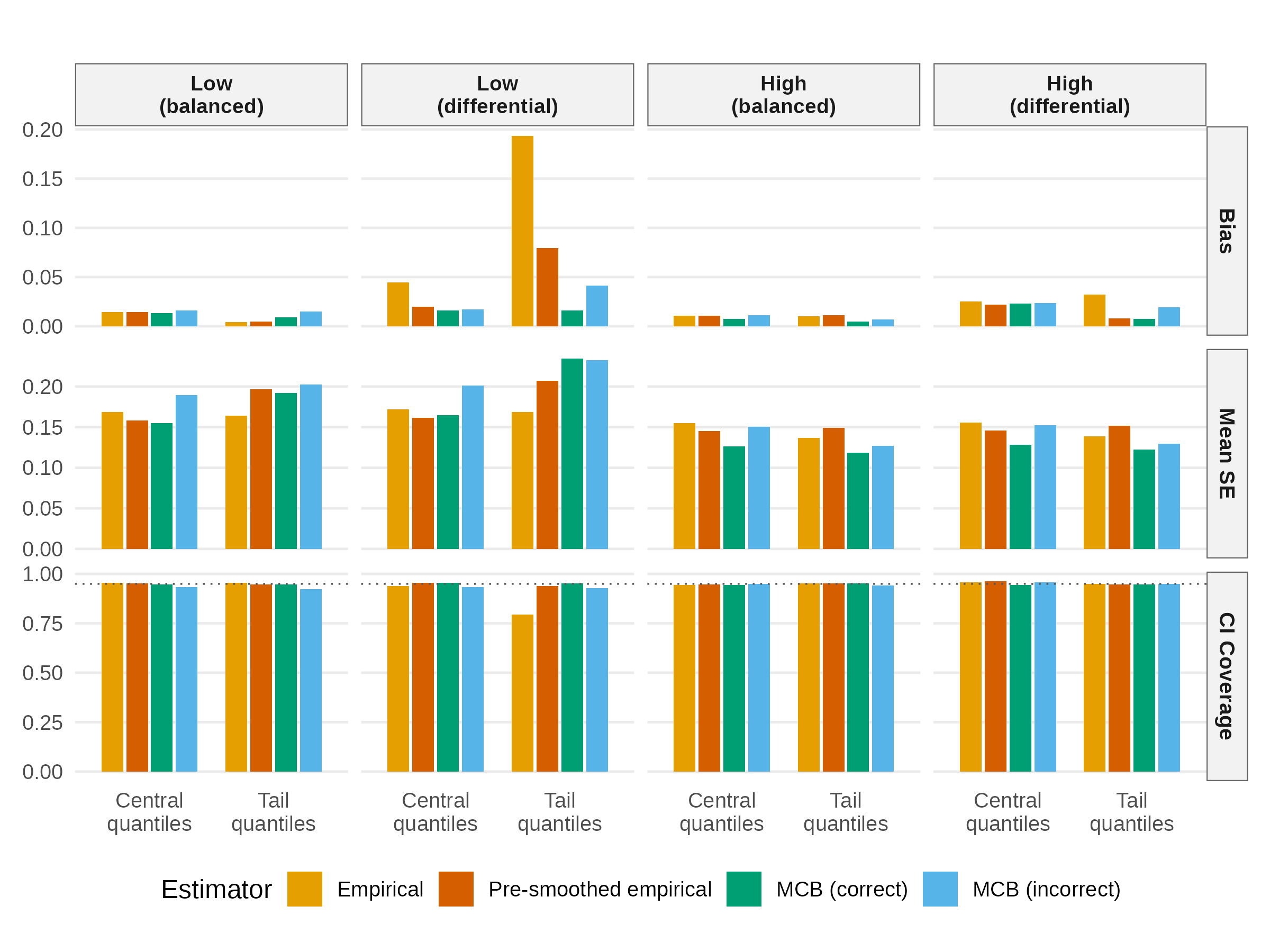}
    \caption{Summary of estimation of the difference in quantiles $\Delta(t)$ under the null scenario for Simulation Study 1. Results are shown across four sequencing depth settings: low balanced, low differential, high balanced, and high differential sequencing depth. For each setting, we compare the empirical barycenter, labeled ``Empirical''; the pre-smoothed empirical barycenter, labeled ``Pre-smoothed empirical''; the correctly specified MCB estimator, labeled ``MCB (correct)''; and the misspecified spline-based MCB estimator, labeled ``MCB (incorrect)''. Performance is summarized by average absolute bias, mean estimated standard error, and pointwise confidence interval coverage for $\Delta(t)$, averaged separately over the central quantiles ($0.06$--$0.94$) and tail quantiles ($0.01$--$0.05$ and $0.95$--$0.99$). The dotted horizontal line in the coverage panels indicates the nominal 95\% coverage level.}
    \label{fig:sim_null_res}
\end{figure}

\begin{figure}[!htpb]
    \centering
     \includegraphics[width=\textwidth]{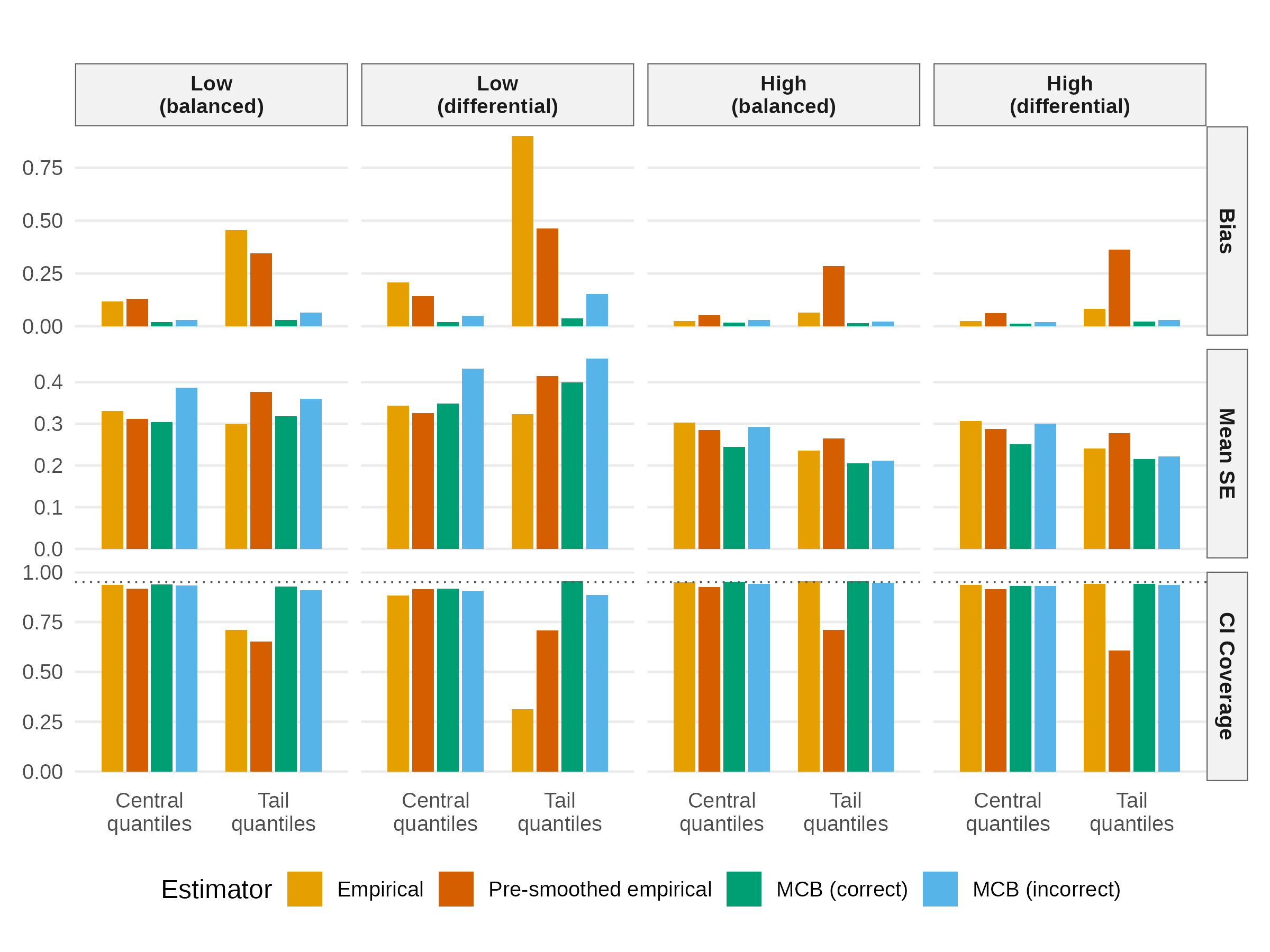}
    \caption{Summary of estimator of the difference in quantiles $\Delta(t)$ under the alternative scenario for Simulation Study 1. Results are shown across four sequencing depth settings: low balanced, low differential, high balanced, and high differential sequencing depth. For each setting, we compare the empirical barycenter, labeled ``Empirical''; the pre-smoothed empirical barycenter, labeled ``Pre-smoothed empirical''; the correctly specified MCB estimator, labeled ``MCB (correct)''; and the misspecified spline-based MCB estimator, labeled ``MCB (incorrect)''. Performance is summarized by average absolute bias, mean estimated standard error, and pointwise confidence interval coverage for $\Delta(t)$, averaged separately over the central quantiles ($0.06$--$0.94$) and tail quantiles ($0.01$--$0.05$ and $0.95$--$0.99$). The dotted horizontal line in the coverage panels indicates the nominal 95\% coverage level.}
    \label{fig:sim_non_null_res}
\end{figure}

\subsection{Simulation Study 2: hypothesis testing}

\subsubsection{Setup}

In the second simulation study, we evaluate the proposed hypothesis testing procedure across a range of two sample scenarios. The goal is to assess Type I error under null scenarios with equal Wasserstein barycenters and power under alternatives with different types of barycenter differences. As in Simulation Study 1, we generate individual-level distributions from Dirichlet processes. We consider four broad classes of scenarios, described in high-level below with full details provided in Supplementary Table~\ref{tab:sim2_scenarios}:

\begin{enumerate}
    \item \textit{Null}: The two groups have identical Wasserstein barycenters. We consider three null settings with different common barycenter shapes: normal, symmetric bimodal, and non-symmetric bimodal.

    \medskip
    For the alternative scenarios (2)--(4), Group 1 is generated from a Dirichlet process with a Normal$(0,1)$ base distribution. The Group 2 data-generating distribution is chosen so that the resulting barycenter differs from the Group 1 barycenter in the following ways:
    \medskip

    \item \textit{Uniform differences}: The Group 2 barycenter differs from the Group 1 barycenter by an approximately constant shift across quantile levels. We examine three shift magnitudes, labeled ``Uniform shift = 0.1, 0.2, 0.3."

    \item \textit{Upper-tail differences}: The Group 2 barycenter differs from the Group 1 barycenter primarily in the upper tail. We examine upper-tail shifts of increasing magnitude, labeled ``Upper-tail shift = 1, 2, 3."

    \item \textit{Shape differences}: The Group 2 barycenter differs from the Group 1 barycenter in overall shape. These alternatives are generated using a bimodal base distribution for Group 2, with increasing component dispersion and consequently greater overlap between the two modes, labeled ``Bimodal = 0.2, 0.4, 0.6."
\end{enumerate}

The group-specific Wasserstein barycenters for all scenarios are shown visually in Supplementary Figure~\ref{fig:sim_test_scenarios}. To focus on the behavior of the test statistics apart from the low sequencing depth bias studied in Simulation Study 1, the main simulation is conducted under the high, balanced sequencing depth setting described in Table~\ref{tab:sim_seq_depth}. (Results for the other sequencing depth settings are provided in Appendix~\ref{sec:additional_sim2}.) For each scenario, we simulate 100 individuals per group and apply each proposed test statistic ($T_{\infty}$, $T_{\mathrm{int}}$, and $T_M$) to test the null hypothesis of equal Wasserstein barycenters at significance level $\alpha = 0.05$. We repeat this procedure over 200 simulated datasets and estimate the rejection probability as the proportion of datasets in which the null hypothesis is rejected. This rejection probability estimates the Type I error under the null scenarios and power under the alternative scenarios. Under the alternative scenarios, we compare the power of the proposed distributional tests with a mean-based analysis, in which the sequence features are averaged within each individual and the groups are compared using an unequal-variance two-sample $t$-test.

\subsubsection{Results}

Results are shown in Figure~\ref{fig:sim_hypothesis_tests}. Under the null settings, all barycenter estimators and all three test statistics had rejection probabilities close to the nominal 5\% level. 

For the uniform difference alternatives, the three test statistics performed similarly, with rejection probabilities increasing as the size of the mean shift increased. In these scenarios, the correctly specified and spline-based MCB estimators tended to have slightly higher power compared to the empirical barycenters. The distributional tests also  generally had similar power to the mean-based analysis. This is expected because the difference between the two group barycenters is approximately constant across quantile levels, so averaging the sequence features within each individual does not remove the main signal. For example, in the mean shift = 0.2 setting, the rejection probabilities for the supremum test were 70.0\%, 66.0\%, 75.5\%, and 78.0\% for the empirical barycenter, pre-smoothed empirical barycenter, correctly specified MCB estimator, and spline-based MCB estimator, respectively, compared with 68.0\% for the mean-based analysis.
 
For the upper-tail alternative scenarios, the supremum test statistic generally had the largest rejection probabilities among the three proposed test statistics, especially as the magnitude of the upper-tail shift increased. The correctly specified and spline-based MCB estimators tended to have the highest power, while the empirical and pre-smoothed empirical barycenters performed similarly. In this setting, the distributional tests had substantially higher power than the mean-based analysis. This is because the group difference is concentrated in the upper tail, and averaging the sequence features within each individual attenuates this localized signal. For example, in the upper-tail shift = 2 setting, the rejection probabilities for the supremum test were 45.0\%, 43.5\%, 61.0\%, and 59.0\% for the empirical barycenter, pre-smoothed empirical barycenter, correctly specified MCB estimator, and spline-based MCB estimator, respectively, compared with 23.0\% for the mean-based analysis.

Finally, for the shape difference alternatives, the Hotelling-type test and supremum test had the largest rejection probabilities, while the integrated squared test tended to have lower power. The Hotelling-type test was particularly powerful for the MCB estimators, consistent with the fact that this statistic accounts for covariance across quantile levels and can be more sensitive to structured departures from the null that occur across the quantile function. As in the previous settings, the correctly specified MCB estimator generally had the largest rejection probabilities. The mean-based analysis had little power in this setting because the two groups have similar average feature values despite having different barycenter shapes. For example, in the bimodal = 0.4 setting, the rejection probabilities for the supremum test were 72.5\%, 83.5\%, 85.0\%, and 72.0\% for the empirical barycenter, pre-smoothed empirical barycenter, correctly specified MCB estimator, and spline-based MCB estimator, respectively, compared with 3.5\% for the mean-based analysis.

\begin{figure}[!htpb]
    \centering
     \includegraphics[page=2, width=\textwidth]{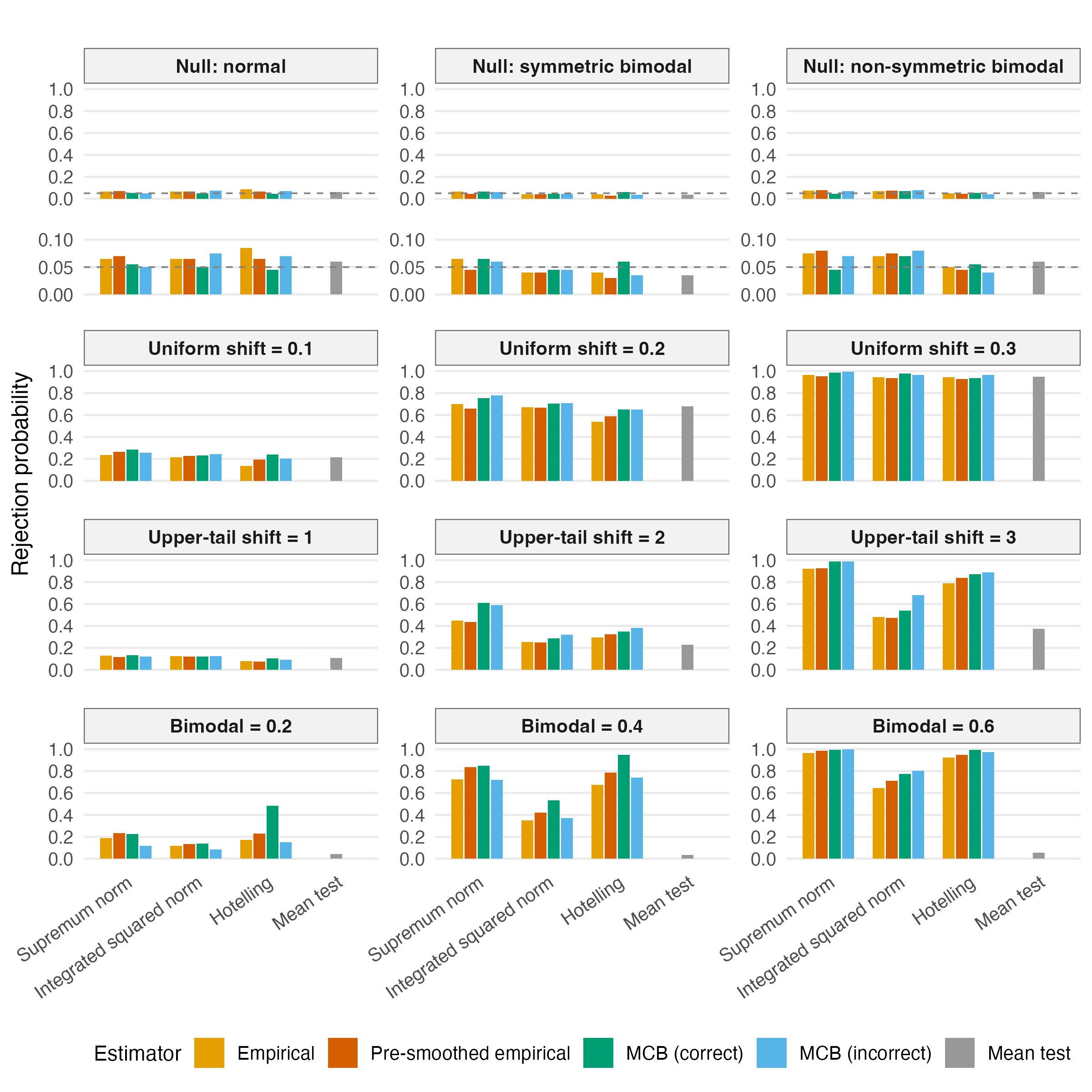}
    \caption{Rejection probabilities across the three null scenarios and nine alternative scenarios for each estimator and test statistic. The four barycenter estimators considered are the empirical barycenter (“Empirical”), the pre-smoothed empirical barycenter (“Pre-smoothed empirical”), the correctly specified MCB estimator (“MCB (correct)”), and the misspecified spline-based MCB estimator (“MCB (incorrect)”). The mean test is based on averaging the sequence features within each individual and applying Welch's two-sample t-test to the resulting individual-level averages. The second row presents enlarged views of the three null scenarios, with the y-axis restricted to 0–0.10. The dashed horizontal line indicates the nominal significance level of 0.05.}
    \label{fig:sim_hypothesis_tests}
\end{figure}

\subsection{Recommendations from simulation studies}

Simulation Study 1 provides guidance for choosing the barycenter estimator in these particular sequencing depth scenarios. In low sequencing depth settings, especially when sequencing depths differed between groups or tail behavior was of interest, the empirical barycenter performed poorly. In these settings, we recommend using an MCB estimator as the primary approach, provided that there are enough individuals to estimate the binomial mixtures. Although a correctly specified MCB estimator is not available in practice, the spline-configured MCB estimator substantially improved performance relative to the empirical barycenter and continued to perform well even with 20 individuals per group (Appendix \ref{sec:additional_sim1}). 

In high sequencing depth settings, the empirical barycenter performed well across both balanced and differential sequencing depth scenarios and may be preferable as the primary estimator because it is less model-dependent. Compared to the empirical barycenter, the spline-based MCB estimator had modestly lower bias particularly in the tails, but this advantage depends on how closely the mixture model specification matches the truth. We therefore recommend using the empirical barycenter as the primary analysis in high sequencing depth studies, with the spline-based MCB estimator potentially used in a sensitivity analysis. The pre-smoothed empirical barycenter showed mixed performance. Pre-smoothing sometimes reduced bias in low sequencing depth settings, but did not perform as well as the MCB estimators; in higher sequencing depth settings, it introduced additional bias. 


Simulation Study 2 provides guidance for choosing among the three global test statistics. The integrated squared statistic $T_{\mathrm{int}}$ performed similarly to the other statistics under uniform mean shifts, but had lower power for upper-tail and shape difference alternatives. The supremum statistic $T_{\infty}$ performed best for localized upper-tail differences, while the Hotelling-type statistic $T_M$ performed well for shape differences. When there is no strong reason to prioritize a specific alternative, we recommend the supremum statistic $T_{\infty}$ as the default as it  consistently performed well across the simulation settings.

Finally, Simulation Study 2 shows when the proposed distributional framework has the most power to detect group differences relative to a mean-based analysis. When group differences were approximately constant across quantile levels, the distributional tests had power similar to the mean-based test. In contrast, when differences were concentrated in the upper tail or reflected differences in barycenter shape, the distributional tests had significantly higher power because they retained information across the full quantile function.

\section{Data applications}\label{sec:application}

We apply our approach to the AMP HIV-1 prevention trials (HVTN 703/704) and a study on vaccine-elicited B cell receptor maturation. The goal of these applications is to illustrate how the proposed distributional analysis compares with a standard mean-based analysis and how the quantile difference map can be used to interpret differences between groups.

\subsection{AMP HIV-1 prevention trials (HVTN 703/704)}

The AMP trials tested whether the broadly neutralizing antibody VRC01 could prevent HIV-1 acquisition \citep{corey_two_2021}. In these trials, HIV-negative participants were randomized to receive a placebo, low-dose VRC01 (10 mg/kg), or high-dose VRC01 (30 mg/kg) every 8 weeks. Overall efficacy was modest, with an estimated efficacy of 29.0\% (95\% CI: -4.0 to 51.6) for the high-dose arm and 7.2\% (95\% CI: -33.3 to 35.3) for the low-dose arm, pooled across both AMP trials. Among participants who acquired HIV-1 during follow-up, viral samples were deep-sequenced using PacBio sequencing technology to characterize within-host diversity in the HIV-1 Env protein \citep{westfall_optimized_2024}. Among participants who acquired HIV-1 during follow-up, sequencing data were available for 67 in the placebo arm and  105 in the pooled VRC01 arms. Previous analyses found that VRC01 had an estimated efficacy of 75.4\% (95\% CI: 45.5\%, 88.9\%) against viruses that were highly sensitive to VRC01 neutralization \citep{corey_two_2021}.

For each viral sequence, we considered three numeric features related to VRC01 susceptibility as previously considered in \citet{juraska_prevention_2024}: (1) Hamming distance to a VRC01-sensitive reference sequence at pre-selected sites, (2) predicted IC80, and (3) epitope distance. Predicted IC80 is the sequence-predicted concentration of VRC01 required to inhibit viral infection by 80\%, with larger values indicating greater predicted resistance to VRC01 neutralization. Epitope distance is a structure-informed measure of how different the VRC01-binding epitope in a given viral sequence is from epitopes in highly VRC01-sensitive reference sequences, with larger values indicating greater divergence from VRC01-sensitive viruses. Our goal was to compare these feature distributions between the placebo arm and the VRC01 arm (pooling both low-dose and high-dose arms). Because this analysis conditioned on HIV-1 infection, the comparison between arms should not be interpreted causally, despite randomization in the original trial \citep{gilbert_sensitivity_2003}. As the primary analysis, we applied our framework using the empirical barycenter on the grid $0.01, 0.02, \dots, 0.99$. We conducted a sensitivity analysis using the spline-configured MCB estimator with 5 degrees of freedom augmented with two point mass basis functions at 0 and 1. We tested for equality of barycenters using the supremum test statistic $T_{\infty}$. For comparison, we also report a mean-based analysis, in which the features are first averaged within each individual and the groups are then compared using an unequal-variance two-sample $t$-test.

Results for the three features are shown in Figure~\ref{fig:data_analysis_amp}, and analogous results using the spline-configured MCB estimator are shown in Supplementary Figure~\ref{fig:data_analysis_mcbary}. In each row of Figure~\ref{fig:data_analysis_amp}, the left-hand panels show the treatment arm-specific mean summaries and estimated barycenter quantile functions, while the right-hand panels show the corresponding estimated mean differences and quantile difference maps. 

The three features show different patterns of separation between the placebo arm and the pooled VRC01 arm. For Hamming distance at the pre-selected sites, the estimated placebo and VRC01 empirical barycenters are nearly indistinguishable. The estimated quantile difference is close to zero across the full range of quantile levels, and there is no evidence of a difference between arms using either the mean-based analysis ($p=0.95$) or the empirical barycenter-based analysis ($p=0.93$). 

For predicted IC80, the VRC01 arm's estimated barycenter quantile function is larger than that of the placebo arm across much of the distribution, with the separation most apparent among higher quantiles. Neither the mean-based analysis ($p = 0.13$) nor the distributional analysis ($p = 0.12$) provide evidence of a difference between arms. Nevertheless, the quantile difference map suggests that the estimated difference between the two arms is not well summarized by a uniform shift: the estimated difference is positive over most quantile levels but becomes larger in the upper tail. The sensitivity analysis using the spline-configured MCB estimator shows a similar but more pronounced upper-tail difference ($p=0.04$).

Finally, for epitope distance, the VRC01 arms have higher values than the placebo arm across nearly all quantile levels. Both the mean-based analysis and the empirical barycenter-based analysis show strong evidence of a difference between arms ($p<0.001$ for both comparisons). The estimated quantile difference map shows a relatively stable difference across most quantile levels, although it decreases at the extreme upper tail, where estimates are typically less stable. Thus, for the epitope distance feature, the distributional analysis supports the same conclusion as the mean-based analysis, while also showing that the result generally reflects uniform difference between the placebo and VRC01 barycenter distributions rather than a difference restricted to a specific portion of the distribution.

\begin{figure}[!htbp]
    \centering
     \includegraphics[page=2, width=\textwidth]{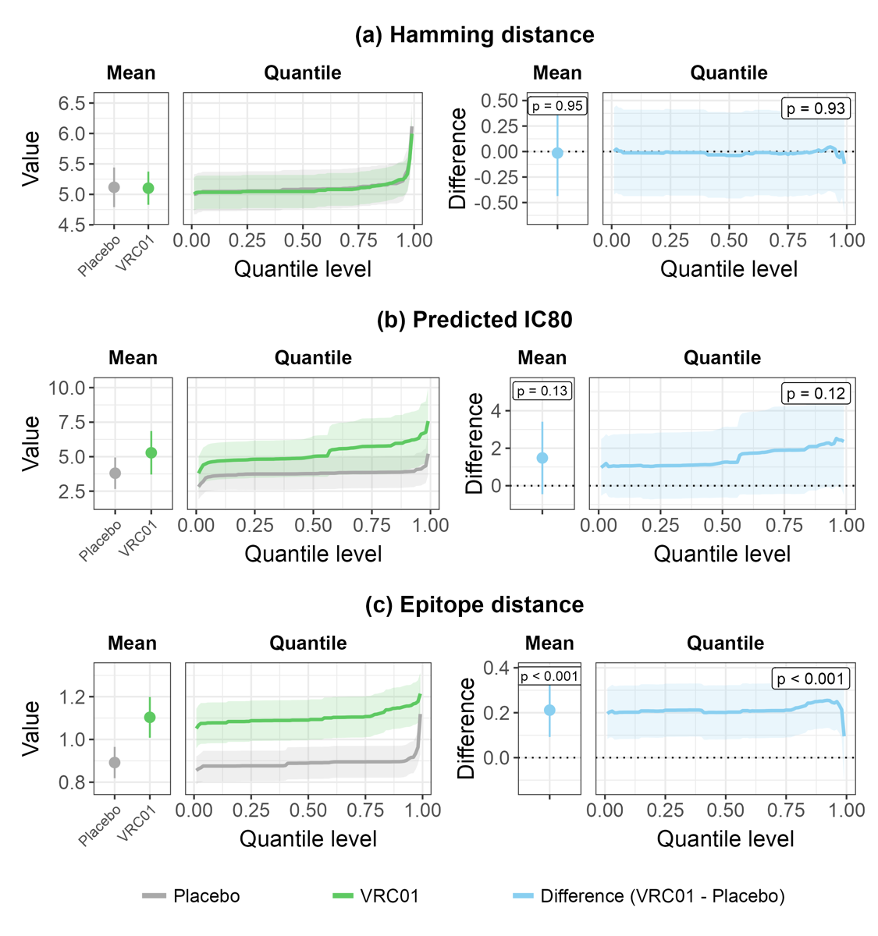}
    \caption{Using the empirical Wasserstein barycenter estimator, analysis of three VRC01 susceptibility-associated features among AMP participants who acquired the primary endpoint of HIV-1 infection: (A) Hamming distance to a VRC01-sensitive reference sequence at pre-selected sites, (B) predicted IC80, and (C) epitope distance. For each feature, the left panels show the arm-specific mean summaries and estimated empirical barycenter quantile functions for the placebo arm and the pooled VRC01 arms. The right panels show the corresponding VRC01-minus-placebo estimated mean difference and quantile difference map. Shaded bands indicate pointwise 95\% confidence intervals. Reported $p$-values correspond to the mean-based unequal-variance two-sample $t$-test for the mean panels and the supremum test for equality of barycenters for the quantile panels.}
    \label{fig:data_analysis_amp}
\end{figure}

\subsection{Vaccine-elicited B cell receptor maturation}

\citet{willis_vaccination_2025} studied a germline-targeting HIV vaccine strategy designed to guide rare B cells toward making VRC01-class broadly neutralizing antibodies. This strategy aimed to induce the accumulation of somatic mutations in these B cells, promoting their maturation toward broadly neutralizing antibodies. The first vaccine, eOD-GT8 60mer (eOD), served as a germline-targeting prime to activate rare VRC01-class precursor B cells, while the second vaccine, core-g28v2 60mer (core), served as a heterologous boost intended to promote further maturation of those cells. 

The study included two phase 1 trials: IAVI G002 in North America, which evaluated both the priming and boosting strategy, and IAVI G003 in Rwanda and South Africa, which evaluated the priming strategy only. In IAVI G002, participants were assigned to different vaccination schedules involving eOD and/or core, while G003 participants received two doses of eOD at weeks 0 and 8, as shown in Figure 1 of \citet{willis_vaccination_2025}. Samples were collected at multiple time points from each participant. Antigen-specific B cells were sorted from PBMC samples and analyzed by B-cell receptor (BCR) sequencing to characterize VRC01-class responses at the antibody-sequence level, allowing the researchers to measure response frequencies, heavy- and light-chain somatic hypermutation, and the accumulation of mutations associated with more mature VRC01-class broadly neutralizing antibodies. \citet{willis_vaccination_2025} compared six vaccination schedule--timepoint combinations, with results shown in Figure 5 of their paper. Labels give the PBMC sampling week followed by the immunogens administered at week 0, at weeks 0 and 8, or at weeks 0, 8, and 16: wk~8 eOD, wk~16
eOD$\rightarrow$eOD, wk~16 eOD$\rightarrow$core, wk~24
eOD$\rightarrow$eOD, wk~24 eOD$\rightarrow$core, and wk~24
eOD$\rightarrow$eOD$\rightarrow$core. The descriptions of each group and their sample sizes are provided in Supplemental Table~\ref{tab:bcr_comparisons}.

We apply our framework to these data by analyzing the heavy-chain percent mutation feature for each sequence. For each group, we estimate the Wasserstein barycenter quantile function on the grid $0.01, \ldots, 0.99$. Because most groups have small participant sample sizes ($<20$) but high sequencing depths, we use the empirical barycenter. We conduct one representative two-sample comparison: wk24 eOD$\rightarrow$core versus wk24 eOD$\rightarrow$eOD$\rightarrow$core.  We test the equality of barycenters using the supremum test statistic $T_{\infty}$. 

The results are shown in Figure~\ref{fig:data_analysis_bcr}. Panel~(a) displays the distributional and mean-based summaries for all six groups. Estimated heavy-chain mutation levels increase over time across nearly all quantiles with the eOD$\rightarrow$core vaccine schedule showing the highest levels at both weeks 16 and 24 compared to the other vaccination regimens. The mean-based summaries show the same ordering across time and regimen. Panel~(b) focuses on the comparison between wk24 eOD$\rightarrow$core and wk24 eOD$\rightarrow$eOD$\rightarrow$core. The estimated barycenter difference, defined as wk24 eOD$\rightarrow$core
minus wk24 eOD$\rightarrow$eOD$\rightarrow$core, is $0.205\%$ at the
0.01 quantile (95\% CI: $-0.696\%$, $1.105\%$), $1.034\%$ at the
median (95\% CI: $0.225\%$, $1.842\%$), and $1.484\%$ at the
0.99 quantile (95\% CI: $0.565\%$, $2.403\%$). Thus, the estimated group difference is not uniform across the mutation distribution; instead, it becomes more pronounced among sequences with higher mutation levels. The supremum test rejects the null hypothesis of equal barycenters ($p=0.001$), providing strong evidence that the two groups differ in their mutation barycenter distributions. The corresponding mean-based analysis estimates an average difference of $0.837\%$ (95\% CI: $0.085\%$, $1.588\%$, $p=0.03$). In the original analysis reported by \citet{willis_vaccination_2025}, which summarized each participant's distribution using the participant-specific median, the difference between these groups was not found to be statistically significant. 

\begin{figure}[!htbp]
    \centering
     \includegraphics[page=2, width=\textwidth]{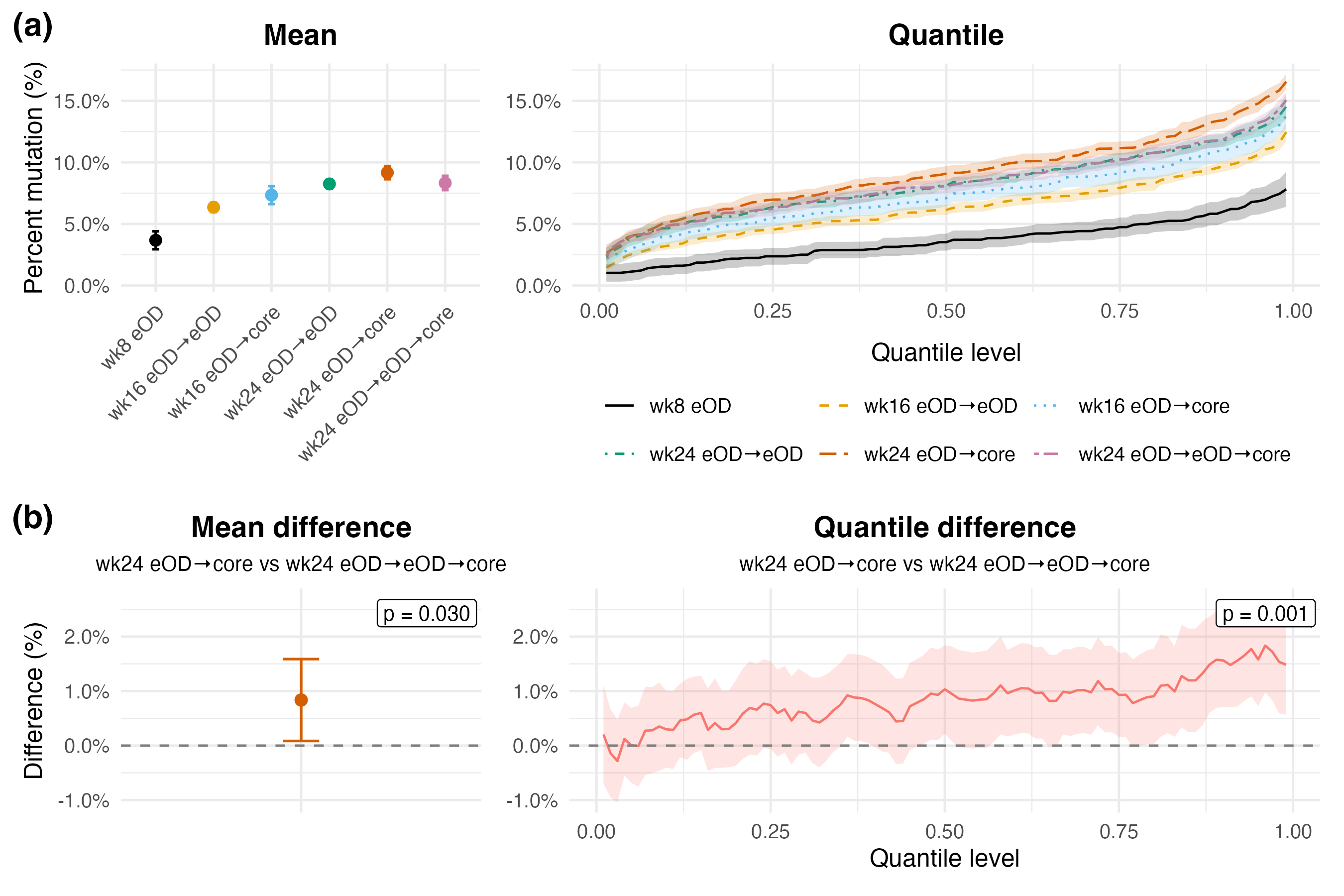}
\caption{
BCR sequencing analysis of vaccine-elicited VRC01-class responses using heavy-chain percent mutation as the sequence-level feature. 
(a) Estimated group summaries for the six vaccination schedule--timepoint groups. The left panel shows the group mean of the participant-specific mean percent mutation values, and the right panel shows the corresponding empirical Wasserstein barycenter quantile functions. 
(b) Two-sample comparison of week~24 eOD$\rightarrow$core versus week~24 eOD$\rightarrow$eOD$\rightarrow$core. The left panel shows the estimated difference in participant-level mean percent mutation, and the right panel shows the estimated quantile difference map, defined as the week~24 eOD$\rightarrow$core barycenter quantile function minus the week~24 eOD$\rightarrow$eOD$\rightarrow$core barycenter quantile function. Vertical bars and shaded bands indicate pointwise 95\% confidence intervals, and the dashed horizontal line marks zero difference. Reported $p$-values correspond to the unequal-variance two-sample $t$-test for the mean comparison and the supremum test for equality of Wasserstein barycenters for the quantile comparison.
}    \label{fig:data_analysis_bcr}
\end{figure}

\section{Discussion}\label{sec:discussion}

We proposed a framework for two-sample comparisons of sequence-derived numeric feature distributions from deep sequencing data. Rather than reducing each individual's sequence set to a scalar summary, the approach treats the sequence-derived feature values for each individual as a distribution-valued outcome. Groups are summarized using Wasserstein barycenters, and group differences are described through quantile difference maps. This allows researchers to compare groups across the full within-host feature distribution and to distinguish uniform shifts from localized tail differences or differences in distributional shape. In contrast, mean-based analyses can identify average differences but may obscure where in the distribution those differences arise.

A key practical consideration is sequencing depth. When sequencing depth is low, an individual's empirical distribution may poorly approximate the corresponding within-host feature distribution, which can bias estimation of the group barycenter, particularly in the tails. In our simulations, the empirical barycenter performed poorly in low sequencing settings, especially when sequencing depths differed between groups. In such settings, the MCB estimator is preferable, provided that there are enough individuals to estimate the binomial mixture components. In higher sequencing depth settings, the empirical barycenter performed well, even when some individuals had low sequencing depths. This is consistent with the theoretical results of \citet{zhou_wasserstein_2024}, which allow some units to have few observations when others have many. Thus, we recommend the empirical barycenter for high sequencing depth applications and the spline-configured MCB estimator for sparse settings. For settings that differ from those studied in this paper, simulations can be used to guide the choice of estimator.

The quantile grid in our implementation should also be chosen with the observed sequencing depths and scientific target of interest. Our simulations considered quantiles from 0.01 to 0.99, but less extreme grids, such as 0.2 to 0.8, may be reliable at lower sequencing depths. Conversely, inference farther into the tails requires greater sequencing depth and is more sensitive to estimator choice. Developing data-adaptive tools for selecting the quantile grid represents an interesting area for future research that could improve the usability of our framework and reduce reliance on manual choices.

Our presentation focused on the two independent sample setting. In many settings, the analysis might seek to adjust for covariates. Existing approaches for causal inference with distributional outcomes provide one path forward for estimating covariate-adjusted treatment effects on barycenter quantile functions \citep{lin_causal_2023}. The empirical barycenter and pre-smoothed empirical barycenter can be incorporated into these approaches directly, whereas extending the MCB estimator to covariate-adjusted settings requires additional work. With discrete covariates, one could estimate covariate-specific barycenters and then marginalize over a target covariate distribution to obtain an adjusted average treatment effect. For continuous covariates, inverse probability weighting, as in Example 7.2 of \citet{peng_estimating_2026}, may provide a useful direction.

Several methodological extensions are also natural. The present work considered one-dimensional sequence features, but some sequence-derived summaries are multivariate and would require alternatives to quantile function representations. Longitudinal extensions would allow investigators to study how within-host feature distributions evolve over time and whether those trajectories differ between groups. Relatedly, paired or repeated-measures versions of the proposed tests would be useful when the same individuals are observed under multiple conditions or across multiple visits.

\section{Acknowledgments}\label{sec:acknowledgments}

The authors thank the participants and investigators of HVTN 703/704 and IAVI G002/G003, and Craig Magaret for constructing the HVTN 703/704 trial datasets. The authors also thank Pamela Shaw, Thomas Richardson, Aaron Hudson, Lillian Cohn, Paul Edlefsen, Raabya Rossenkhan, and James Ludwig for their helpful discussions and insights. This research was supported by the National Institute of Allergy and Infectious Diseases of the National Institutes of Health under grants UM1 AI068635 and R37AI054165. The content is the sole responsibility of the authors and does not necessarily represent the official views of the National Institute of Allergy and Infectious Diseases and the National Institutes of Health.

\section{Data availability }\label{sec:data_availability}

The HVTN 703/704 data will be posted on an HVTN public repository at Harvard Dataverse. The IAVI G002/G003 data are available at \href{https://github.com/SchiefLab/G002-and-G003}{https://github.com/SchiefLab/G002-and-G003}.

\newpage

\printbibliography

\newpage

\appendix

\numberwithin{figure}{section}
\numberwithin{table}{section}

\renewcommand{\thefigure}{\thesection.\arabic{figure}}
\renewcommand{\thetable}{\thesection.\arabic{table}}

\section*{Appendices}
\startcontents[appendices]
\printcontents[appendices]{}{1}{}

\newpage

\section{\texttt{mcbarycenter} R package}\label{sec:vignette}

In this vignette, we describe the \texttt{mcbarycenter} R package, which can be used to implement the three tasks presented in the main body of this paper. The R package is available \href{https://github.com/jpspeng/mcbarycenter}{here}.

\begin{lstlisting}[language=R]
devtools::install_github("jpspeng/mcbarycenter")
\end{lstlisting}

We discuss usage of the package in a general two-sample distributional setting. Specifically, suppose we have $n$ independent units, indexed by $i = 1, \ldots, n$, each with an associated underlying distribution $Y_i$. We do not observe each $Y_i$ directly. Instead, for each unit $i$, we observe $m_i$ independent samples from $Y_i$:
\[
Y_{ij} \sim Y_i, \qquad j = 1, \ldots, m_i.
\]
Additionally, each unit $i$ may have an associated group label $G_i \in \{1,2\}$, and we may want to compare units across groups.

In this vignette, we use the sparse dataset example \texttt{sparse\_dist\_data} data frame available in the package. We also use the \texttt{dplyr} package for data manipulation. 

\begin{lstlisting}[language=R]
library(mcbarycenter) 
library(dplyr)
\end{lstlisting}

\subsection{Preliminaries}

The dataset should include two columns: an ID column, which contains the individual ID, and a value column, which contains samples from the individual's underlying distribution. In many functions, we specify the names of these columns using the arguments \texttt{id\_col} and \texttt{val\_col}.

The sample data also has a group column, which takes the values 1 or 2 and describes the groups that we wish to compare.

\begin{lstlisting}[language=R]
head(sparse_dist_data)
\end{lstlisting}

\begin{verbatim}
      group id value
1     1  1  0.20
2     1  1  0.20
3     1  1  0.20
4     1  2 -0.35
5     1  2 -0.12
6     1  2  0.08
\end{verbatim}

We subset the data into the two groups:

\begin{lstlisting}[language=R]
sample_data1 <- sparse_dist_data %>% filter(group == 1)
sample_data2 <- sparse_dist_data %>% filter(group == 2)
\end{lstlisting}

We will also need to specify a grid to do our analysis. The functions default to the grid 0.01, ..., 0.99, but a more resolved grid could be of interest if the number of samples per distribution is higher.

\subsection{Exploratory analysis}

To begin, we may wish to visualize each individual's empirical quantile function. This allows us to see general patterns in the raw data and identify individuals whose data look different from others.

To visualize the empirical quantile curve for each individual distribution, we can use the \texttt{graph\_individual\_quantiles()} function. This function computes empirical quantiles separately for each ID and then plots the resulting curves. The code below displays the empirical quantile functions for the first three ids.  

\begin{lstlisting}[language=R]
graph_individual_quantiles(sample_data1 %>% filter(id %in% c(1, 2, 3)), 
                           id_col = "id", 
                           val_col = "value")
\end{lstlisting}

When \texttt{interactive = TRUE}, which is the default, the function returns an interactive \texttt{plotly} widget. The curves are visible by default, and clicking a line highlights the selected individual. Double-clicking resets the plot.

\begin{figure}[H]
    \centering
     \includegraphics[width=\textwidth]{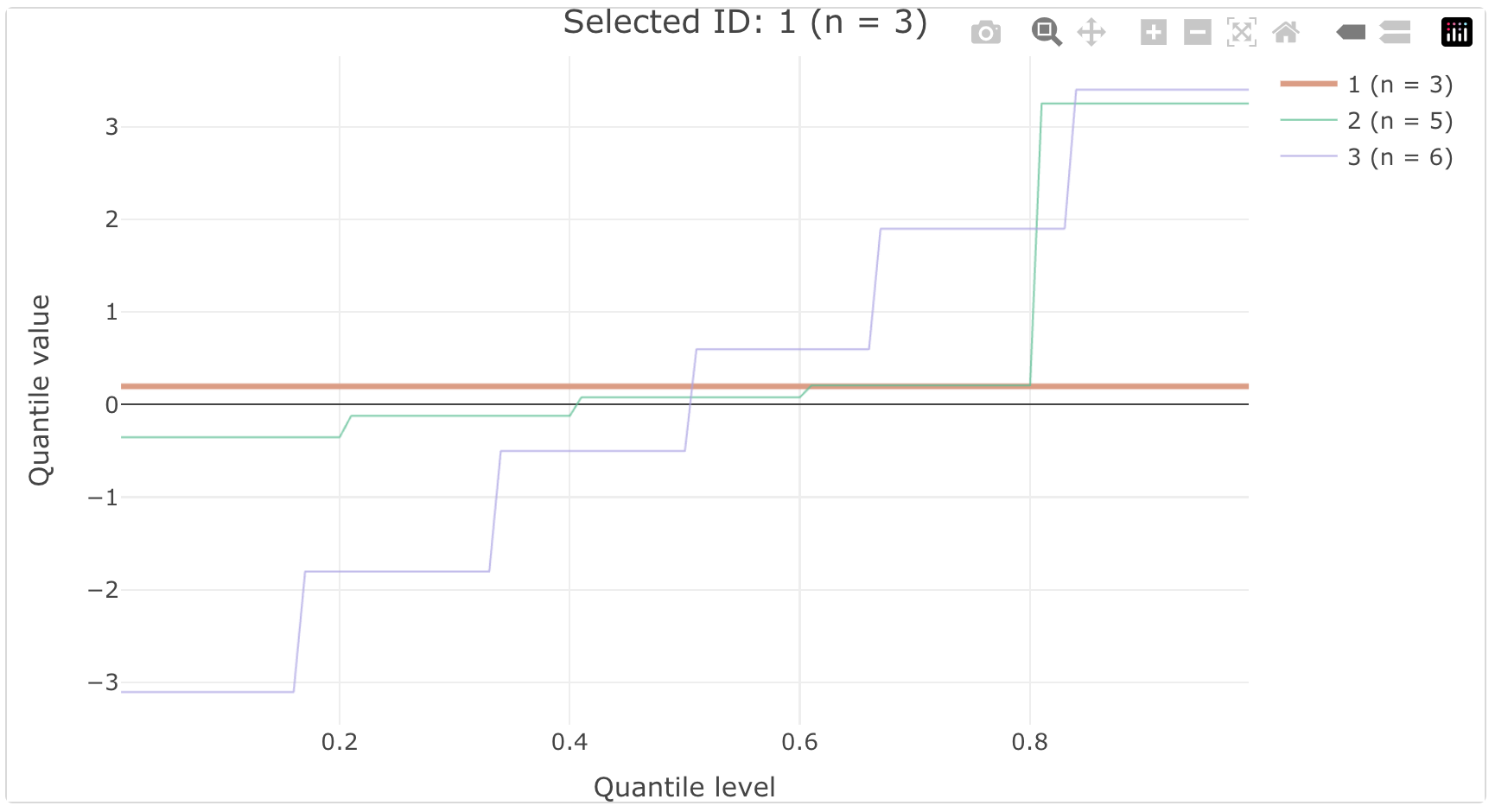}
    \caption{Output of the \texttt{graph\_individual\_quantiles()} function.}
    \label{fig:mcbary_ind_quantiles}
\end{figure}

In the figure above, we can see the individual quantile functions for the first three ids. As shown, individual 1 (sample size $m_1 = 3$) has a homogeneous sample, indicated by a flat empirical quantile function. Individual 2 (sample size $m_2 = 5$) has an extended right tail; that is, their empirical distribution has a small proportion of high values. Individual 3 (sample size $m_3 = 6$) has wide variability in their observed empirical distribution.

\subsection{Task 1: Wasserstein barycenter estimation}

We begin with the first analytic task: estimating the Wasserstein barycenter's quantile function, equivalently the mean quantile function, for each group. We denote the Wasserstein barycenter quantile function for Group 1 as $\mu_1(t)$ and that for Group 2 as $\mu_2(t)$, for $t \in (0,1)$.

\paragraph{Empirical Wasserstein barycenter.}
To estimate the empirical Wasserstein barycenter's quantile function, we can use the \texttt{empbary()} function. This function computes empirical quantiles for each subject and then averages those quantiles pointwise. The main arguments are:
\begin{itemize}
    \item \texttt{df}: the input data frame containing the observed samples.
    \item \texttt{id\_col}: the name of the column identifying the individual or unit.
    \item \texttt{val\_col}: the name of the numeric column containing observations from each individual's underlying distribution.
    \item \texttt{weight\_col}: an optional column containing id-level weights. If \texttt{NULL}, the empirical barycenter is unweighted.
    \item \texttt{quantile\_grid}: the grid of quantile levels at which the barycenter quantile function is estimated. The default is \texttt{seq(0.01, 0.99, 0.01)}.
    \item \texttt{presmooth}: whether to pre-smooth each
individual's empirical distribution before estimating
quantiles. If \texttt{TRUE}, a Gaussian kernel density
estimate is used to estimate each individual's quantile
function. The default is \texttt{FALSE}.
    \item \texttt{quantile\_type}: the interpolation rule
passed to \texttt{stats::quantile()} through its
\texttt{type} argument. The default is \texttt{1}.
This argument is ignored when \texttt{presmooth = TRUE}.
\end{itemize}

\begin{lstlisting}[language=R]
emp_res_grp1 <- empbary(
  df = sample_data1,
  id_col = "id",
  val_col = "value", 
  quantile_grid = seq(0.01, 0.99, 0.01), # optional, defaulted to 0.01, 0.02, ..., 0.99
  quantile_type = 1 # optional, defaulted to quantile_type = 1 
)

emp_res_grp2 <- empbary(
  df = sample_data2,
  id_col = "id",
  val_col = "value"
)
\end{lstlisting}

This function call returns an \texttt{empbary\_result} object containing the estimated quantile curve in \texttt{res}. The \texttt{res} object contains estimates, standard error estimates, and 95\% pointwise confidence intervals for each quantile level. The returned object also includes an estimated covariance matrix in \texttt{cov} and the standardized input data in \texttt{data}. The covariance matrix is estimated as the sample covariance matrix of the individual-level empirical quantile functions.

\begin{lstlisting}[language=R]
head(emp_res_grp1$res)
\end{lstlisting}

\begin{verbatim}
  quantile  estimate         se     ci_lo     ci_hi
1     0.01 -1.233814 0.05900223 -1.349456 -1.118172
2     0.02 -1.233814 0.05900223 -1.349456 -1.118172
3     0.03 -1.233814 0.05900223 -1.349456 -1.118172
4     0.04 -1.233814 0.05900223 -1.349456 -1.118172
5     0.05 -1.233814 0.05900223 -1.349456 -1.118172
6     0.06 -1.233814 0.05900223 -1.349456 -1.118172
\end{verbatim}

\paragraph{Marginal-constructed barycenter.} To estimate the marginal-constructed barycenter estimator, we can use the \texttt{mcbary()} function, which implements the procedure discussed in \citet{peng_estimating_2026}. Briefly, the method transforms the problem of estimating the Wasserstein barycenter's quantile function into the problem of estimating the mixing distributions of a sequence of binomial mixture data.

The main arguments include:
\begin{itemize}
    \item \texttt{df}: the input data frame containing the observed samples.
    \item \texttt{id\_col}: the name of the column identifying the individual or unit.
    \item \texttt{val\_col}: the name of the numeric column containing observations from each individual's underlying distribution.
    \item \texttt{method}: the mixture estimation method. Options include \texttt{"spline"}, \texttt{"npmle"}, and \texttt{"beta"}.
    \item \texttt{x\_grid}: a user-specified threshold grid used to construct the binomial mixture problems.
    \item \texttt{cutpoints}: an optional number of evenly spaced thresholds. Either \texttt{x\_grid} or \texttt{cutpoints} must be supplied.
    \item \texttt{bootstrap\_samples}: the number of bootstrap resamples used for uncertainty estimation.
    \item \texttt{quantile\_grid}: the grid of quantile levels at which the barycenter's quantile function is estimated.
    \item \texttt{use\_midpoint}: whether to use interval midpoints between adjacent thresholds when constructing the estimated distribution of the \(\alpha\)-quantile. Defaulted to $\texttt{FALSE}$.
    \item \texttt{estimate\_endpoints}: whether to estimate endpoint intervals when constructing the estimated distribution of the \(\alpha\)-quantile. Defaulted to $\texttt{FALSE}$.
    \item \texttt{use\_isotonic\_dist}: whether to enforce monotonicity in the cumulative mixture curves across threshold values. Defaulted to $\texttt{FALSE}$.
    \item \texttt{use\_isotonic\_output}: if \texttt{TRUE}, whether to enforce monotonicity in the final quantile function estimate. Defaulted to $\texttt{FALSE}$.
\end{itemize}

For the mixing distribution specification \texttt{method = "spline"}, we can additionally specify the arguments \texttt{pDegree}, which sets the degrees of freedom of the natural spline basis used for estimation, and \texttt{c0}, which sets the penalty constant in the spline deconvolution procedure. We may also specify \texttt{mass\_at\_endpoints}, which, when set to \texttt{TRUE}, includes point masses at $0$ and $1$ in the spline specification. For the nonparametric maximum likelihood specification \texttt{method = "npmle"}, we can additionally specify the argument \texttt{backend}, which can be either \texttt{"mixsqp"} (the default) or \texttt{"REBayes"}; the latter is typically faster, but requires a working installation of Mosek.

In the example below, we use the Beta mixture specification with 100 cutpoints and 20 bootstrap samples. The number of bootstrap samples is kept small for illustration.

\begin{lstlisting}[language=R]
mcb_res_grp1 <- mcbary(
  df = sample_data1,
  id_col = "id",
  val_col = "value",
  method = "beta",
  cutpoints = 100,
  use_midpoint = TRUE,
  estimate_endpoints = TRUE,
  bootstrap_samples = 20
)

mcb_res_grp2 <- mcbary(
  df = sample_data2,
  id_col = "id",
  val_col = "value",
  method = "beta",
  cutpoints = 100,
  use_midpoint = TRUE,
  estimate_endpoints = TRUE,
  bootstrap_samples = 20
)
\end{lstlisting}

The function returns an \texttt{mcbary\_result} object. As with \texttt{empbary()}, the estimated quantile curve and uncertainty summaries are stored in \texttt{res}. The returned object also includes a bootstrap covariance matrix in \texttt{cov}, the estimated mixing distributions in \texttt{mixtures}, the chosen estimation method in \texttt{method}, and the standardized input data in \texttt{data}.

\begin{lstlisting}[language=R]
head(mcb_res_grp1$res)
\end{lstlisting}

\begin{verbatim}
  quantile  estimate estimate_bs         se     ci_lo     ci_hi pct_ci_lo pct_ci_hi
1     0.01 -2.425729   -2.548356 0.23763321 -2.891481 -1.959976 -2.949559 -2.189461
2     0.02 -2.118623   -2.159557 0.19366656 -2.498203 -1.739044 -2.559116 -1.883249
3     0.03 -1.959824   -1.994593 0.15997799 -2.273375 -1.646273 -2.325571 -1.783354
4     0.04 -1.836431   -1.867426 0.14085874 -2.112509 -1.560353 -2.145595 -1.643759
5     0.05 -1.734236   -1.750998 0.11389409 -1.957464 -1.511008 -2.012996 -1.570377
6     0.06 -1.645746   -1.654283 0.09297791 -1.827980 -1.463513 -1.902211 -1.505443
\end{verbatim}

\paragraph{Graphing.}
To visualize these quantile functions, we can use the \texttt{graph\_quantiles()} function. This function takes as input \texttt{empbary\_result} objects, \texttt{mcbary\_result} objects, or any data frame with the columns \texttt{quantile} and \texttt{estimate}. If confidence intervals are available, they can also be displayed using the columns \texttt{ci\_lo} and \texttt{ci\_hi}. We can include multiple objects in the function call, and the function will graph all of the results on the same plot.

The code below graphs the empirical barycenter and MCB estimator for each group.

\begin{lstlisting}[language=R]
graph_quantiles(
  emp_grp1 = emp_res_grp1,
  emp_grp2 = emp_res_grp2,
  mcb_grp1 = mcb_res_grp1,
  mcb_grp2 = mcb_res_grp2, 
  show_ci = TRUE # optional, defaults as TRUE
)
\end{lstlisting}

\begin{figure}[H]
    \centering
    \includegraphics[width=\textwidth]{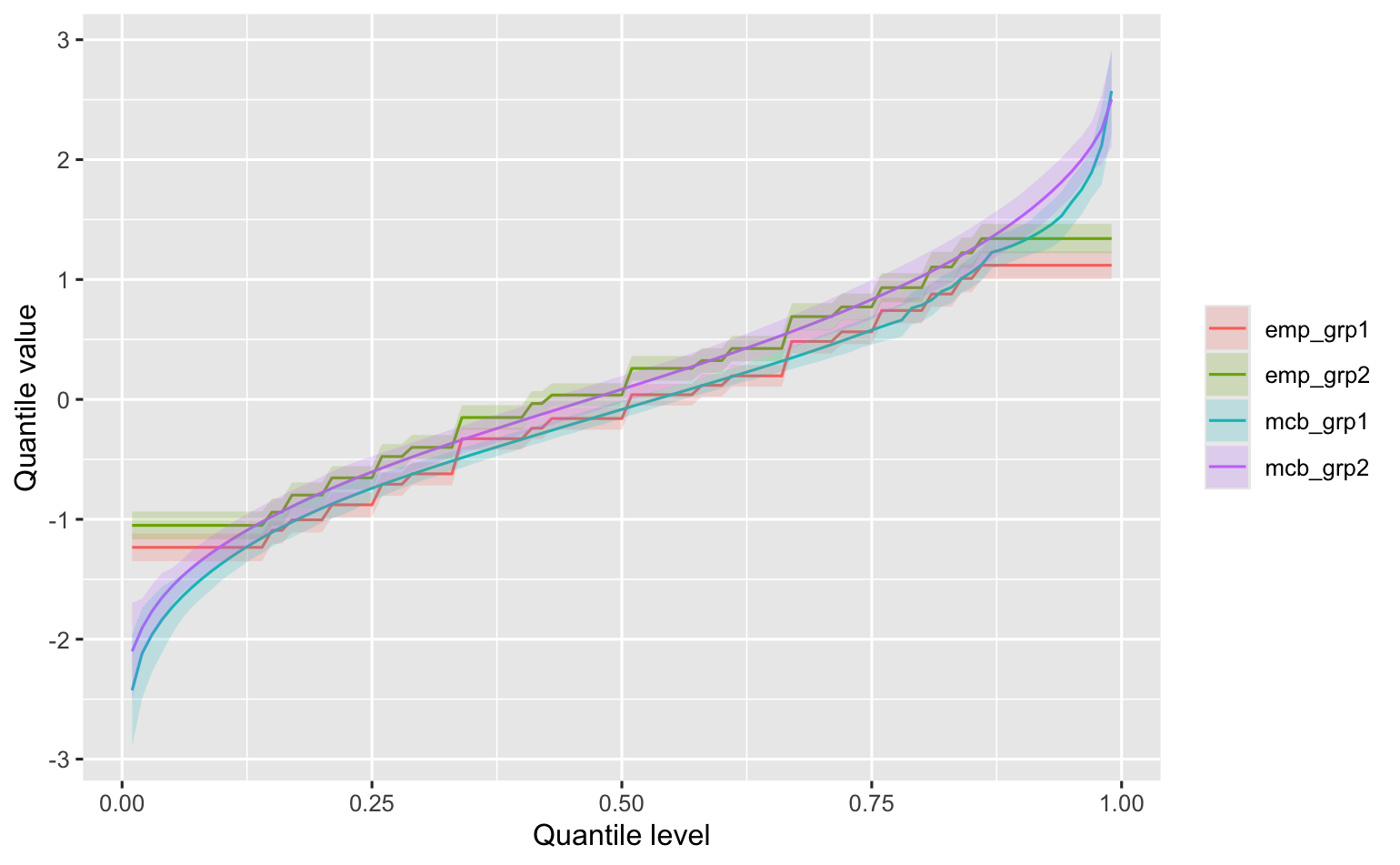}
    \caption{Empirical and marginal-constructed barycenter estimates for Groups 1 and 2.}
    \label{fig:mcbary_barycenter_quantiles}
\end{figure}

We can also visualize the intermediate estimated mixing distributions using \texttt{graph\_mixtures()}. This can be useful for checking the intermediate objects used to construct the MCB estimator.

\begin{lstlisting}[language=R]
graph_mixtures(mcb_res_grp1)
\end{lstlisting}

\begin{figure}[H]
    \centering
    \includegraphics[width=\textwidth]{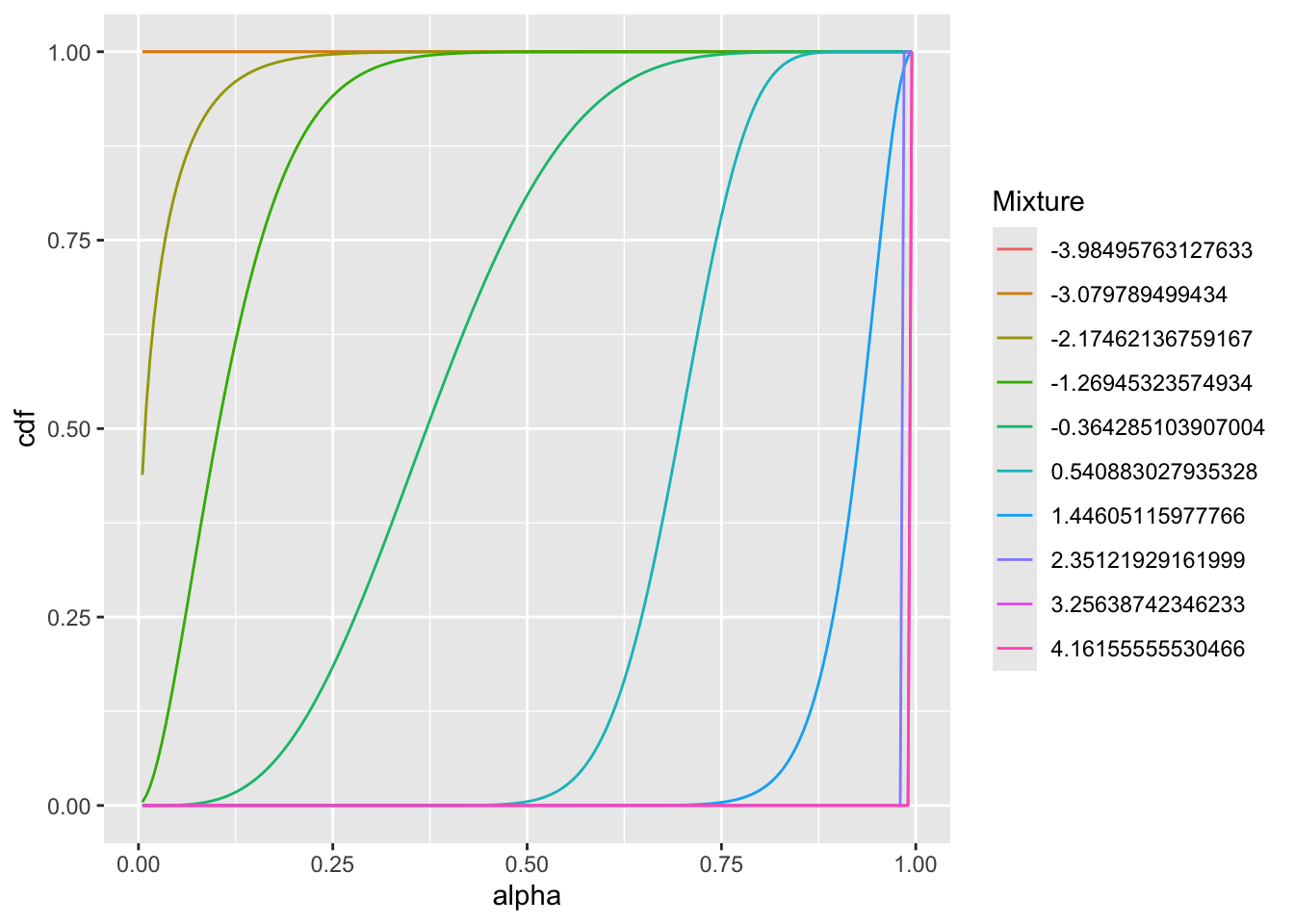}
    \caption{Estimated mixing distributions from the \texttt{mcbary()} procedure for Group 1.}
    \label{fig:mcbary_mixtures_grp1}
\end{figure}

\subsection{Task 2: Difference in quantiles}

The second analytic task is to estimate the difference between the barycenter quantile functions for the two groups. We define the difference in quantiles as
\[
\Delta(t) = \mu_2(t) - \mu_1(t), \qquad t \in (0,1).
\]
Our goal is to estimate $\Delta(t)$ over a grid of quantile levels.

To compute this difference, we can use the \texttt{quantile\_diff()} function. This function takes two estimated barycenter objects and returns the pointwise difference between their estimated quantile functions, along with their estimated standard errors and 95\% confidence intervals. In the example below, we compute the difference between the two MCB estimates.

\begin{lstlisting}[language=R]
quantile_diff_res <- quantile_diff(mcb_res_grp2, mcb_res_grp1)

head(quantile_diff_res)
\end{lstlisting}

\begin{verbatim}
  quantile  estimate        se       ci_lo     ci_hi
1     0.01 0.3247018 0.3151857 -0.29305083 0.9424544
2     0.02 0.2118083 0.2306834 -0.24032298 0.6639395
3     0.03 0.1946393 0.1952820 -0.18810627 0.5773850
4     0.04 0.1840011 0.1755934 -0.16015565 0.5281578
5     0.05 0.1762177 0.1376481 -0.09356763 0.4460030
6     0.06 0.1695939 0.1168332 -0.05939498 0.3985828
\end{verbatim}

We can graph the estimated difference curve using \texttt{graph\_quantiles()}.

\begin{lstlisting}[language=R]
graph_quantiles(mcb_difference = quantile_diff_res)
\end{lstlisting}

\begin{figure}[H]
    \centering
    \includegraphics[width=\textwidth]{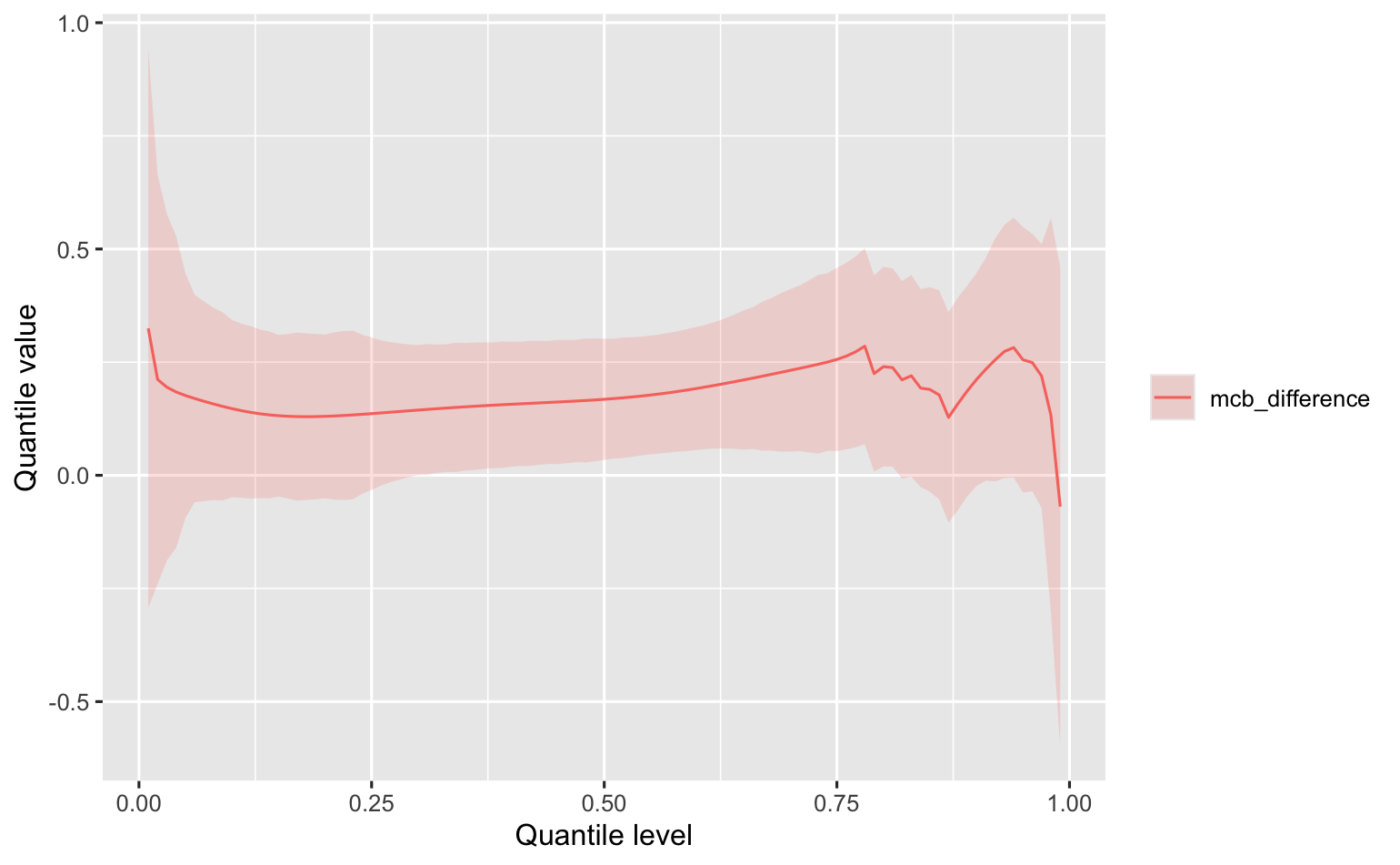}
    \caption{Estimated difference between the MCB barycenter quantile functions for Group 2 versus Group 1.}
    \label{fig:mcbary_quantile_difference}
\end{figure}

This plot shows how the estimated group difference varies across the distribution. Positive values indicate quantile levels at which the estimated Group 2 barycenter is larger than the estimated Group 1 barycenter, while negative values indicate quantile levels at which the estimated Group 1 barycenter is larger.

\subsection{Task 3: Hypothesis testing}

The third analytic task is to test whether the barycenter quantile functions differ between groups. In this setting, the null hypothesis is that the two groups have the same Wasserstein barycenter over the unit interval:
\[
H_0: \mu_1(t) = \mu_2(t) \quad \text{for all } t \in (0,1).
\]
Equivalently, in terms of the difference in quantiles,
\[
H_0: \Delta(t) = 0 \quad \text{for all } t \in (0,1).
\]

To carry out this test, we can use the \texttt{equalbary\_test()} function. The function takes two estimated barycenter objects as input. These objects should be estimated on the same quantile grid; the hypothesis test is performed on this grid. The argument \texttt{B} specifies the number of resamples used to approximate the null distribution, and \texttt{seed} can be specified for reproducibility. For the Hotelling's $T^2$-type test, we can specify the proportion of total variance to retain through the \texttt{variance\_explained} argument, and the test uses the first $K$ empirical functional principal components needed to explain at least this proportion of the variance; the default is 0.95.

\begin{lstlisting}[language=R]
comparison <- equalbary_test(
  mcb_res_grp1,
  mcb_res_grp2,
  B = 1000,
  seed = 1
)
\end{lstlisting}

The output includes results for three test statistics: the supremum norm (in sublist \texttt{\$sup\_norm}), the integrated squared norm (in sublist \texttt{\$int\_sq\_norm}), and the Hotelling's $T^2$-type test statistic (in sublist \texttt{\$hotelling}). Each of these test statistics are described in Table \ref{tab:hypothesis_testing}. Each sublist contains the two-sided p-value for the test (\texttt{pval}), the observed test statistic (\texttt{T\_obs}), and the simulated null (\texttt{T\_sim}) test statistics. In the code below, we extract the p-values from using each of the three test statistics.
\begin{lstlisting}[language=R]
comparison$sup_norm$pval
comparison$int_sq_norm$pval
comparison$hotelling$pval
\end{lstlisting}

\begin{verbatim}
[1] 0.07
[1] 0.005
[1] 0.1442255
\end{verbatim}

\newpage

\section{Supplemental tables and figures}

\begin{table}[H]
\centering
\singlespacing
\caption{Example HIV-1 prevention trials with viral sequencing data, with number of study participants, diagnosed HIV-1 infections, and HIV+ individuals with sequencing data.}
\label{tab:hiv_studies}

\small
\renewcommand{\arraystretch}{1.25}
\begin{tabular}{
  >{\raggedright\arraybackslash}m{2.2cm}
  >{\raggedright\arraybackslash}m{2.0cm}
  >{\raggedright\arraybackslash}m{3.0cm}
  >{\raggedright\arraybackslash}m{3.3cm}
  >{\centering\arraybackslash}m{0.9cm}
  >{\centering\arraybackslash}m{1.4cm}
  >{\centering\arraybackslash}m{2.0cm}
}
\toprule
\textbf{Study (Years)} &
\textbf{Region} &
\textbf{Intervention} &
\textbf{Study Population} &
\textbf{N} &
\makecell{\textbf{HIV+}} &
\makecell{\textbf{HIV+ with}\\\textbf{sequence data}\\\textbf{available}} \\
\midrule
Phambili HVTN 503 (2007--2012) &
South Africa &
MRKAd5 HIV-1 gag/pol/nef subtype B vaccine &
Sexually active HIV-1 seronegative men and women &
801 &
62 &
43 \\
AMP HVTN 703/704 (2016--2022) &
Sub-Saharan Africa (HVTN 703); Americas, Switzerland (HVTN 704) &
VRC01 broadly neutralizing monoclonal antibody &
Women age 18--35 (HVTN 703); MSM and transgender persons age 18--50 (HVTN 704) &
4,623 &
174 &
172 \\
Imbokodo HVTN 705 (2017--2022) &
Sub-Saharan Africa &
Ad26.Mos4.HIV plus clade C gp140 mosaic vaccine regimen &
Women age 18--35 &
2,637 &
119 &
116 \\
\bottomrule
\end{tabular}
\renewcommand{\arraystretch}{1}
\end{table}

\begin{table}[H]
\centering
\singlespacing
\small
\caption{Base measure configurations for the Dirichlet processes used in Simulation Study 2. All Dirichlet processes were configured with a concentration parameter of 2.}
\label{tab:sim2_scenarios}
\begin{tabularx}{\textwidth}{@{}lXX@{}}
\toprule
\textbf{Scenario name} & \textbf{Group 1 configuration} & \textbf{Group 2 configuration} \\
\midrule

Null: normal
& $N(0,1)$
& $N(0,1)$ \\

Null: symmetric bimodal
& $50\%\,N(-1.5,0.5) + 50\%\,N(1.5,0.5)$
& $50\%\,N(-1.5,0.5) + 50\%\,N(1.5,0.5)$ \\

Null: non-symmetric bimodal
& $70\%\,N(-1,0.5) + 30\%\,N(2,0.5)$
& $70\%\,N(-1,0.5) + 30\%\,N(2,0.5)$ \\

\midrule

Mean shift = 0.1
& $N(0,1)$
& $N(0.1,1)$ \\

Mean shift = 0.2
& $N(0,1)$
& $N(0.2,1)$ \\

Mean shift = 0.3
& $N(0,1)$
& $N(0.3,1)$ \\

\midrule

Upper-tail shift = 1
& $N(0,1)$
& $95\%\,N(0,1) + 5\%\,N(1,0.25)$ \\

Upper-tail shift = 2
& $N(0,1)$
& $95\%\,N(0,1) + 5\%\,N(2,0.25)$ \\

Upper-tail shift = 3
& $N(0,1)$
& $95\%\,N(0,1) + 5\%\,N(3,0.25)$ \\

\midrule

Bimodal: 0.2
& $N(0,1)$
& $50\%\,N(-1,0.2) + 50\%\,N(1,0.2)$ \\

Bimodal: 0.4
& $N(0,1)$
& $50\%\,N(-1,0.4) + 50\%\,N(1,0.4)$ \\

Bimodal: 0.6
& $N(0,1)$
& $50\%\,N(-1,0.6) + 50\%\,N(1,0.6)$ \\

\bottomrule
\end{tabularx}
\end{table}

\newpage

\begin{figure}[H]
    \centering
     \includegraphics[page=2, width=\textwidth]{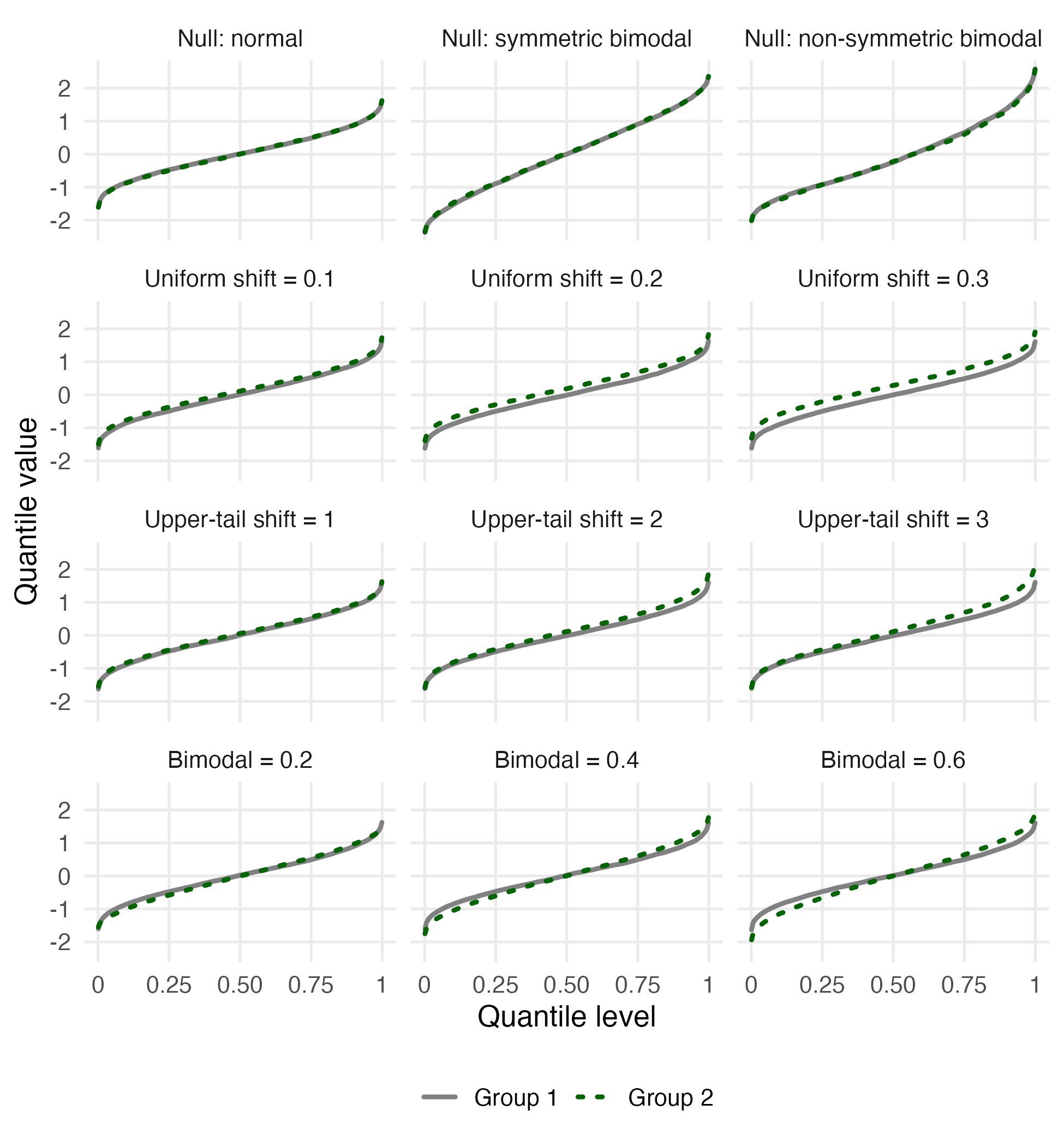}
    \caption{Wasserstein barycenter quantile functions for Group 1 (solid gray) and Group 2 (dotted green) for the conditions studied in Simulation Study 2.}
    \label{fig:sim_test_scenarios}
\end{figure}

\begin{figure}[H]
    \centering
     \includegraphics[width=\textwidth]{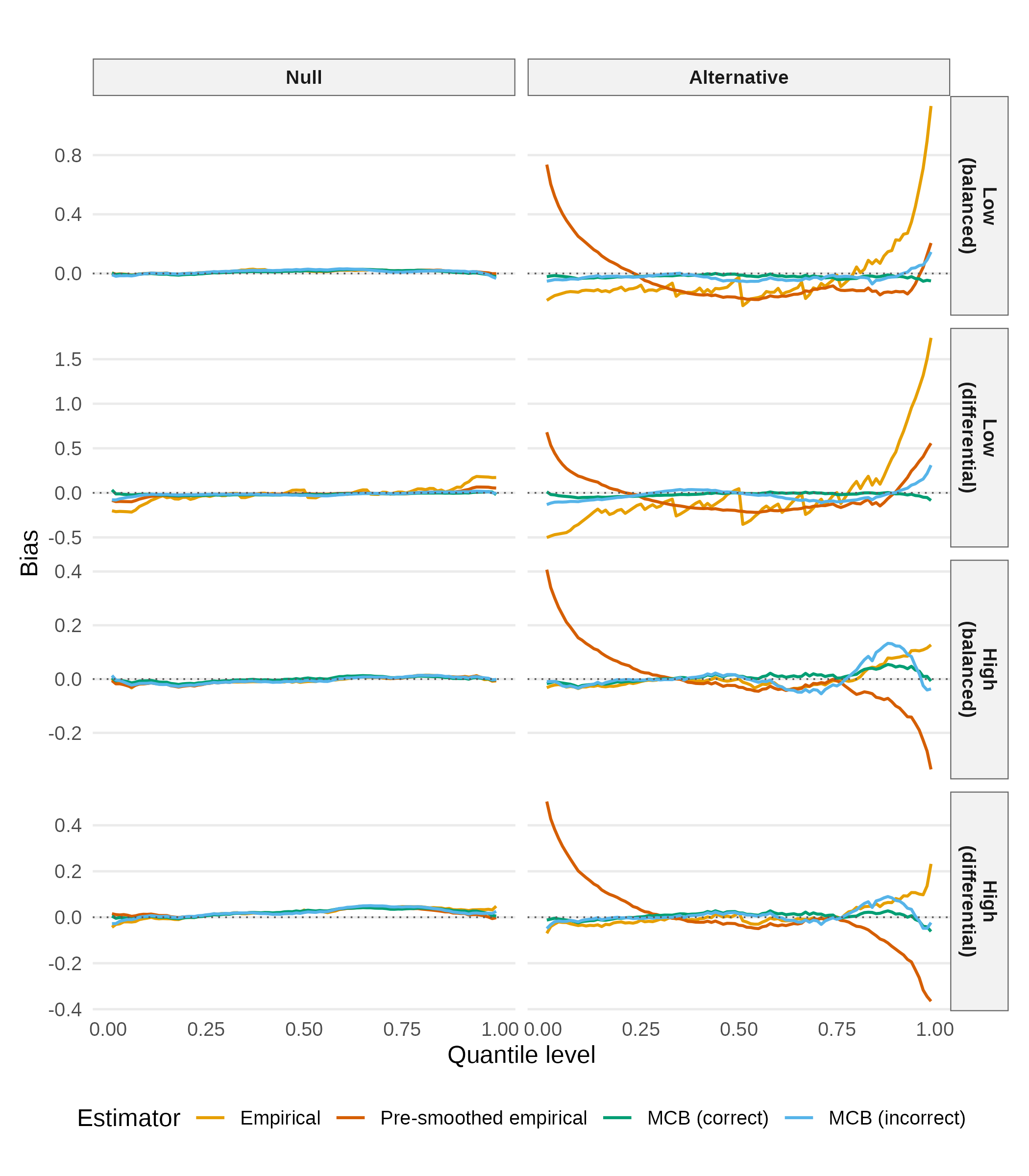}
\caption{Bias of the estimated quantile difference map $\widehat{\Delta}(t)$ in Simulation Study 1. Results are shown under the null and alternative sequence feature scenarios, across sequencing depth settings: low balanced, low differential, high balanced, and high differential. Each curve shows the bias of $\widehat{\Delta}(t)$ across quantile levels $t \in {0.01,0.02,\ldots,0.99}$. Estimators include the empirical Wasserstein barycenter, the pre-smoothed empirical Wasserstein barycenter, the correctly specified MCB estimator, and the misspecified spline-based MCB estimator.}
\label{fig:diff_bias}
\end{figure}

\begin{figure}[H]
    \centering
     \includegraphics[width=\textwidth]{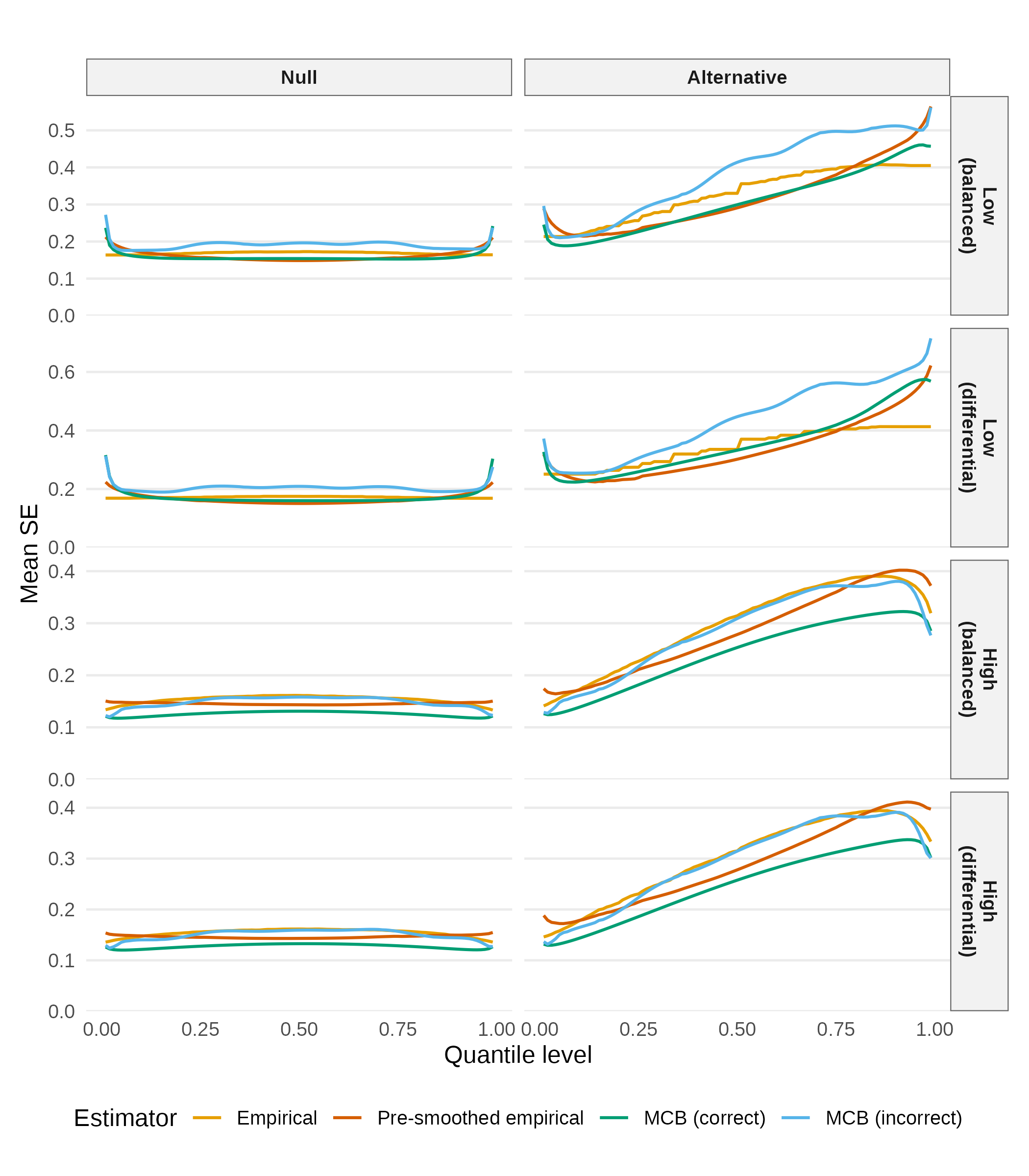}
\caption{Mean estimated standard error of the estimated quantile difference map $\widehat{\Delta}(t)$ in Simulation Study 1. Results are shown under the null and alternative sequence feature scenarios, across sequencing depth settings: low balanced, low differential, high balanced, and high differential. Each curve shows the mean estimated standard error of $\widehat{\Delta}(t)$ across quantile levels $t \in {0.01,0.02,\ldots,0.99}$. Estimators include the empirical Wasserstein barycenter, the pre-smoothed empirical Wasserstein barycenter, the correctly specified MCB estimator, and the misspecified spline-based MCB estimator.}
    \label{fig:delta_se}
\end{figure}

\begin{figure}[H]
    \centering
     \includegraphics[width=\textwidth]{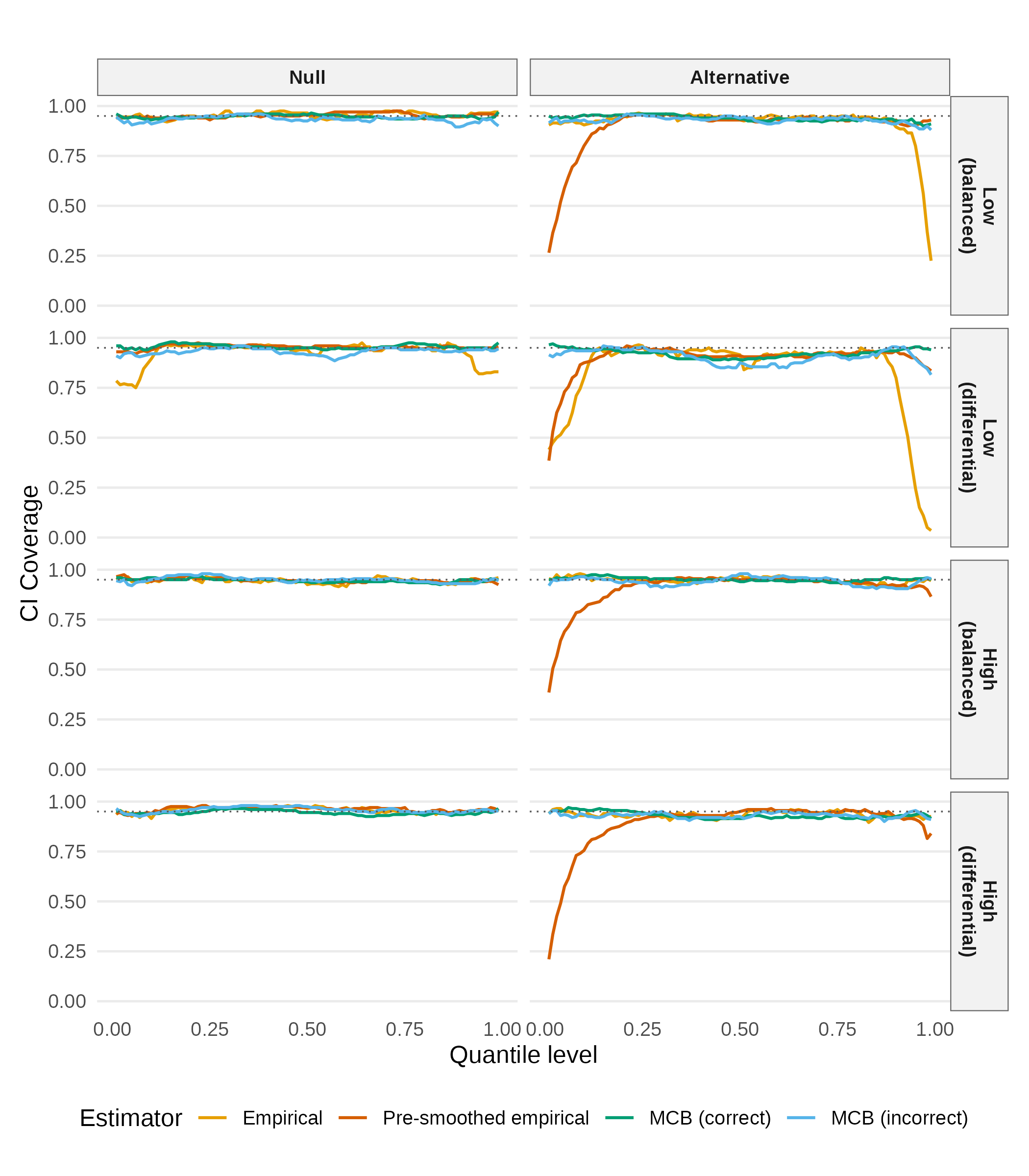}
\caption{95\% confidence interval coverage of the estimated quantile difference map $\widehat{\Delta}(t)$ in Simulation Study 1. Results are shown under the null and alternative sequence feature scenarios, across sequencing depth settings: low balanced, low differential, high balanced, and high differential. Each curve shows the 95\% confidence interval coverage of $\widehat{\Delta}(t)$ across quantile levels $t \in \{0.01,0.02,\ldots,0.99\}$. Estimators include the empirical Wasserstein barycenter, the pre-smoothed empirical Wasserstein barycenter, the correctly specified MCB estimator, and the misspecified spline-based MCB estimator.}
\label{fig:delta_ci}
\end{figure}

\begin{figure}[H]
    \centering
     \includegraphics[width=\textwidth]{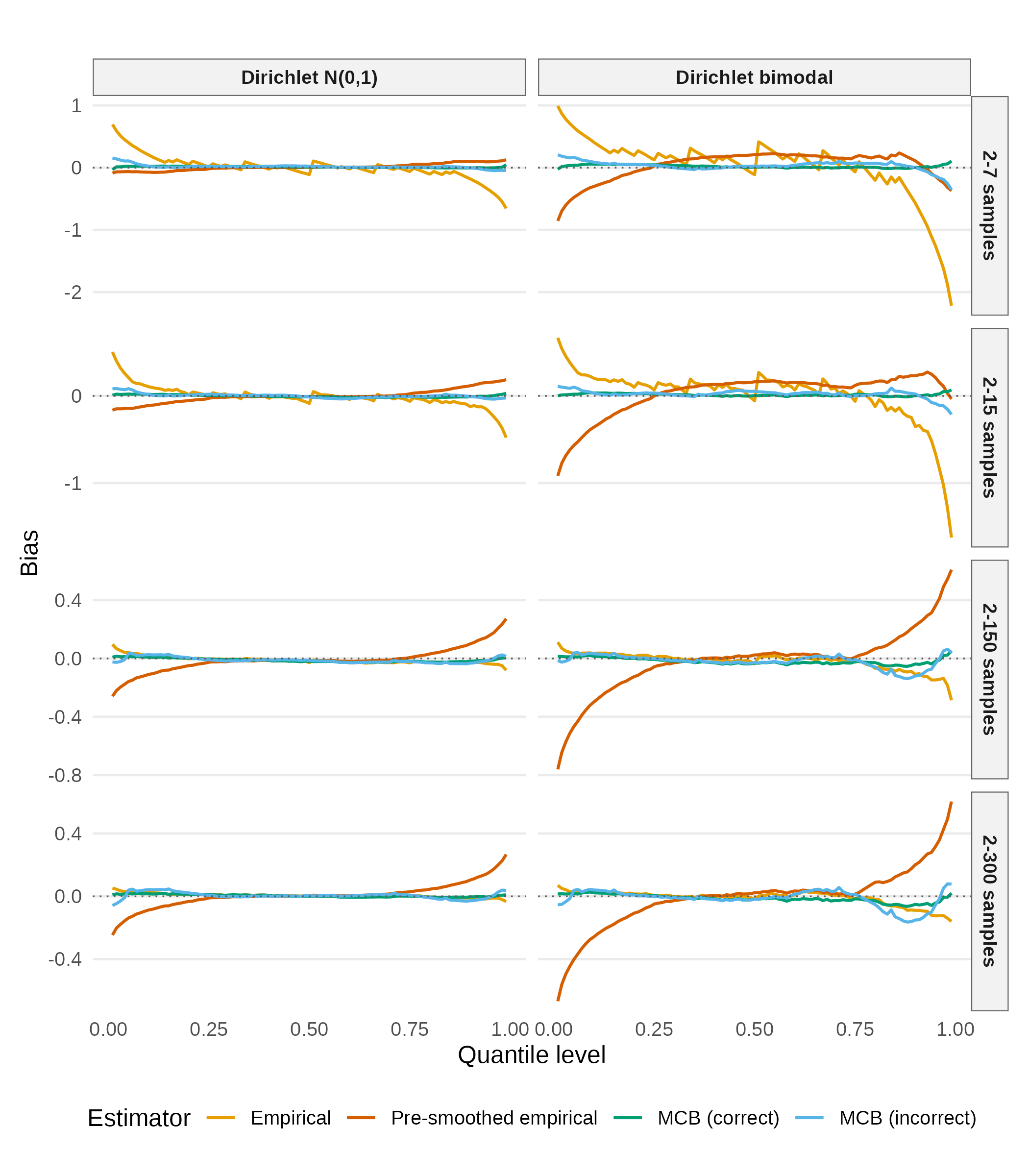}
\caption{Bias of the estimated group-specific Wasserstein barycenter quantile functions in Simulation Study 1. Results are shown for two data-generating settings: individual-level distributions generated from a Dirichlet process with a Normal$(0,1)$ base measure, and individual-level distributions generated from a Dirichlet process with a bimodal Normal mixture base measure. Results are shown across sequencing depths of 2--7, 2--15, 2--150, and 2--300. Each curve shows the bias of the estimated barycenter quantile function across quantile levels $t \in \{0.01,0.02,\ldots,0.99\}$. Estimators include the empirical Wasserstein barycenter, the pre-smoothed empirical Wasserstein barycenter, the correctly specified MCB estimator, and the misspecified spline-based MCB estimator.}
\label{fig:bary_bias}
\end{figure}

\begin{figure}[H]
    \centering
     \includegraphics[width=\textwidth]{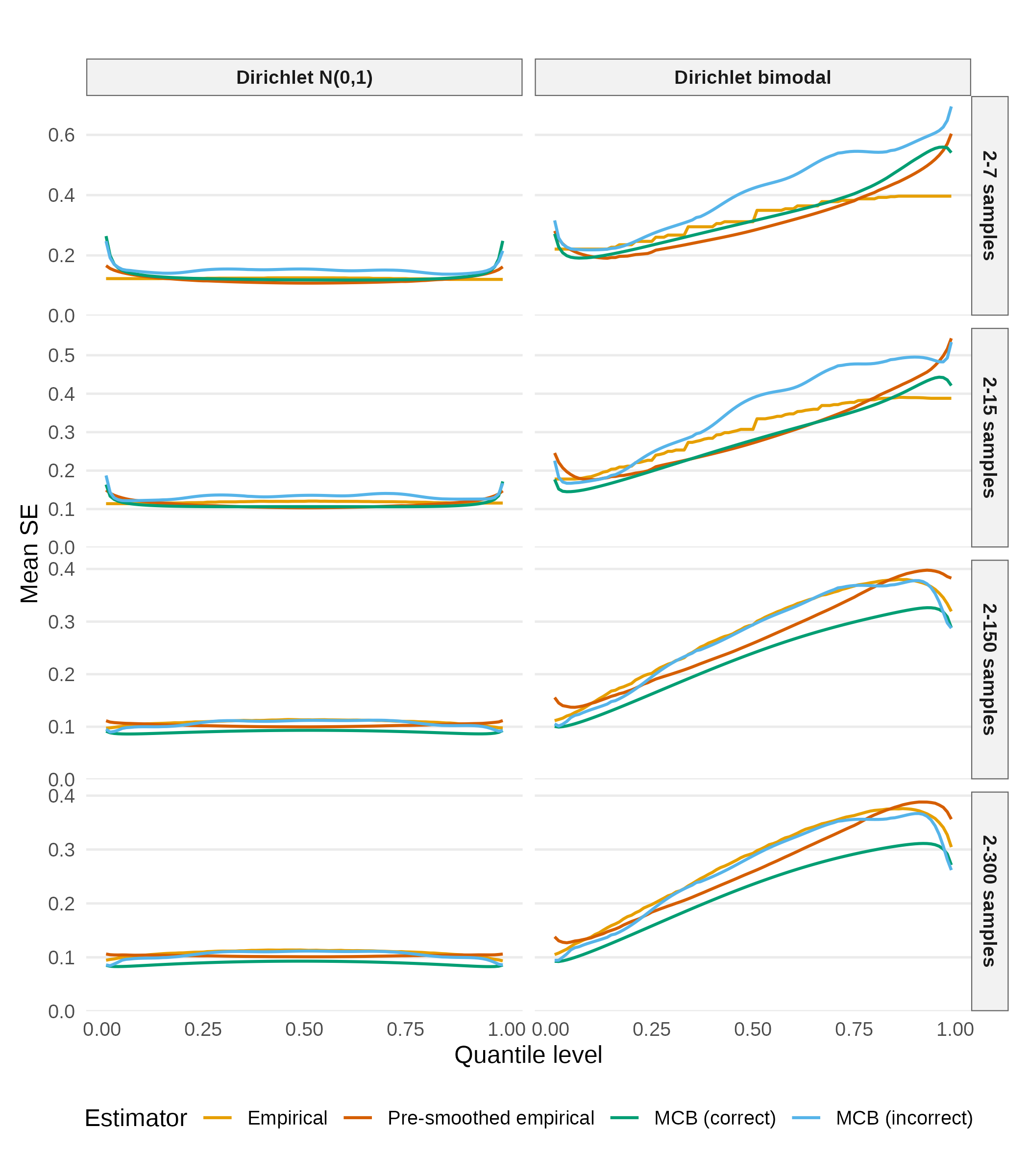}
\caption{Mean estimated standard error of the estimated group-specific Wasserstein barycenter quantile functions in Simulation Study 1. Results are shown for two data-generating settings: individual-level distributions generated from a Dirichlet process with a Normal$(0,1)$ base measure, and individual-level distributions generated from a Dirichlet process with a bimodal Normal mixture base measure. Results are shown across sequencing depths of 2--7, 2--15, 2--150, and 2--300. Each curve shows the mean estimated standard error of the estimated barycenter quantile function across quantile levels $t \in \{0.01,0.02,\ldots,0.99\}$. Estimators include the empirical Wasserstein barycenter, the pre-smoothed empirical Wasserstein barycenter, the correctly specified MCB estimator, and the misspecified spline-based MCB estimator.}
\label{fig:bary_se}
\end{figure}

\begin{figure}[H]
    \centering
     \includegraphics[width=\textwidth]{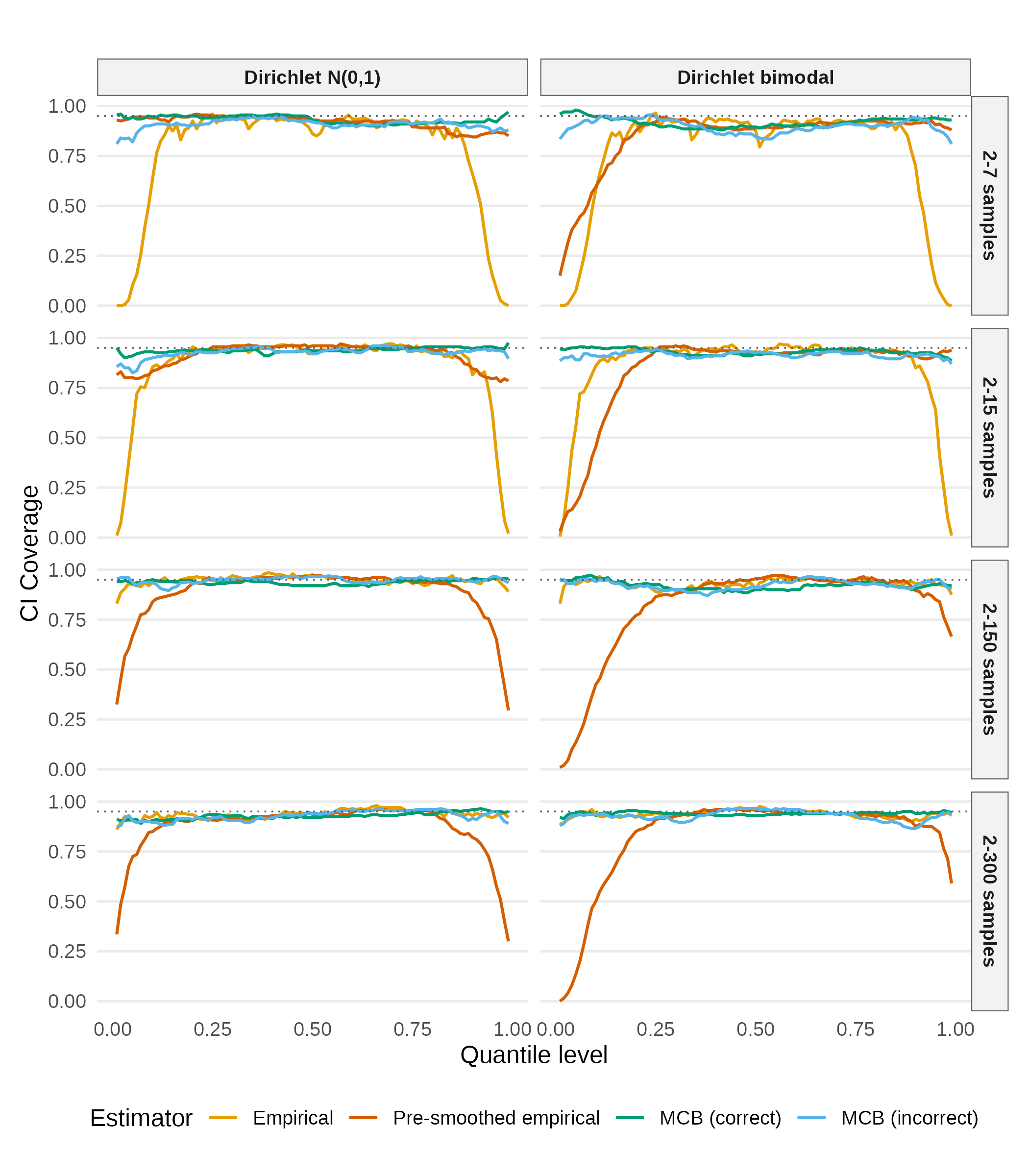}
\caption{95\% confidence interval coverage of the estimated group-specific Wasserstein barycenter quantile functions in Simulation Study 1. Results are shown for two data-generating settings: individual-level distributions generated from a Dirichlet process with a Normal$(0,1)$ base measure, and individual-level distributions generated from a Dirichlet process with a bimodal Normal mixture base measure. Results are shown across sequencing depths of 2--7, 2--15, 2--150, and 2--300. Each curve shows the 95\% confidence interval coverage of the estimated barycenter quantile function across quantile levels $t \in \{0.01,0.02,\ldots,0.99\}$. Estimators include the empirical Wasserstein barycenter, the pre-smoothed empirical Wasserstein barycenter, the correctly specified MCB estimator, and the misspecified spline-based MCB estimator.}
\label{fig:bary_ci}
\end{figure}

\begin{figure}[H]
    \centering
     \includegraphics[width=\textwidth]{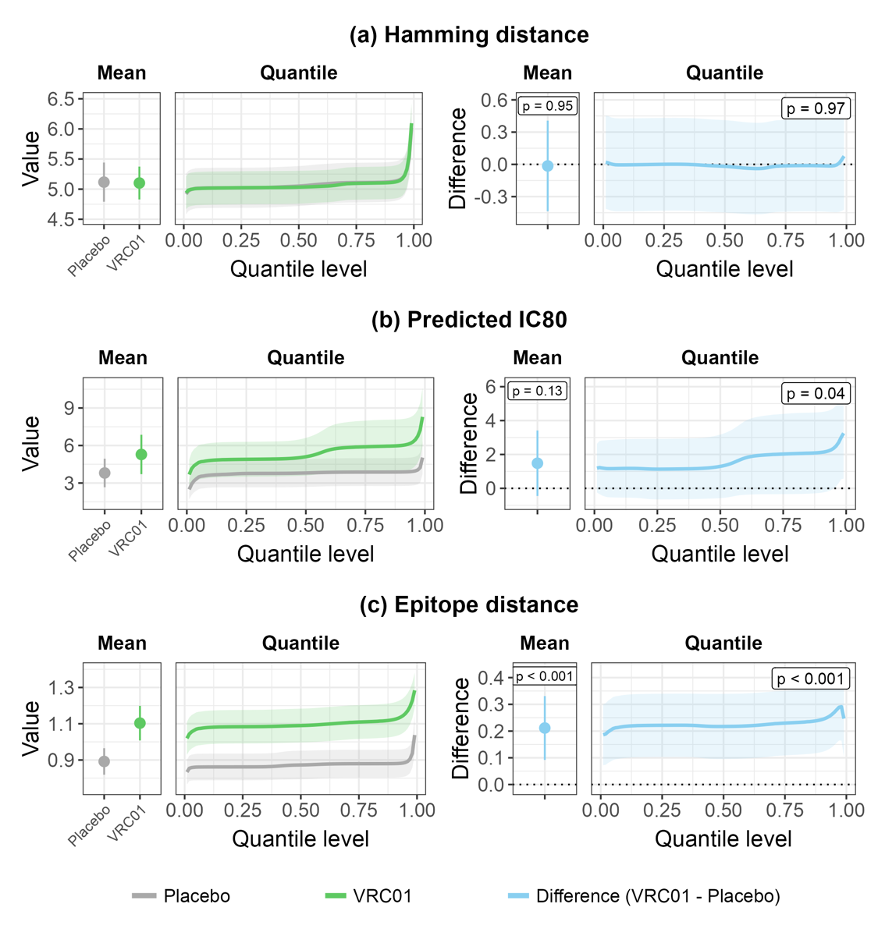}
    \caption{Using the spline-configured MCB estimator, analysis of three VRC01 susceptibility-associated features among AMP breakthrough infections: (A) Hamming distance to a VRC01-sensitive reference sequence at pre-selected sites, (B) predicted IC80, and (C) epitope distance. For each feature, the left panels show the arm-specific mean summaries and estimated barycenter quantile functions for the placebo arm and the pooled VRC01 arms. The right panels show the corresponding VRC01-minus-placebo estimated mean difference and quantile difference map. Shaded bands indicate pointwise 95\% confidence intervals. Reported $p$-values correspond to the mean-based unequal-variance two-sample $t$-test for the mean panels and the supremum test for equality of barycenters for the quantile panels.}    \label{fig:data_analysis_mcbary}
\end{figure}

\begin{table}[H]
\singlespacing
\centering
\caption{BCR data application vaccination schedules, PBMC sampling timepoints,
number of participants, and median (IQR) number of sequences per participant.}
\label{tab:bcr_comparisons}
\begin{tabular}{llccc}
\hline
\makecell[l]{Analysis\\label} &
\makecell[l]{Vaccination\\schedule} &
\makecell[c]{PBMC\\sample} &
\makecell[c]{Unique\\participants} &
\makecell[c]{Median sequences\\(IQR)} \\
\hline
wk8 eOD
& \makecell[l]{eOD at\\week~0}
& Week~8
& 16
& 34 (21.8, 61.2) \\[15pt]

wk16 eOD$\rightarrow$eOD
& \makecell[l]{eOD at week~0;\\eOD at week~8}
& Week~16
& 24
& 101.5 (53, 262.8)  \\[15pt]

wk16 eOD$\rightarrow$core
& \makecell[l]{eOD at week~0;\\core at week~8}
& Week~16
& 16
& 121 (62.8, 318.5)  \\[15pt]

wk24 eOD$\rightarrow$eOD
& \makecell[l]{eOD at week~0;\\eOD at week~8}
& Week~24
& 13
& 160 (53, 187)  \\[15pt]

wk24 eOD$\rightarrow$core
& \makecell[l]{eOD at week~0;\\core at week~8}
& Week~24
& 17
& 265 (36, 494)  \\[15pt]

\makecell[l]{wk24 eOD$\rightarrow$eOD$\rightarrow$core}
& \makecell[l]{eOD at week~0;\\eOD at week~8;\\core at week~16}
& Week~24
& 13
& 217 (143, 593) \\
\hline
\end{tabular}
\end{table}

\newpage

\section{Conditions for pointwise and uniform convergence of the estimated barycenter quantile function}\label{sec:asymptotic_conditions}

Throughout this appendix, we maintain the sampling assumptions described in Section~\ref{sec:methodology}. For notational simplicity, we consider a single group and let $Y_1,\ldots,Y_n$ denote independent and identically distributed realizations of the distribution-valued random object $Y$. For each individual $i$, we assume that the observed sequence-derived features $Y_{i1},\ldots,Y_{im_i}$ are i.i.d. draws from $Y_i$, conditional on $Y_i$. Let $\widehat{Y}^{\mathrm{emp}}_i$ denote the observed empirical distribution for individual $i$ and $\widehat{Y}^{\mathrm{sm}}_i$ denote the pre-smoothed empirical distribution. Suppose that our target is the Wasserstein barycenter $\mathbb{B}[Y]$, as its quantile function $F^{-}_{\mathbb{B}[Y]} (t)$. In the following, we establish conditions for pointwise asymptotic normality and weak convergence to a Gaussian process for the quantile functions of the empirical barycenter $\widehat{\mu}^{\mathrm{emp}}(t) := \frac{1}{n} \sum_{i=1}^{n} F^{-}_{\hat{Y}_i^{\mathrm{emp}}}(t)$, the pre-smoothed empirical barycenter $\widehat{\mu}^{\mathrm{sm}}(t) := \frac{1}{n} \sum_{i=1}^{n} F^{-}_{\hat{Y}_i^{\mathrm{sm}}}(t)$, and the marginal-constructed barycenter $\widehat{\mu}^{\mathrm{mcb}}(t)$ \citep{peng_estimating_2026}.

\subsection{Empirical and pre-smoothed empirical barycenters}
\label{sec:emp_sm_asymptotics}

We first establish high-level conditions for pointwise asymptotic normality and weak convergence to a Gaussian process of the empirical and pre-smoothed empirical Wasserstein barycenters. The two estimators can be handled using the same argument, because both are obtained by first estimating each individual-level quantile function and then averaging these estimated quantile functions across individuals.

Let
\begin{equation*}
    Q_i(t) := F^{-}_{Y_i}(t),
    \qquad
    \mu(t) := F^{-}_{\mathbb{B}[Y]}(t)
    =
    \mathbb{E}\{Q_i(t)\},
    \qquad t \in (0,1),
\end{equation*}
where the last equality follows from the one-dimensional representation of the Wasserstein barycenter. For $a \in \{\mathrm{emp},\mathrm{sm}\}$, define
\begin{equation*}
    \widehat Q_i^{a}(t)
    :=
    F^{-}_{\widehat Y_i^{a}}(t),
    \qquad
    \widehat \mu^{a}(t)
    :=
    \frac{1}{n}\sum_{i=1}^n \widehat Q_i^{a}(t),
    \qquad t\in(0,1),
\end{equation*}
where $a=\mathrm{emp}$ corresponds to the empirical distribution $\widehat Y_i^{\mathrm{emp}}$, and $a=\mathrm{sm}$ corresponds to the pre-smoothed empirical distribution $\widehat Y_i^{\mathrm{sm}}$. Define the average first-stage quantile estimation error
\begin{equation*}
    R_n^{a}(t)
    :=
    \frac{1}{n}
    \sum_{i=1}^n
    \left\{
    \widehat Q_i^{a}(t)-Q_i(t)
    \right\}.
\end{equation*}

The role of the following conditions is to require that the first-stage error from estimating the individual-level quantile functions is asymptotically negligible relative to the $n^{-1/2}$ variation from averaging the latent quantile functions $Q_i$. These conditions may hold for either the empirical or pre-smoothed estimator when the within-individual sample sizes grow sufficiently rapidly as the number of subjects increases, and any smoothing bias vanishes sufficiently quickly. In practice, this means that, for a given dataset, the within-individual sample sizes must be sufficiently large relative to the number of subjects.

\begin{proposition}[Pointwise asymptotic normality of empirical and pre-smoothed empirical barycenters]
\label{prop:emp_sm_pointwise_an}
Fix $t\in(0,1)$ and let $a\in\{\mathrm{emp},\mathrm{sm}\}$. Suppose that the following conditions hold:
\begin{enumerate}
    \item The individual-level quantile value has finite, nonzero variance:
    \begin{equation*}
        0
        <
        \operatorname{Var}\{Q_i(t)\}
        <
        \infty.
    \end{equation*}

    \item The average first-stage quantile estimation error is negligible at the $n^{-1/2}$ scale:
    \begin{equation*}
        \sqrt n\, R_n^{a}(t)
        =
        \sqrt n
        \left[
        \frac{1}{n}
        \sum_{i=1}^n
        \left\{
        \widehat Q_i^{a}(t)-Q_i(t)
        \right\}
        \right]
        =
        o_p(1).
    \end{equation*}
\end{enumerate}
Then
\begin{equation*}
    \sqrt n
    \left\{
    \widehat \mu^{a}(t)-\mu(t)
    \right\}
    \rightsquigarrow
    N\left(
    0,
    \operatorname{Var}\{Q_i(t)\}
    \right).
\end{equation*}
Equivalently, both $\widehat \mu^{\mathrm{emp}}(t)$ and $\widehat \mu^{\mathrm{sm}}(t)$ are pointwise asymptotically normal whenever their corresponding first-stage estimation errors satisfy Condition 2.
\end{proposition}

\begin{proof}
For either $a=\mathrm{emp}$ or $a=\mathrm{sm}$, add and subtract the average of the latent individual-level quantile functions:
\begin{equation*}
\begin{aligned}
    \widehat \mu^{a}(t)-\mu(t)
    &=
    \frac{1}{n}
    \sum_{i=1}^n
    \left\{
    Q_i(t)-\mu(t)
    \right\} \\
    &\quad+
    \frac{1}{n}
    \sum_{i=1}^n
    \left\{
    \widehat Q_i^{a}(t)-Q_i(t)
    \right\}.
\end{aligned}
\end{equation*}
Since $\mu(t)=\mathbb{E}\{Q_i(t)\}$, the first term satisfies
\begin{equation*}
    \sqrt n
    \left[
    \frac{1}{n}
    \sum_{i=1}^n
    \left\{
    Q_i(t)-\mu(t)
    \right\}
    \right]
    \rightsquigarrow
    N\left(
    0,
    \operatorname{Var}\{Q_i(t)\}
    \right)
\end{equation*}
by the central limit theorem. By Condition 2, the second term is $o_p(1)$ after multiplication by $\sqrt n$. The result follows by Slutsky's theorem.
\end{proof}

\begin{proposition}[Weak convergence of empirical and pre-smoothed empirical barycenters]
\label{prop:emp_sm_weak_convergence}
Let $\mathcal T=[\epsilon,1-\epsilon]$ for some $\epsilon\in(0,1/2)$, and let $a\in\{\mathrm{emp},\mathrm{sm}\}$. Define the latent empirical process
\begin{equation*}
    \mathbb Z_n(t)
    :=
    \frac{1}{\sqrt n}
    \sum_{i=1}^n
    \left\{
    Q_i(t)-\mathbb E Q_i(t)
    \right\},
    \qquad t\in\mathcal T.
\end{equation*}
Suppose that the following conditions hold:
\begin{enumerate}
    \item The latent quantile functions $Q_i$ are i.i.d. random elements of $\ell^\infty(\mathcal T)$ with finite second moment in sup-norm:
    \begin{equation*}
        \mathbb E
        \left[
        \|Q_i\|_{\infty,\mathcal T}^{2}
        \right]
        <
        \infty,
        \qquad
        \|Q_i\|_{\infty,\mathcal T}
        :=
        \sup_{t\in\mathcal T}|Q_i(t)|.
    \end{equation*}

    \item The latent empirical process $\mathbb Z_n$ is stochastically equicontinuous over $\mathcal T$. That is, for every $\varepsilon>0$,
    \begin{equation*}
        \lim_{\delta\downarrow 0}
        \limsup_{n\to\infty}
        P
        \left(
        \sup_{\substack{s,t\in\mathcal T\\ |s-t|<\delta}}
        |\mathbb Z_n(s)-\mathbb Z_n(t)|
        >
        \varepsilon
        \right)
        =
        0.
    \end{equation*}

    \item The average first-stage quantile estimation error is negligible uniformly over $\mathcal T$:
    \begin{equation*}
        \sqrt n
        \left\|
        R_n^{a}
        \right\|_{\infty,\mathcal T}
        =
        o_p(1).
    \end{equation*}
\end{enumerate}
Then
\begin{equation*}
    \sqrt n
    \left\{
    \widehat \mu^{a}-\mu
    \right\}
    \rightsquigarrow
    \mathbb G
    \qquad
    \text{in } \ell^\infty(\mathcal T),
\end{equation*}
where $\mathbb G$ is a tight mean-zero Gaussian process with covariance kernel
\begin{equation*}
    \Sigma(s,t)
    =
    \operatorname{Cov}\{Q_i(s),Q_i(t)\},
    \qquad s,t\in\mathcal T.
\end{equation*}
\end{proposition}

\begin{proof}
For either $a=\mathrm{emp}$ or $a=\mathrm{sm}$,
\begin{equation*}
\begin{aligned}
    \widehat \mu^{a}-\mu
    &=
    \frac{1}{n}
    \sum_{i=1}^n
    \left\{
    Q_i-\mu
    \right\}
    +
    R_n^{a} \\
    &=
    \frac{1}{n}
    \sum_{i=1}^n
    \left\{
    Q_i-\mathbb E(Q_i)
    \right\}
    +
    R_n^{a}.
\end{aligned}
\end{equation*}
Multiplying by $\sqrt n$ gives
\begin{equation*}
    \sqrt n
    \left\{
    \widehat \mu^{a}-\mu
    \right\}
    =
    \mathbb Z_n
    +
    \sqrt n R_n^{a}.
\end{equation*}

For any finite collection $t_1,\ldots,t_k\in\mathcal T$, the random vectors
\begin{equation*}
    \left(
    Q_i(t_1)-\mathbb E Q_i(t_1),
    \ldots,
    Q_i(t_k)-\mathbb E Q_i(t_k)
    \right)^\top
\end{equation*}
are i.i.d. mean-zero random vectors with finite second moments by Condition 1. Therefore, by the multivariate central limit theorem,
\begin{equation*}
    \left(
    \mathbb Z_n(t_1),
    \ldots,
    \mathbb Z_n(t_k)
    \right)^\top
    \rightsquigarrow
    N_k(0,\boldsymbol{\Sigma}),
\end{equation*}
where
\begin{equation*}
    \boldsymbol{\Sigma}_{j\ell}
    =
    \operatorname{Cov}\{Q_i(t_j),Q_i(t_\ell)\},
    \qquad 1\leq j,\ell\leq k.
\end{equation*}
For each fixed collection $t_1,\ldots,t_k\in\mathcal T$, 
$(\mathbb Z_n(t_1),\ldots,\mathbb Z_n(t_k))^\top$ converges to a mean-zero multivariate Gaussian vector with covariance matrix 
$\{\Sigma(t_j,t_\ell)\}_{j,\ell=1}^k$.

Condition 2 gives stochastic equicontinuity of $\mathbb Z_n$ over $\mathcal T$. Under the usual asymptotic measurability conditions for empirical processes, finite-dimensional convergence together with stochastic equicontinuity yields
\begin{equation*}
    \mathbb Z_n
    \rightsquigarrow
    \mathbb G
    \qquad
    \text{in } \ell^\infty(\mathcal T),
\end{equation*}
where $\mathbb G$ is a tight mean-zero Gaussian process with covariance kernel $\Sigma$  \citep[Theorem~18.14]{van2000asymptotic} .

By Condition 3,
\begin{equation*}
    \sqrt n
    \left\|
    R_n^{a}
    \right\|_{\infty,\mathcal T}
    =
    o_p(1).
\end{equation*}
Therefore,
\begin{equation*}
    \sqrt n
    \left\{
    \widehat \mu^{a}-\mu
    \right\}
    =
    \mathbb Z_n
    +
    o_p(1)
\end{equation*}
in $\ell^\infty(\mathcal T)$. The result follows by Slutsky's theorem in $\ell^\infty(\mathcal T)$ \citep[Theorem~18.10]{van2000asymptotic}.
\end{proof}

\begin{remark}[Interpretation of the first-stage negligibility conditions]
The empirical and pre-smoothed empirical barycenters differ only through the first-stage estimate $\widehat Q_i^{a}$ of the latent quantile function $Q_i$. The propositions above therefore separate the asymptotic variation into two parts: the between-individual variation in the latent quantile functions, and the within-individual first-stage estimation error. For the latter to be asymptotically negligible at the $n^{-1/2}$ scale, the within-individual sample sizes must increase sufficiently rapidly with the number of subjects, with sufficiently small smoothing bias if applicable. Under these conditions, the first-stage estimation error does not contribute to the first-order limiting distribution. In practice, these conditions require sufficiently large within-individual sample sizes relative to the number of subjects and may therefore be unrealistic in sparse sequencing settings. When these conditions are not satisfied, particularly in the tails, the first-stage estimation error may remain non-negligible, potentially introducing bias and preventing root-$n$ consistent estimation of the target barycenter quantile function.
\end{remark}

\subsection{Marginal-constructed barycenter}

We now state high-level conditions for pointwise asymptotic normality and weak convergence of the marginal-constructed barycenter estimator. The key idea is to estimate the distribution of the latent individual-level quantiles indirectly, through binomial mixture models for latent CDF values.

For each individual, let $Q_i(t)=F_{Y_i}^{-}(t)$ denote the latent individual-level quantile function. Fix $y\in\mathcal I$, where $\mathcal I\subseteq\mathbb R$ denotes the support of the sequence-derived feature. The latent CDF value at this threshold is $U_i(y)=F_{Y_i}(y)$, which is the probability that a draw from the $i$th latent distribution is less than or equal to $y$. Although $U_i(y)$ is not observed directly, the observed sample $Y_{i1},\ldots,Y_{im_i}$ provides the count $C_i(y)=\sum_{j=1}^{m_i}\mathbf 1\{Y_{ij}\le y\}$. Conditional on $m_i$ and $U_i(y)$,
\begin{equation*}
    C_i(y)\mid m_i,U_i(y)
    \sim
    \operatorname{Bin}\{m_i,U_i(y)\}.
\end{equation*}
Thus, for each fixed threshold $y$, the observed counts follow a binomial mixture model whose mixing distribution is the distribution of $U_i(y)$ across individuals. We denote this mixing distribution by
\begin{equation*}
    G(y,u)
    :=
    P\{U_i(y)\le u\}
    =
    P\{F_{Y_i}(y)\le u\},
    \qquad u\in[0,1].
\end{equation*}

The connection to barycenter estimation comes from translating statements about latent quantiles into statements about latent CDF values. For a fixed quantile level $t\in(0,1)$, let $H_t(y)=P\{Q_i(t)\le y\}$ denote the distribution function of the latent $t$-quantiles across individuals. Since the event $Q_i(t)\le y$ is equivalent to the event $F_{Y_i}(y)\ge t$, we have $H_t(y)=1-G(y,t-)$, where $G(y,t-):=\lim_{u\uparrow t}G(y,u)$. Therefore, the barycenter quantile function satisfies
\begin{equation*}
    \mu(t)
    =
    \mathbb E\{Q_i(t)\}
    =
    \int_{\mathcal I} y\,dH_t(y) = - \int_{\mathcal I} y\,dG(y, t-) 
\end{equation*}

The marginal-constructed barycenter estimator replaces $G(y,\cdot)$ by an estimator $\widehat G_n(y,\cdot)$ obtained from the binomial mixture model. Define $\widehat H_{n,t}(y):=1-\widehat G_n(y,t-)$ and
\begin{equation*}
    \widehat \mu^{\mathrm{mcb}}(t)
    :=
    \int_{\mathcal I} y\,d\widehat H_{n,t}(y).
\end{equation*}
In the following, $O_i=(m_i,Y_{i1},\ldots,Y_{im_i})$ denotes the observed data for individual $i$. The following propositions state high-level conditions under which $\widehat \mu^{\mathrm{mcb}}$ is asymptotically normal pointwise and converges weakly as a process. The proofs follow from arguments in  \citet{peng_estimating_2026}.

\begin{proposition}[Pointwise asymptotic normality of the marginal-constructed barycenter]
\label{prop:mcb_pointwise_an}
Fix $t\in(0,1)$. Suppose that the following conditions hold:
\begin{enumerate}
    \item The feature support is bounded, $\mathcal I=[a,b]\subset\mathbb R$ with $|a|,|b|<\infty$. With probability tending to one, $y\mapsto \widehat H_{n,t}(y)$ is a cumulative distribution function on $\mathcal I$.

    \item The estimator $\widehat H_{n,t}$ admits the integrated asymptotic linear expansion
    \begin{equation*}
        \sqrt n
        \left\{
        \widehat H_{n,t}(y)-H_t(y)
        \right\}
        =
        \frac{1}{\sqrt n}
        \sum_{i=1}^n
        \psi_t(O_i;y)
        +
        r_{n,t}(y),
        \qquad y\in\mathcal I,
    \end{equation*}
    where $\mathbb E\{\psi_t(O_i;y)\}=0$ for each fixed $y$,
    \begin{equation*}
        \mathbb E
        \left[
        \int_{\mathcal I}
        |\psi_t(O_i;y)|\,dy
        \right]
        <
        \infty,
    \end{equation*}
    and
    \begin{equation*}
        \int_{\mathcal I}
        |r_{n,t}(y)|\,dy
        =
        o_p(1).
    \end{equation*}

    \item Define the integrated influence function
    \begin{equation*}
        \phi^{\mathrm{mcb}}(O_i;t)
        :=
        -\int_{\mathcal I}
        \psi_t(O_i;y)\,dy.
    \end{equation*}
    Assume that
    \begin{equation*}
        0
        <
        \operatorname{Var}
        \left\{
        \phi^{\mathrm{mcb}}(O_i;t)
        \right\}
        =
        \sigma_{\mathrm{mcb}}^2(t)
        <
        \infty.
    \end{equation*}
\end{enumerate}
Then
\begin{equation*}
    \sqrt n
    \left\{
    \widehat \mu^{\mathrm{mcb}}(t)-\mu(t)
    \right\}
    =
    \frac{1}{\sqrt n}
    \sum_{i=1}^n
    \phi^{\mathrm{mcb}}(O_i;t)
    +
    o_p(1),
\end{equation*}
and hence
\begin{equation*}
    \sqrt n
    \left\{
    \widehat \mu^{\mathrm{mcb}}(t)-\mu(t)
    \right\}
    \rightsquigarrow
    N
    \left(
    0,
    \sigma_{\mathrm{mcb}}^2(t)
    \right).
\end{equation*}
\end{proposition}

\begin{proof}
See \citet{peng_estimating_2026}.
\end{proof}

\begin{proposition}[Weak convergence of the marginal-constructed barycenter]
\label{prop:mcb_weak_convergence}
Let $\mathcal T=[\epsilon,1-\epsilon]$ for some $\epsilon\in(0,1/2)$. Suppose that the following conditions hold:
\begin{enumerate}
    \item The feature support is bounded, $\mathcal I=[a,b]\subset\mathbb R$ with $|a|,|b|<\infty$. With probability tending to one, $y\mapsto \widehat H_{n,t}(y)$ is a cumulative distribution function on $\mathcal I$ uniformly over $t\in\mathcal T$.

    \item Uniformly over $t\in\mathcal T$, the estimator $\widehat H_{n,t}$ admits the asymptotic linear expansion
    \begin{equation*}
        \sqrt n
        \left\{
        \widehat H_{n,t}(y)-H_t(y)
        \right\}
        =
        \frac{1}{\sqrt n}
        \sum_{i=1}^n
        \psi_t(O_i;y)
        +
        r_{n,t}(y),
        \qquad y\in\mathcal I,
    \end{equation*}
    where $\mathbb E\{\psi_t(O_i;y)\}=0$ for all $t\in\mathcal T$ and $y\in\mathcal I$, and
    \begin{equation*}
        \sup_{t\in\mathcal T}
        \int_{\mathcal I}
        |r_{n,t}(y)|\,dy
        =
        o_p(1).
    \end{equation*}

    \item For each $t\in\mathcal T$, define
    \begin{equation*}
        \phi^{\mathrm{mcb}}(O_i;t)
        :=
        -\int_{\mathcal I}
        \psi_t(O_i;y)\,dy.
    \end{equation*}
    The class of functions $\mathcal F_{\mathrm{mcb}}
        :=
        \left\{
        O\mapsto \phi^{\mathrm{mcb}}(O;t)
        :
        t\in\mathcal T
        \right\}$
    is $P$-Donsker.
\end{enumerate}
Then
\begin{equation*}
    \sqrt n
    \left\{
    \widehat \mu^{\mathrm{mcb}}-\mu
    \right\}
    \rightsquigarrow
    \mathbb G_{\mathrm{mcb}}
    \qquad
    \text{in } \ell^\infty(\mathcal T),
\end{equation*}
where $\mathbb G_{\mathrm{mcb}}$ is a tight mean-zero Gaussian process with covariance kernel
\begin{equation*}
    \Sigma_{\mathrm{mcb}}(s,t)
    =
    \mathbb E
    \left[
    \phi^{\mathrm{mcb}}(O_i;s)
    \phi^{\mathrm{mcb}}(O_i;t)
    \right],
    \qquad s,t\in\mathcal T.
\end{equation*}
\end{proposition}

\begin{proof}
See \citet{peng_estimating_2026}.
\end{proof}

\newpage

\section{Additional simulation results}\label{sec:additional_sims}

\subsection{Simulation Study 1}\label{sec:additional_sim1}

In addition to 50 individuals per group, we conducted Simulation Study 1 with 20 and 100 individuals per group. Results are shown below. 

\subsubsection{20 individuals per group}

\paragraph{Null scenario summary results}\leavevmode\par\vspace{0.5em}

\begin{center}
    \includegraphics[width=\textwidth,height=0.75\textheight,keepaspectratio]{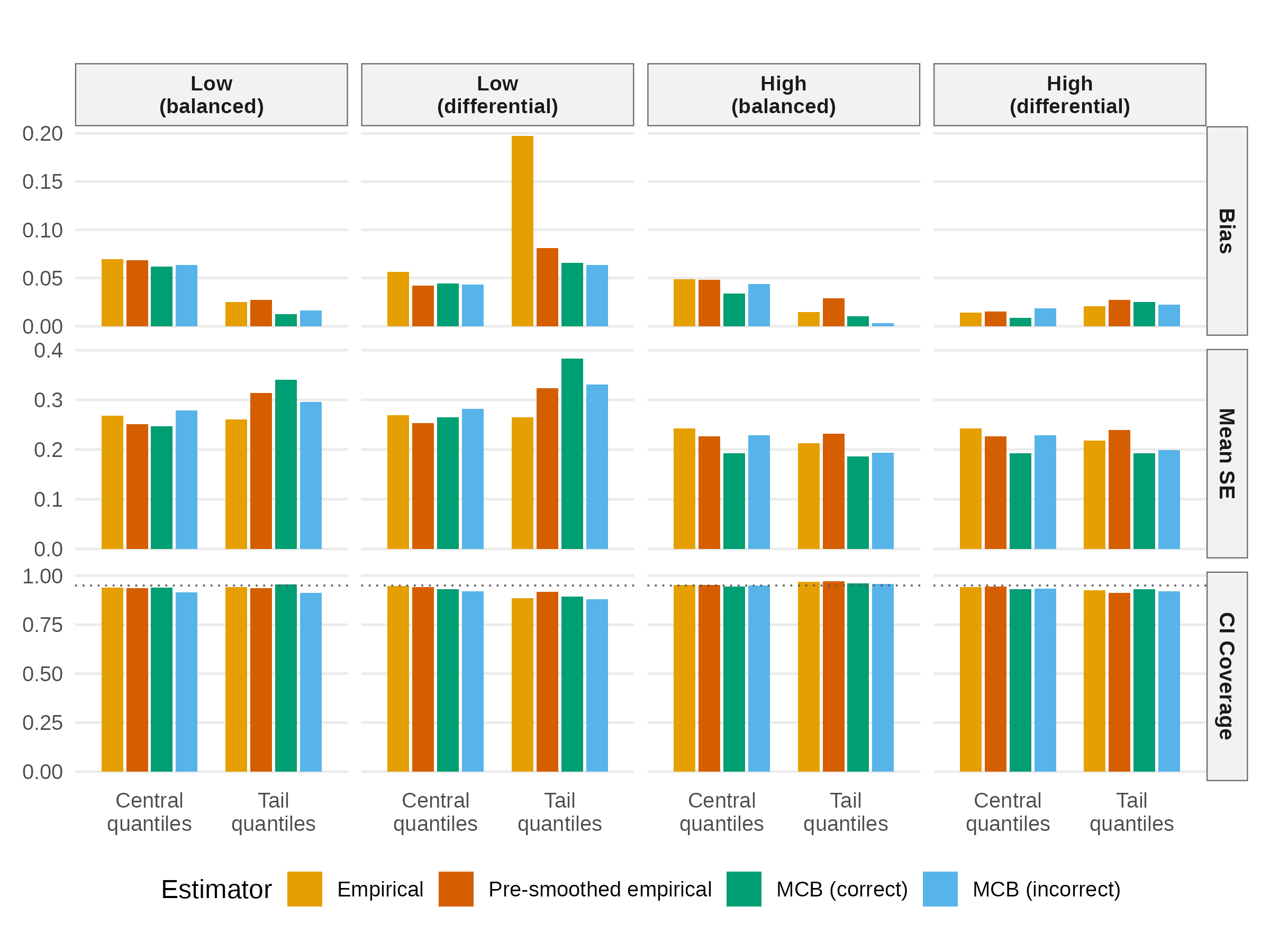}
\captionof{figure}{Summary of estimation of the difference in quantiles $\Delta(t)$ under the null scenario for Simulation Study 1 with 20 individuals per group. Results are shown across four sequencing depth settings: low balanced, low differential, high balanced, and high differential sequencing depth. For each setting, we compare the empirical Wasserstein barycenter, pre-smoothed empirical barycenter, correctly specified MCB estimator, and misspecified spline-based MCB estimator. Performance is summarized by average absolute bias, mean estimated standard error, and pointwise confidence interval coverage for $\Delta(t)$, averaged separately over the central quantiles ($0.06$--$0.94$) and tail quantiles ($0.01$--$0.05$ and $0.95$--$0.99$). The dotted horizontal line in the coverage panels indicates the nominal 95\% coverage level.}    \label{fig:sim1_summary_diff_null_n20}
\end{center}

\clearpage
\paragraph{Alternative scenario summary results}\leavevmode\par\vspace{0.5em}

\begin{center}
    \includegraphics[width=\textwidth,height=0.75\textheight,keepaspectratio]{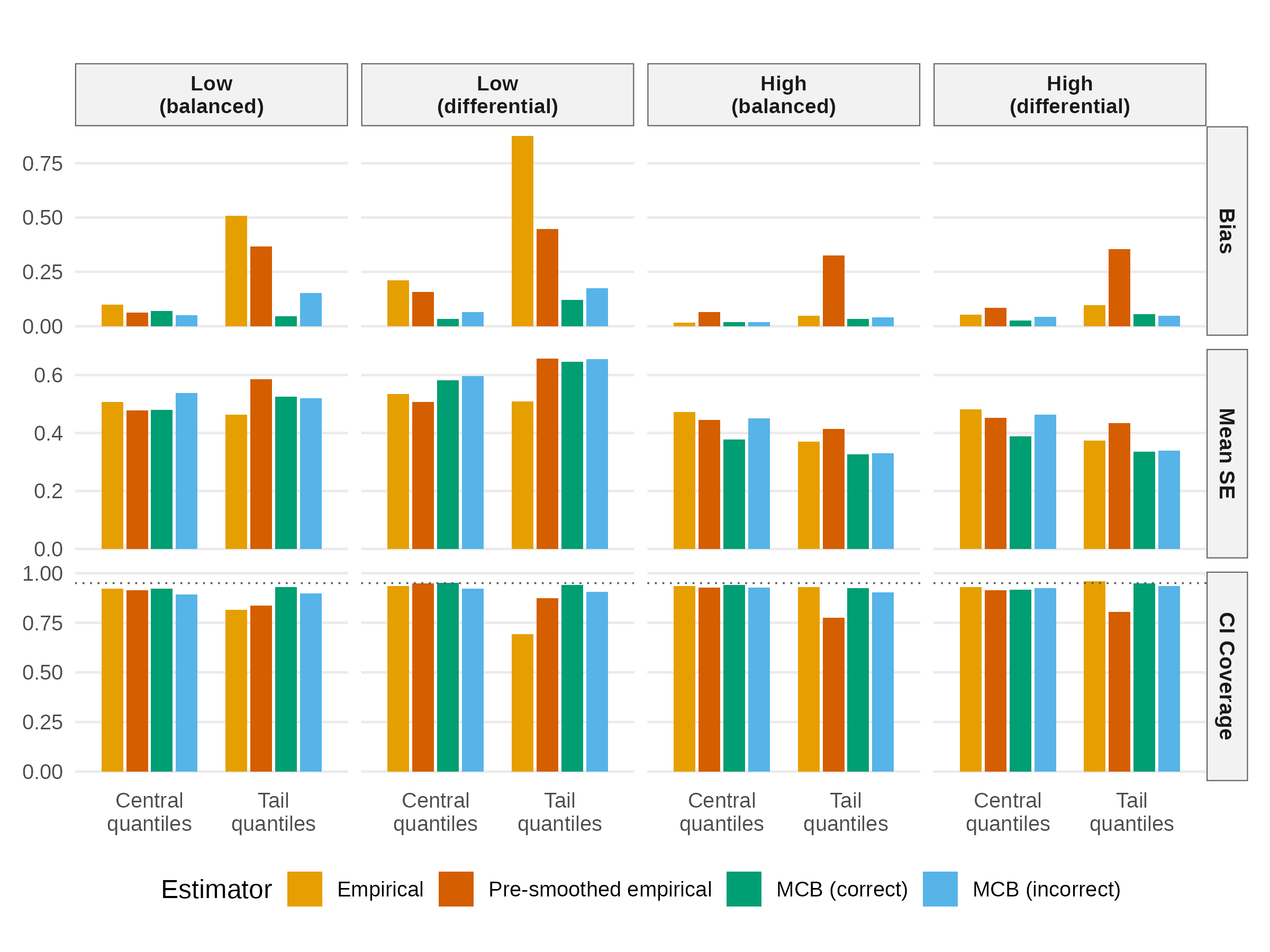}
\captionof{figure}{Summary of estimation of the difference in quantiles $\Delta(t)$ under the alternative scenario for Simulation Study 1 with 20 individuals per group. Results are shown across four sequencing depth settings: low balanced, low differential, high balanced, and high differential sequencing depth. For each setting, we compare the empirical Wasserstein barycenter, pre-smoothed empirical barycenter, correctly specified MCB estimator, and misspecified spline-based MCB estimator. Performance is summarized by average absolute bias, mean estimated standard error, and pointwise confidence interval coverage for $\Delta(t)$, averaged separately over the central quantiles ($0.06$--$0.94$) and tail quantiles ($0.01$--$0.05$ and $0.95$--$0.99$). The dotted horizontal line in the coverage panels indicates the nominal 95\% coverage level.}    \label{fig:sim1_summary_diff_alt_n20}
\end{center}

\clearpage
\paragraph{Quantile difference map $\Delta(t)$: pointwise bias}\leavevmode\par\vspace{0.5em}

\begin{center}    \includegraphics[width=\textwidth,height=0.75\textheight,keepaspectratio]{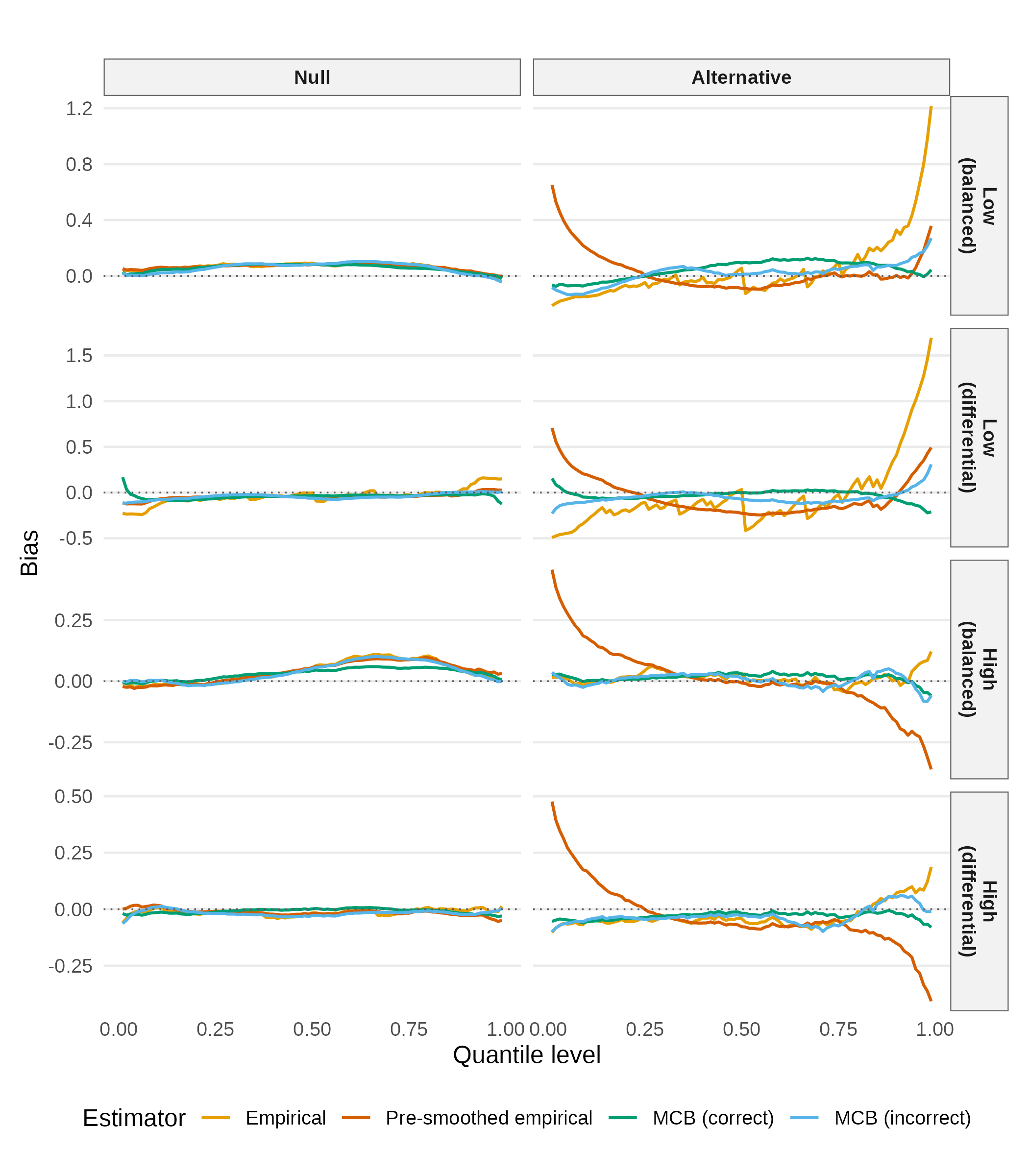}
    \captionof{figure}{Bias of the estimated quantile difference map $\widehat{\Delta}(t)$ in Simulation Study 1 with 20 individuals per group. Results are shown under the null and alternative sequence feature scenarios, across sequencing depth scenarios: low balanced, low differential, high balanced, and high differential. Each curve shows the bias of $\widehat{\Delta}(t)$ across quantile levels $t \in \{0.01,0.02,\ldots,0.99\}$. Estimators include the empirical Wasserstein barycenter, the pre-smoothed empirical Wasserstein barycenter, the correctly specified MCB estimator, and the misspecified spline-based MCB estimator.}
    \label{fig:sim1_diff_bias_n20}
\end{center}

\clearpage
\paragraph{Quantile difference map $\Delta(t)$: pointwise mean estimated standard error}\leavevmode\par\vspace{0.5em}

\begin{center}
    \includegraphics[width=\textwidth,height=0.75\textheight,keepaspectratio]{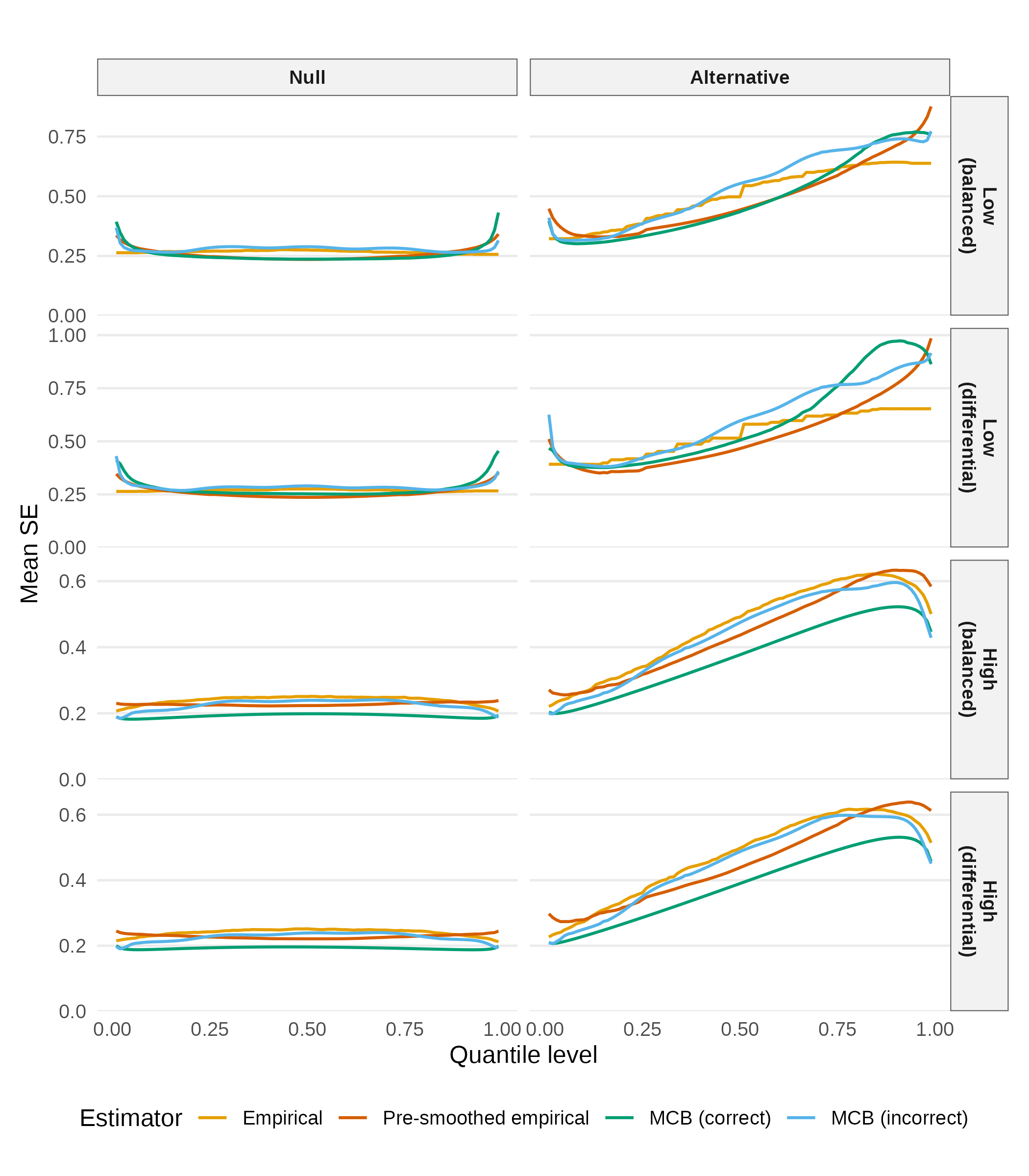}
    \captionof{figure}{Mean estimated standard error of the estimated quantile difference map $\widehat{\Delta}(t)$ in Simulation Study 1 with 20 individuals per group. Results are shown under the null and alternative sequence feature scenarios, across sequencing depth scenarios: low balanced, low differential, high balanced, and high differential. Each curve shows the mean estimated standard error of $\widehat{\Delta}(t)$ across quantile levels $t \in \{0.01,0.02,\ldots,0.99\}$. Estimators include the empirical Wasserstein barycenter, the pre-smoothed empirical Wasserstein barycenter, the correctly specified MCB estimator, and the misspecified spline-based MCB estimator.}
    \label{fig:sim1_diff_mean_se_n20}
\end{center}

\clearpage
\paragraph{Quantile difference map $\Delta(t)$: pointwise confidence coverage}\leavevmode\par\vspace{0.5em}

\begin{center}
    \includegraphics[width=\textwidth,height=0.75\textheight,keepaspectratio]{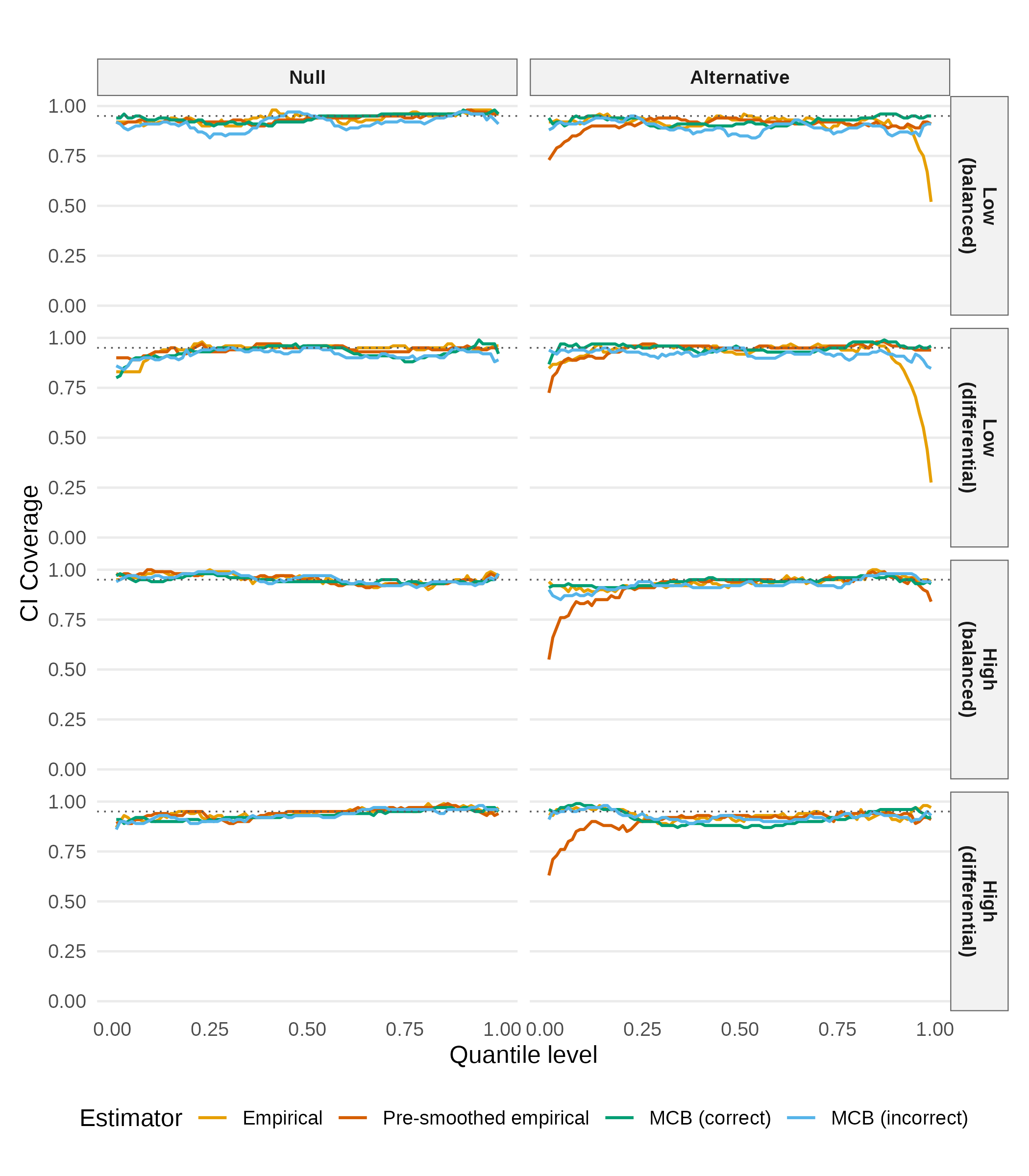}
    \captionof{figure}{95\% confidence interval coverage of the estimated quantile difference map $\widehat{\Delta}(t)$ in Simulation Study 1 with 20 individuals per group. Results are shown under the null and alternative sequence feature scenarios, across sequencing depth scenarios: low balanced, low differential, high balanced, and high differential. Each curve shows the 95\% confidence interval coverage of $\widehat{\Delta}(t)$ across quantile levels $t \in \{0.01,0.02,\ldots,0.99\}$. Estimators include the empirical Wasserstein barycenter, the pre-smoothed empirical Wasserstein barycenter, the correctly specified MCB estimator, and the misspecified spline-based MCB estimator.}
    \label{fig:sim1_diff_ci_coverage_n20}
\end{center}

\clearpage
\paragraph{Group-specific barycenter $\mu_g(t)$: pointwise bias}\leavevmode\par\vspace{0.5em}

\begin{center}
    \includegraphics[width=\textwidth,height=0.75\textheight,keepaspectratio]{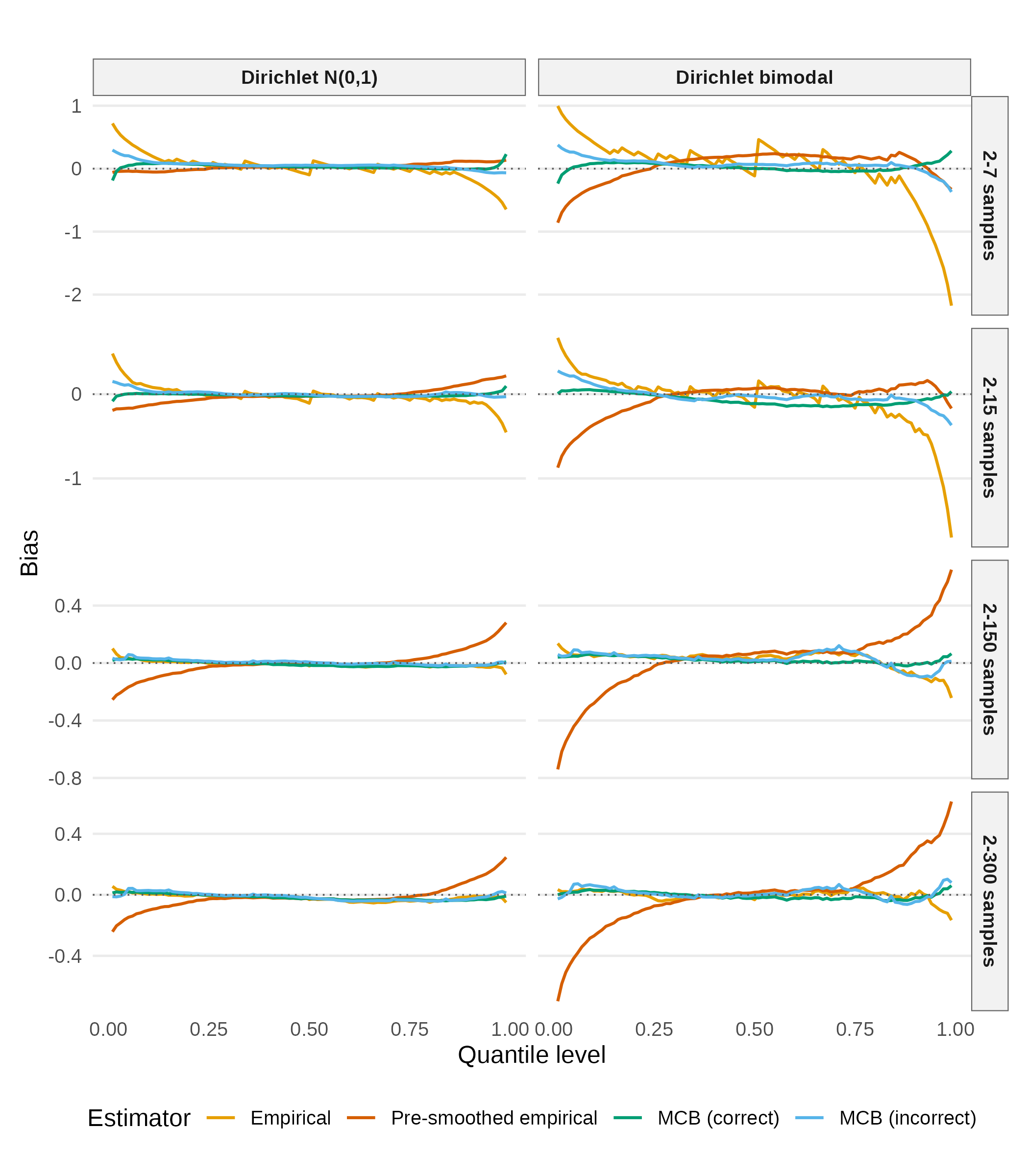}
    \captionof{figure}{Bias of the estimated group-specific Wasserstein barycenter quantile functions in Simulation Study 1 with 20 individuals per group. Results are shown for two data-generating settings: individual-level distributions generated from a Dirichlet process with a Normal$(0,1)$ base measure, and individual-level distributions generated from a Dirichlet process with a bimodal Normal mixture base measure. Results are shown across sequencing depths of 2--7, 2--15, 2--150, and 2--300. Each curve shows the bias of the estimated barycenter quantile function across quantile levels $t \in \{0.01,0.02,\ldots,0.99\}$. Estimators include the empirical Wasserstein barycenter, the pre-smoothed empirical Wasserstein barycenter, the correctly specified MCB estimator, and the misspecified spline-based MCB estimator.}
    \label{fig:sim1_bary_bias_n20}
\end{center}

\clearpage
\paragraph{Group-specific barycenter $\mu_g(t)$: pointwise mean estimated standard error}\leavevmode\par\vspace{0.5em}

\begin{center}
    \includegraphics[width=\textwidth,height=0.75\textheight,keepaspectratio]{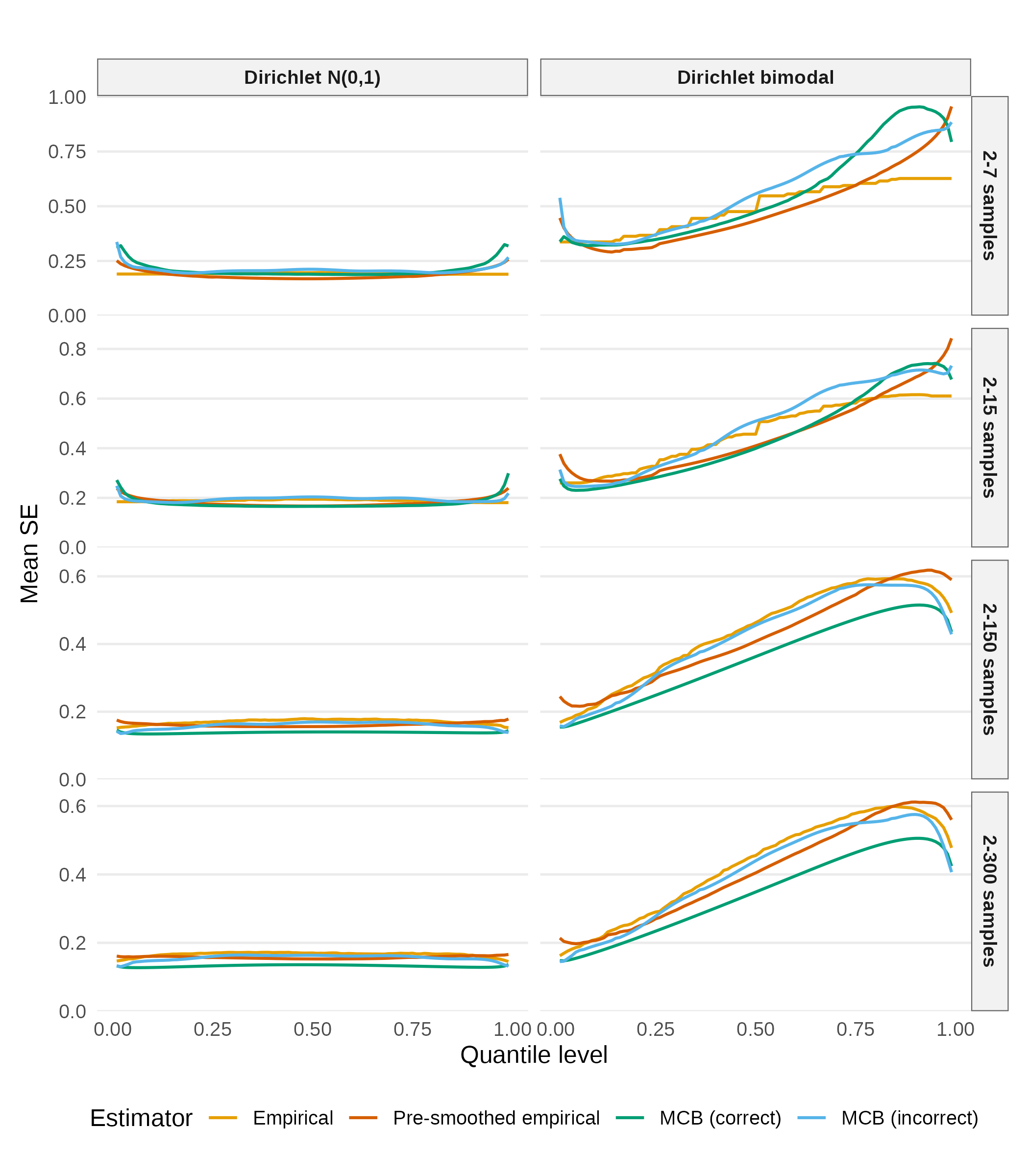}
    \captionof{figure}{Mean estimated standard error of the estimated group-specific Wasserstein barycenter quantile functions in Simulation Study 1 with 20 individuals per group. Results are shown for two data-generating settings: individual-level distributions generated from a Dirichlet process with a Normal$(0,1)$ base measure, and individual-level distributions generated from a Dirichlet process with a bimodal Normal mixture base measure. Results are shown across sequencing depths of 2--7, 2--15, 2--150, and 2--300. Each curve shows the mean estimated standard error of the estimated barycenter quantile function across quantile levels $t \in \{0.01,0.02,\ldots,0.99\}$. Estimators include the empirical Wasserstein barycenter, the pre-smoothed empirical Wasserstein barycenter, the correctly specified MCB estimator, and the misspecified spline-based MCB estimator.}
    \label{fig:sim1_bary_mean_se_n20}
\end{center}

\clearpage
\paragraph{Group-specific barycenter $\mu_g(t)$: pointwise confidence interval coverage}\leavevmode\par\vspace{0.5em}

\begin{center}
    \includegraphics[width=\textwidth,height=0.75\textheight,keepaspectratio]{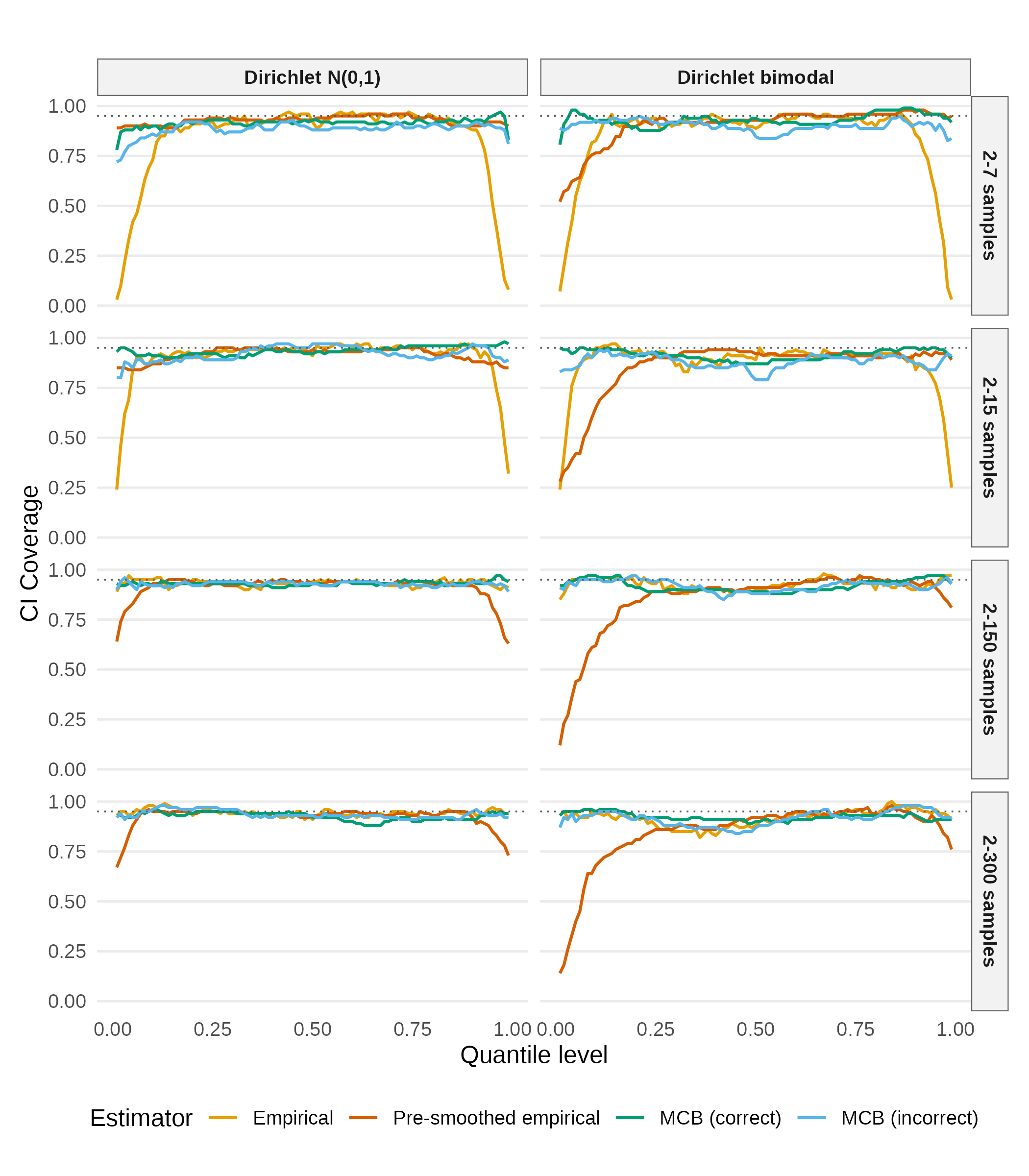}
    \captionof{figure}{95\% confidence interval coverage of the estimated group-specific Wasserstein barycenter quantile functions in Simulation Study 1 with 20 individuals per group. Results are shown for two data-generating settings: individual-level distributions generated from a Dirichlet process with a Normal$(0,1)$ base measure, and individual-level distributions generated from a Dirichlet process with a bimodal Normal mixture base measure. Results are shown across sequencing depths of 2--7, 2--15, 2--150, and 2--300. Each curve shows the 95\% confidence interval coverage of the estimated barycenter quantile function across quantile levels $t \in \{0.01,0.02,\ldots,0.99\}$. Estimators include the empirical Wasserstein barycenter, the pre-smoothed empirical Wasserstein barycenter, the correctly specified MCB estimator, and the misspecified spline-based MCB estimator.}
    \label{fig:sim1_bary_ci_coverage_n20}
\end{center}

\clearpage
\subsubsection{100 individuals per group}

\paragraph{Null scenario summary results}\leavevmode\par\vspace{0.5em}

\begin{center}
    \includegraphics[width=\textwidth,height=0.75\textheight,keepaspectratio]{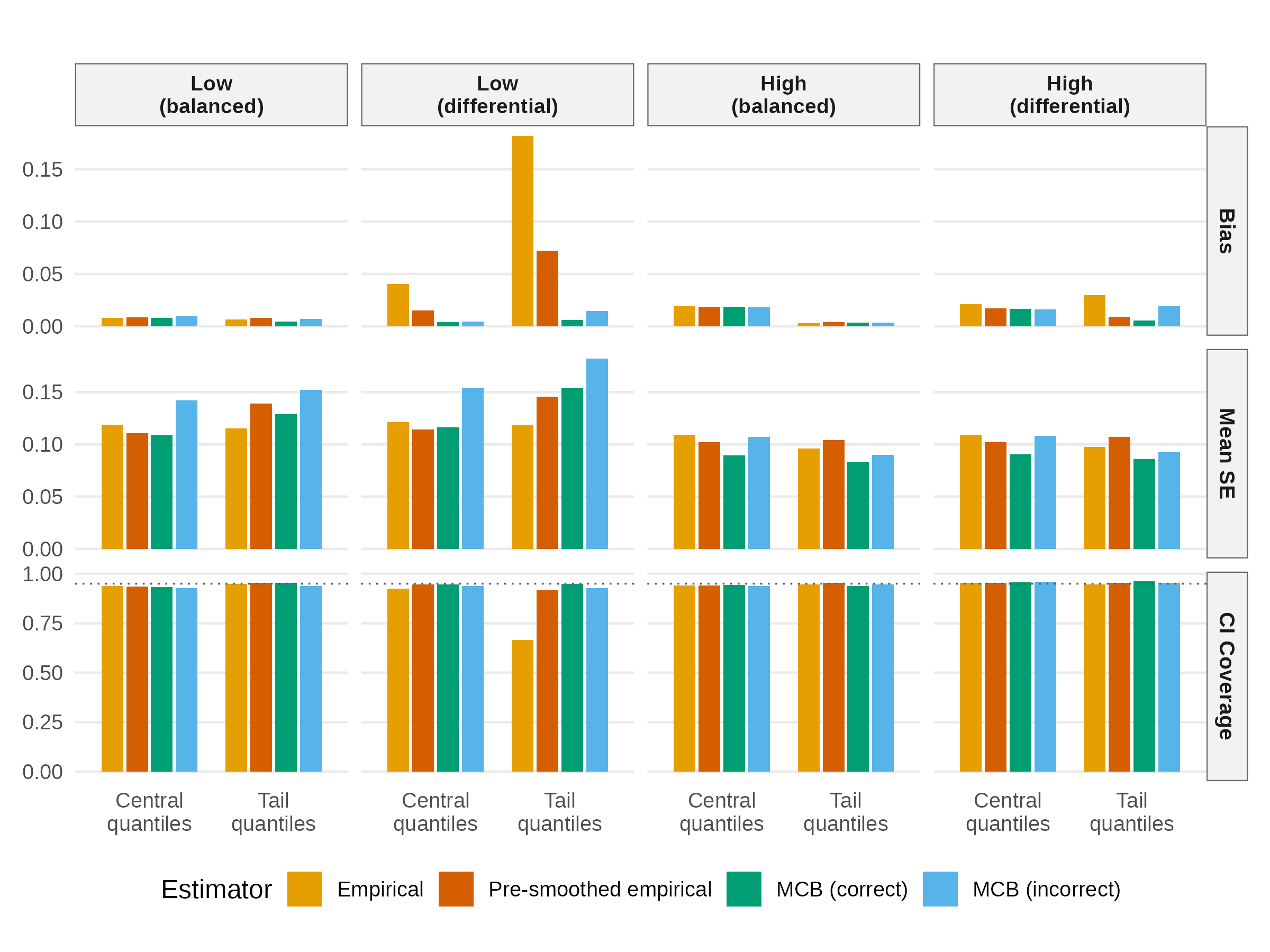}
\captionof{figure}{Summary of estimation of the difference in quantiles $\Delta(t)$ under the null scenario for Simulation Study 1 with 100 individuals per group. Results are shown across four sequencing depth settings: low balanced, low differential, high balanced, and high differential sequencing depth. For each setting, we compare the empirical Wasserstein barycenter, pre-smoothed empirical barycenter, correctly specified MCB estimator, and misspecified spline-based MCB estimator. Performance is summarized by average absolute bias, mean estimated standard error, and pointwise confidence interval coverage for $\Delta(t)$, averaged separately over the central quantiles ($0.06$--$0.94$) and tail quantiles ($0.01$--$0.05$ and $0.95$--$0.99$). The dotted horizontal line in the coverage panels indicates the nominal 95\% coverage level.}    \label{fig:sim1_summary_diff_null_n100}
\end{center}

\clearpage
\paragraph{Alternative scenario summary results}\leavevmode\par\vspace{0.5em}

\begin{center}
    \includegraphics[width=\textwidth,height=0.75\textheight,keepaspectratio]{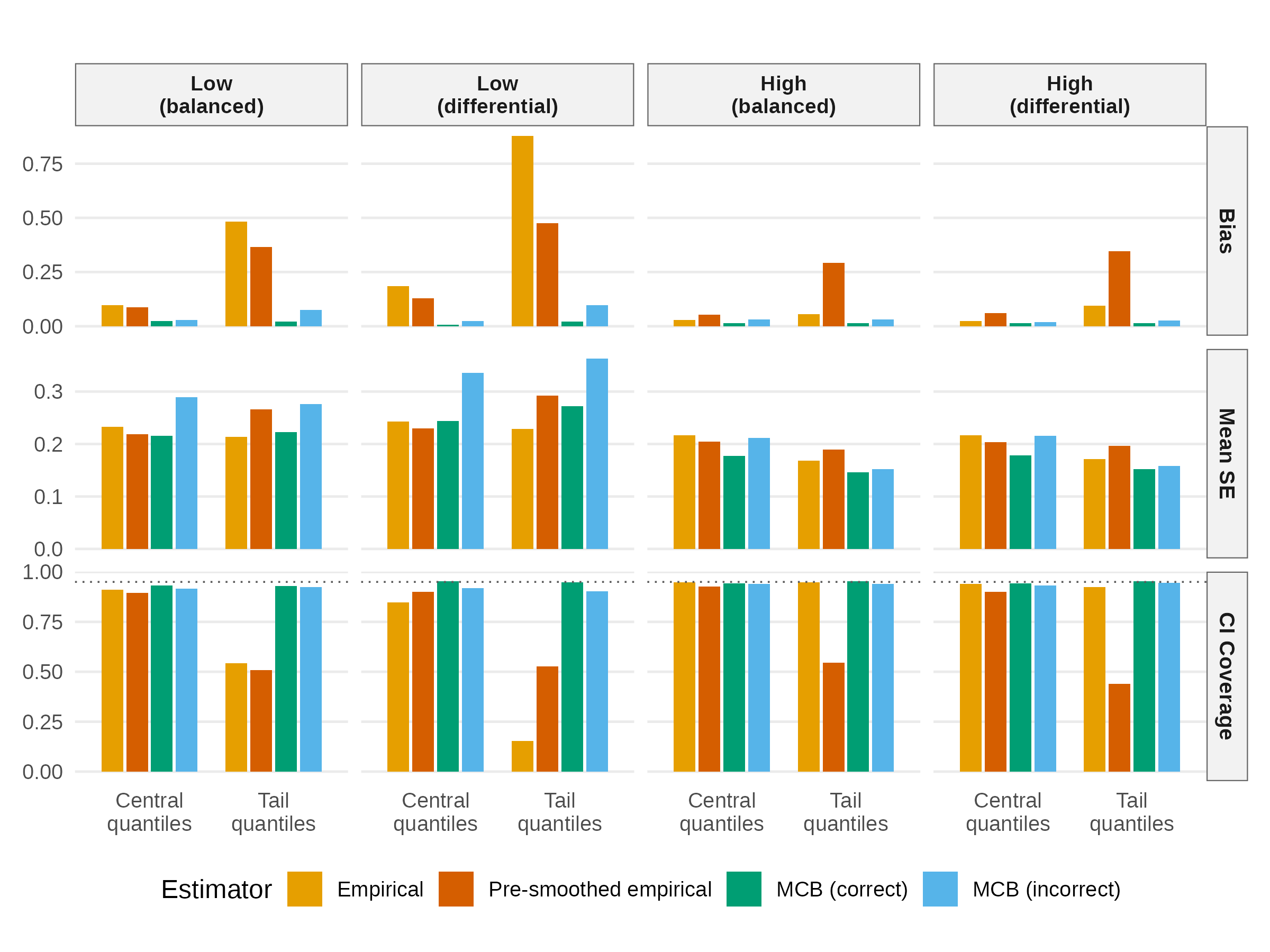}
\captionof{figure}{Summary of estimation of the difference in quantiles $\Delta(t)$ under the alternative scenario for Simulation Study 1 with 100 individuals per group. Results are shown across four sequencing depth settings: low balanced, low differential, high balanced, and high differential sequencing depth. For each setting, we compare the empirical Wasserstein barycenter, pre-smoothed empirical barycenter, correctly specified MCB estimator, and misspecified spline-based MCB estimator. Performance is summarized by average absolute bias, mean estimated standard error, and pointwise confidence interval coverage for $\Delta(t)$, averaged separately over the central quantiles ($0.06$--$0.94$) and tail quantiles ($0.01$--$0.05$ and $0.95$--$0.99$). The dotted horizontal line in the coverage panels indicates the nominal 95\% coverage level.}    \label{fig:sim1_summary_diff_alt_n100}
\end{center}

\clearpage
\paragraph{Quantile difference map $\Delta(t)$: pointwise bias}\leavevmode\par\vspace{0.5em}

\begin{center}
    \includegraphics[width=\textwidth,height=0.75\textheight,keepaspectratio]{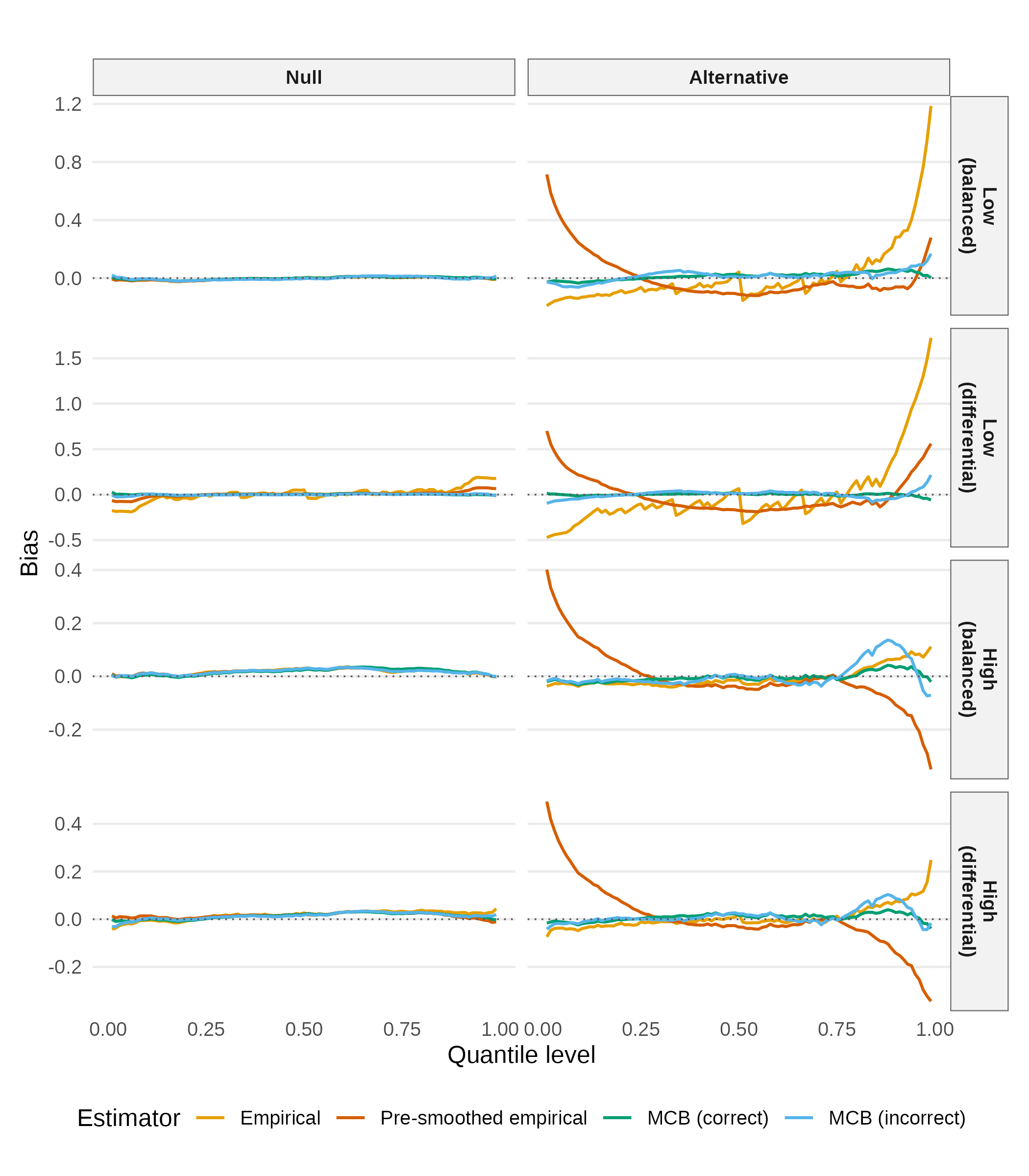}
    \captionof{figure}{Bias of the estimated quantile difference map $\widehat{\Delta}(t)$ in Simulation Study 1 with 100 individuals per group. Results are shown under the null and alternative sequence feature scenarios, across sequencing depth scenarios: low balanced, low differential, high balanced, and high differential. Each curve shows the bias of $\widehat{\Delta}(t)$ across quantile levels $t \in \{0.01,0.02,\ldots,0.99\}$. Estimators include the empirical Wasserstein barycenter, the pre-smoothed empirical Wasserstein barycenter, the correctly specified MCB estimator, and the misspecified spline-based MCB estimator.}
    \label{fig:sim1_diff_bias_n100}
\end{center}

\clearpage
\paragraph{Quantile difference map $\Delta(t)$: pointwise mean estimated standard error}\leavevmode\par\vspace{0.5em}

\begin{center}
    \includegraphics[width=\textwidth,height=0.75\textheight,keepaspectratio]{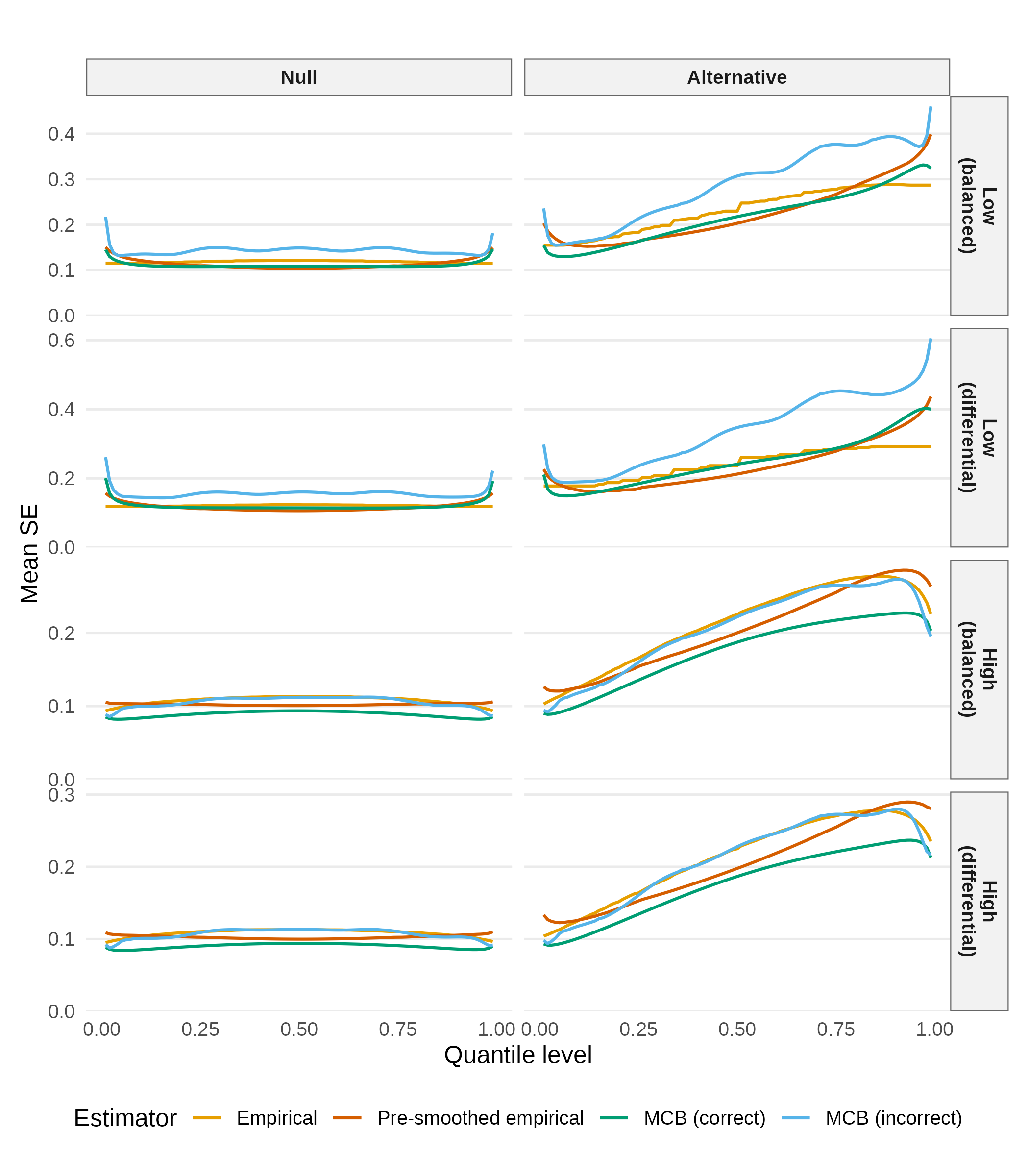}
    \captionof{figure}{Mean estimated standard error of the estimated quantile difference map $\widehat{\Delta}(t)$ in Simulation Study 1 with 100 individuals per group. Results are shown under the null and alternative sequence feature scenarios, across sequencing depth scenarios: low balanced, low differential, high balanced, and high differential. Each curve shows the mean estimated standard error of $\widehat{\Delta}(t)$ across quantile levels $t \in \{0.01,0.02,\ldots,0.99\}$. Estimators include the empirical Wasserstein barycenter, the pre-smoothed empirical Wasserstein barycenter, the correctly specified MCB estimator, and the misspecified spline-based MCB estimator.}
    \label{fig:sim1_diff_mean_se_n100}
\end{center}

\clearpage
\paragraph{Quantile difference map $\Delta(t)$: pointwise confidence coverage}\leavevmode\par\vspace{0.5em}

\begin{center}
    \includegraphics[width=\textwidth,height=0.75\textheight,keepaspectratio]{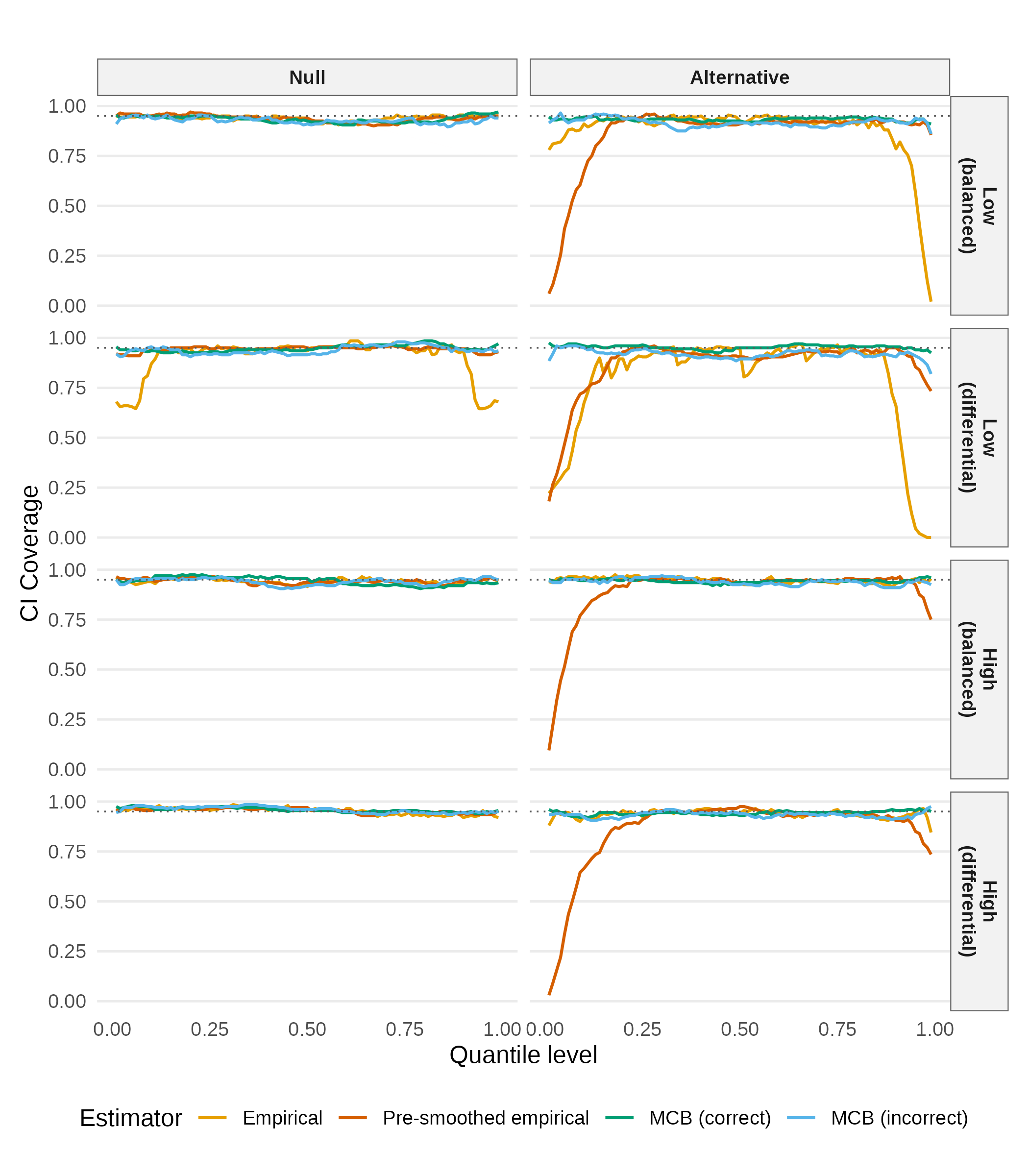}
    \captionof{figure}{95\% confidence interval coverage of the estimated quantile difference map $\widehat{\Delta}(t)$ in Simulation Study 1 with 100 individuals per group. Results are shown under the null and alternative sequence feature scenarios, across sequencing depth scenarios: low balanced, low differential, high balanced, and high differential. Each curve shows the 95\% confidence interval coverage of $\widehat{\Delta}(t)$ across quantile levels $t \in \{0.01,0.02,\ldots,0.99\}$. Estimators include the empirical Wasserstein barycenter, the pre-smoothed empirical Wasserstein barycenter, the correctly specified MCB estimator, and the misspecified spline-based MCB estimator.}
    \label{fig:sim1_diff_ci_coverage_n100}
\end{center}

\clearpage
\paragraph{Group-specific barycenter $\mu_g(t)$: pointwise bias}\leavevmode\par\vspace{0.5em}

\begin{center}
    \includegraphics[width=\textwidth,height=0.75\textheight,keepaspectratio]{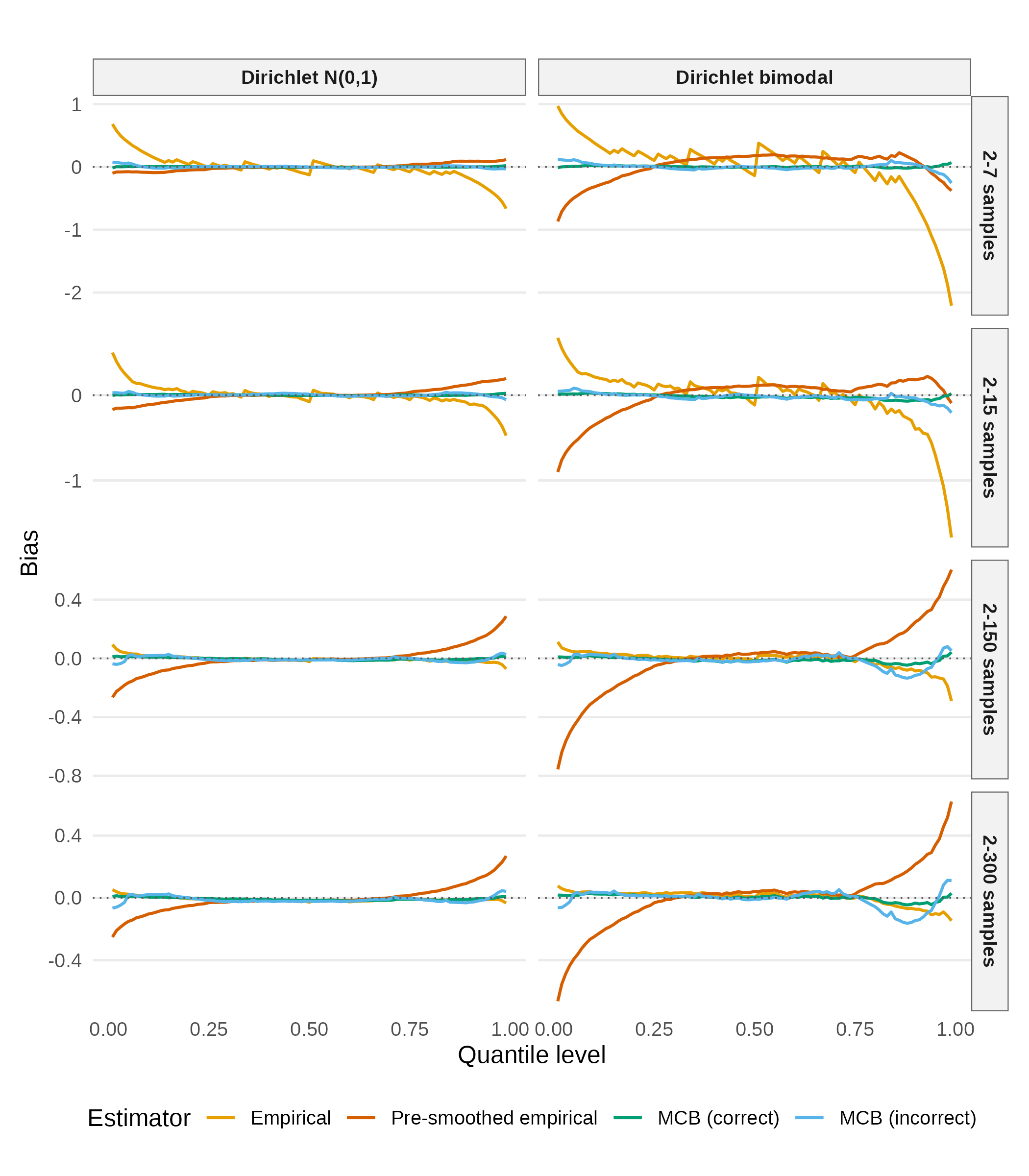}
    \captionof{figure}{Bias of the estimated group-specific Wasserstein barycenter quantile functions in Simulation Study 1 with 100 individuals per group. Results are shown for two data-generating settings: individual-level distributions generated from a Dirichlet process with a Normal$(0,1)$ base measure, and individual-level distributions generated from a Dirichlet process with a bimodal Normal mixture base measure. Results are shown across sequencing depths of 2--7, 2--15, 2--150, and 2--300. Each curve shows the bias of the estimated barycenter quantile function across quantile levels $t \in \{0.01,0.02,\ldots,0.99\}$. Estimators include the empirical Wasserstein barycenter, the pre-smoothed empirical Wasserstein barycenter, the correctly specified MCB estimator, and the misspecified spline-based MCB estimator.}
    \label{fig:sim1_bary_bias_n100}
\end{center}

\clearpage
\paragraph{Group-specific barycenter $\mu_g(t)$: pointwise mean estimated standard error}\leavevmode\par\vspace{0.5em}

\begin{center}
    \includegraphics[width=\textwidth,height=0.75\textheight,keepaspectratio]{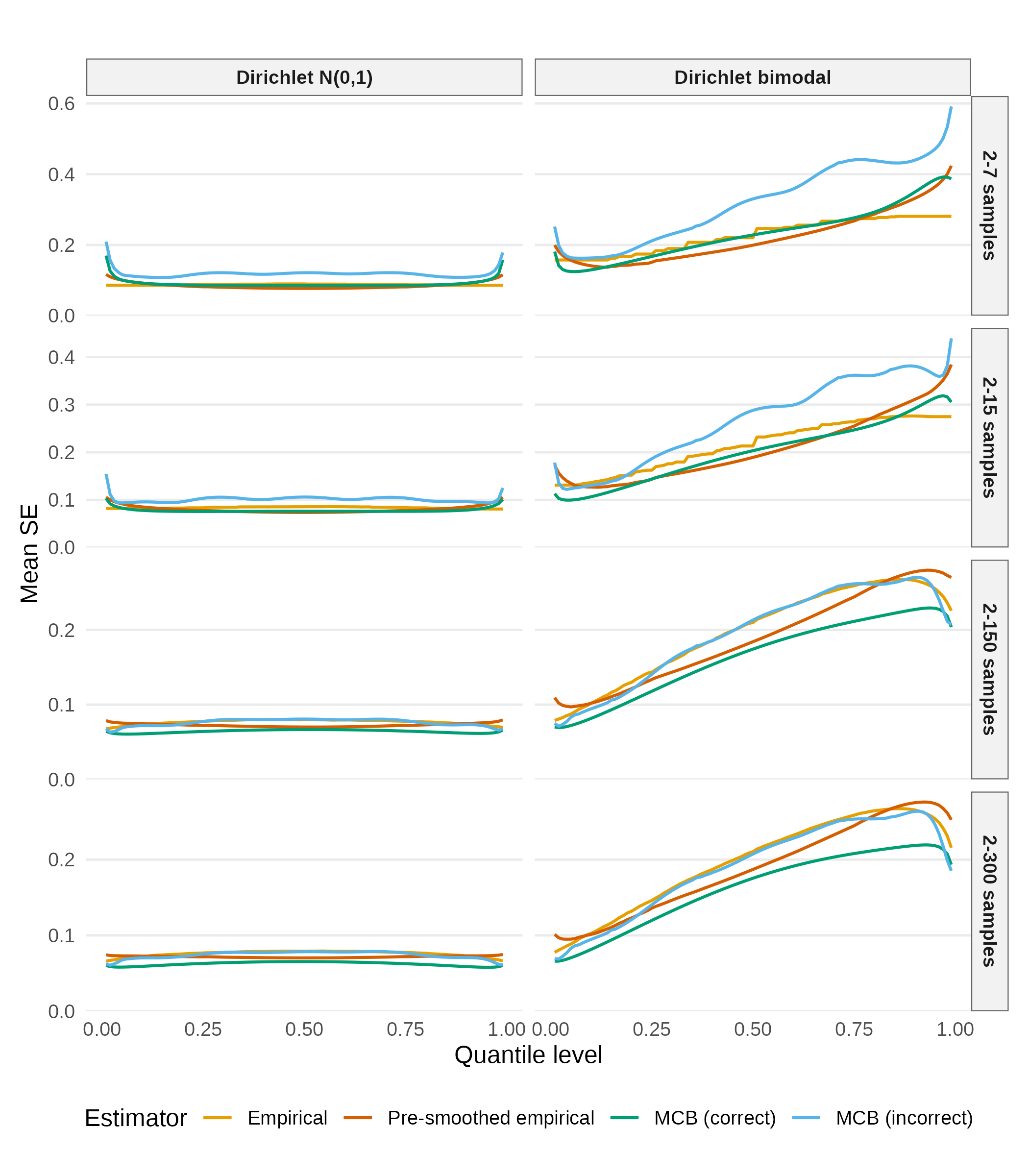}
    \captionof{figure}{Mean estimated standard error of the estimated group-specific Wasserstein barycenter quantile functions in Simulation Study 1 with 100 individuals per group. Results are shown for two data-generating settings: individual-level distributions generated from a Dirichlet process with a Normal$(0,1)$ base measure, and individual-level distributions generated from a Dirichlet process with a bimodal Normal mixture base measure. Results are shown across sequencing depths of 2--7, 2--15, 2--150, and 2--300. Each curve shows the mean estimated standard error of the estimated barycenter quantile function across quantile levels $t \in \{0.01,0.02,\ldots,0.99\}$. Estimators include the empirical Wasserstein barycenter, the pre-smoothed empirical Wasserstein barycenter, the correctly specified MCB estimator, and the misspecified spline-based MCB estimator.}
    \label{fig:sim1_bary_mean_se_n100}
\end{center}

\clearpage
\paragraph{Group-specific barycenter $\mu_g(t)$: pointwise confidence interval coverage}\leavevmode\par\vspace{0.5em}

\begin{center}
    \includegraphics[width=\textwidth,height=0.75\textheight,keepaspectratio]{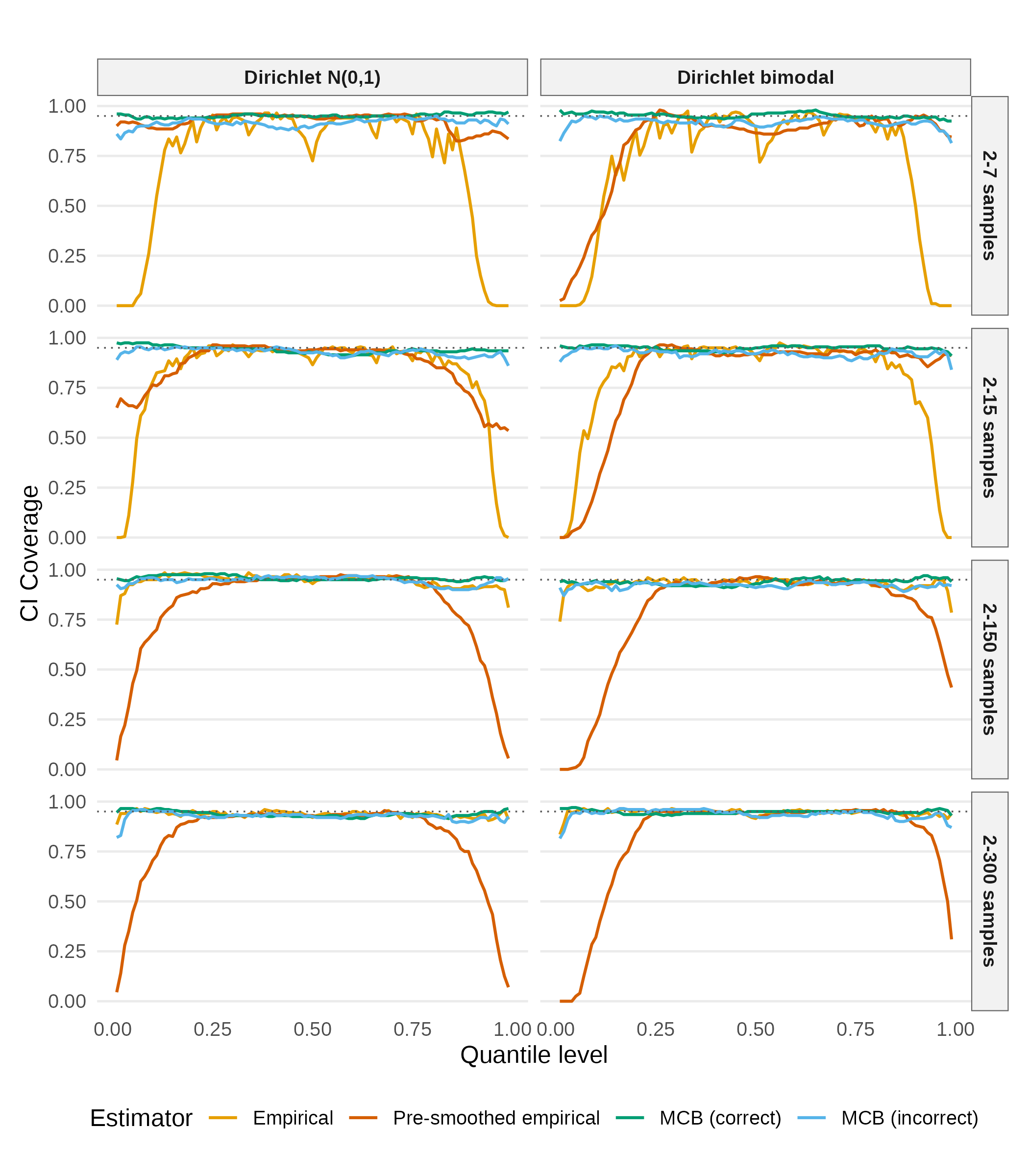}
    \captionof{figure}{95\% confidence interval coverage of the estimated group-specific Wasserstein barycenter quantile functions in Simulation Study 1 with 100 individuals per group. Results are shown for two data-generating settings: individual-level distributions generated from a Dirichlet process with a Normal$(0,1)$ base measure, and individual-level distributions generated from a Dirichlet process with a bimodal Normal mixture base measure. Results are shown across sequencing depths of 2--7, 2--15, 2--150, and 2--300. Each curve shows the 95\% confidence interval coverage of the estimated barycenter quantile function across quantile levels $t \in \{0.01,0.02,\ldots,0.99\}$. Estimators include the empirical Wasserstein barycenter, the pre-smoothed empirical Wasserstein barycenter, the correctly specified MCB estimator, and the misspecified spline-based MCB estimator.}
    \label{fig:sim1_bary_ci_coverage_n100}
\end{center}

\clearpage
\subsection{Simulation Study 2}\label{sec:additional_sim2}

In the main text, Simulation Study 2 was presented under the high balanced sequencing depth scenario. Here, we provide the corresponding results under the remaining sequencing depth scenarios: low balanced, low differential, and high differential. Results are shown below.

\newpage

\subsubsection{Low balanced}

\begin{center}
    \includegraphics[width=\textwidth,height=0.68\textheight,keepaspectratio]{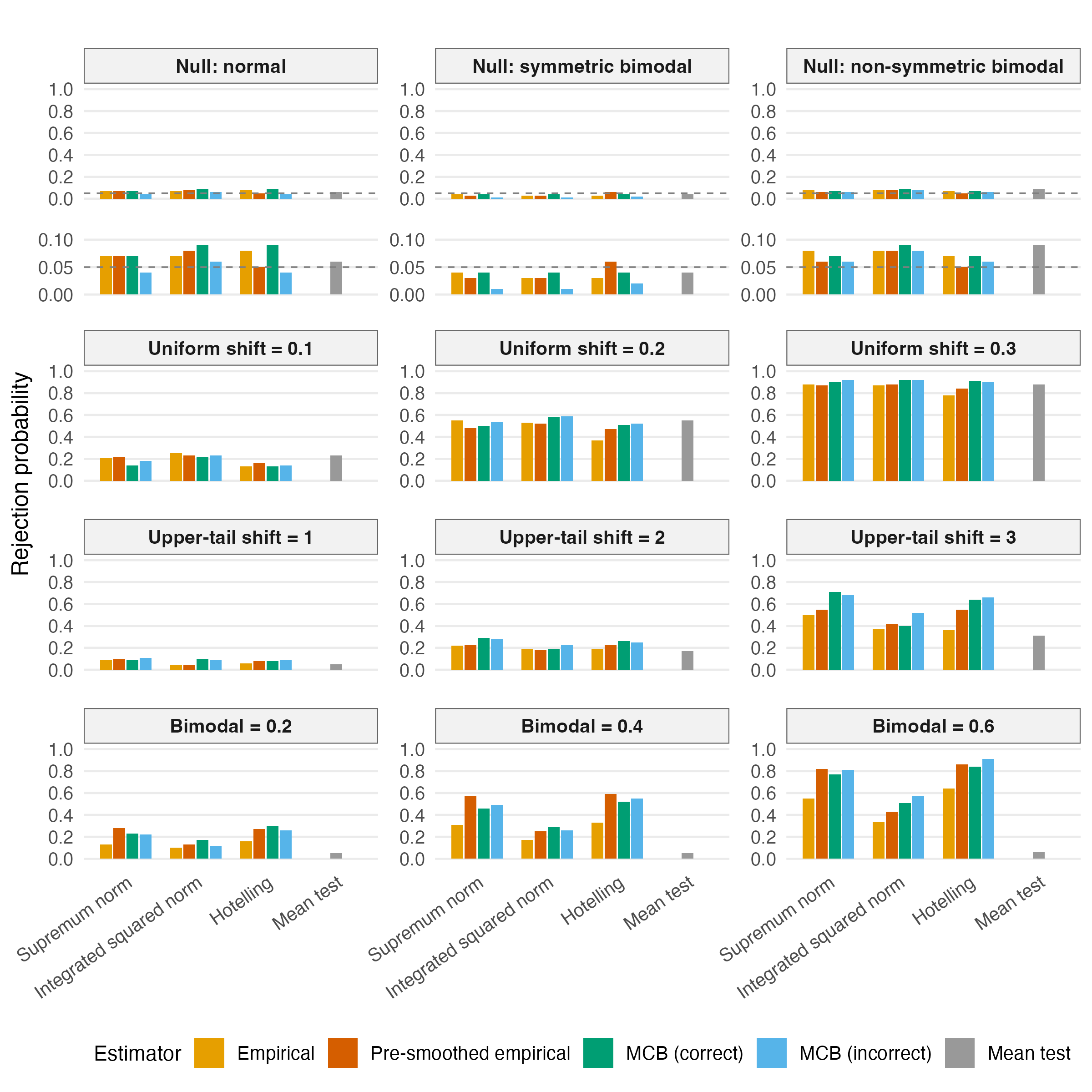}
    \captionof{figure}{Rejection probabilities across the three null scenarios and nine alternative scenarios for each estimator and test statistic, under low balanced sequencing depth. The four barycenter estimators considered are the empirical barycenter (“Empirical”), the pre-smoothed empirical barycenter (“Pre-smoothed empirical”), the correctly specified MCB estimator (“MCB (correct)”), and the misspecified spline-based MCB estimator (“MCB (incorrect)”). The mean test is based on averaging the sequence features within each individual and applying Welch's two-sample t-test to the resulting individual-level averages. The second row presents enlarged views of the three null scenarios, with the y-axis restricted to 0–0.10. The dashed horizontal line indicates the nominal significance level of 0.05.}
    \label{fig:sim2_power_low_balanced}
\end{center}

\clearpage

\subsubsection{Low differential}

\begin{center}
    \includegraphics[width=\textwidth,height=0.68\textheight,keepaspectratio]{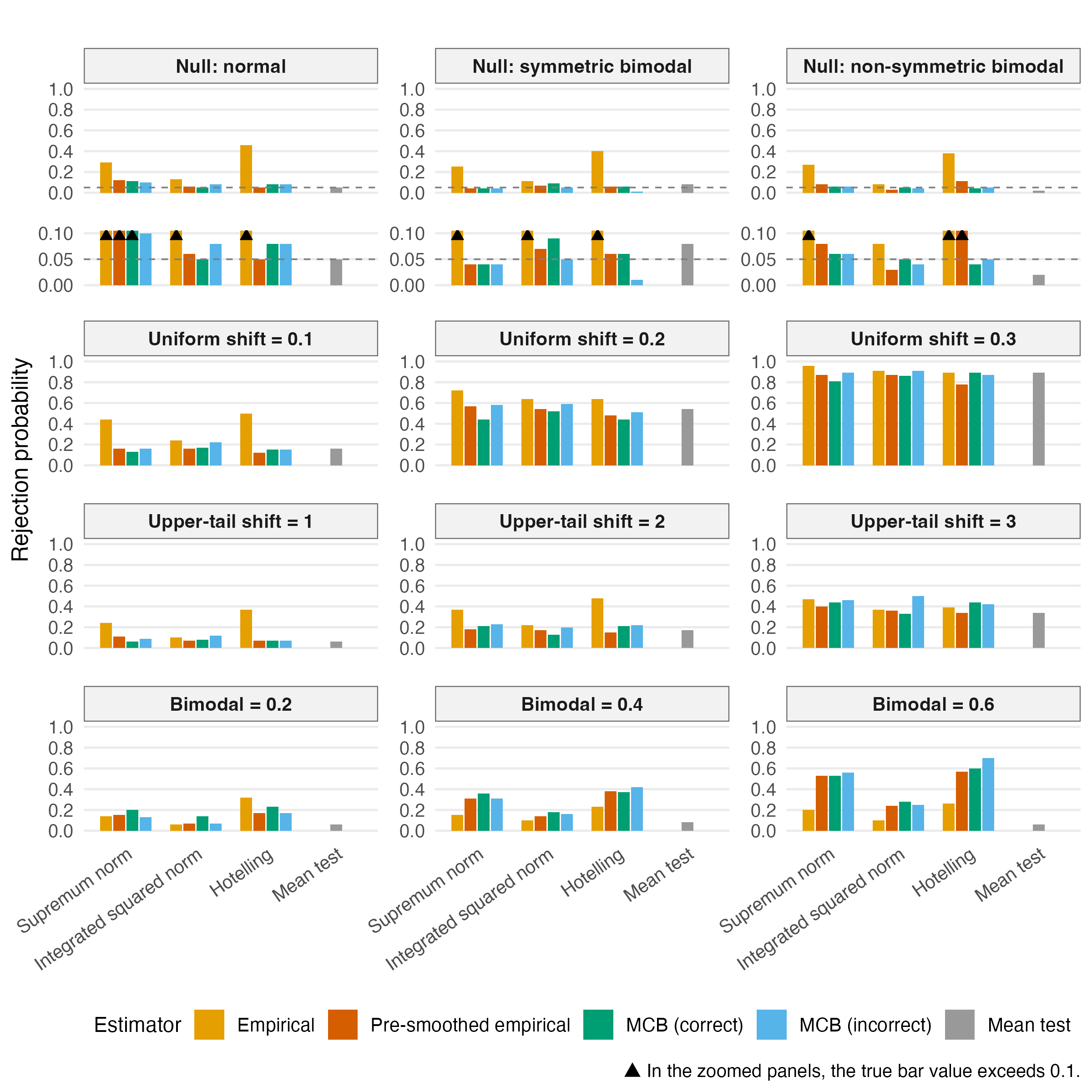}
    \captionof{figure}{Rejection probabilities across the three null scenarios and nine alternative scenarios for each estimator and test statistic, under low differential sequencing depth. The four barycenter estimators considered are the empirical barycenter (“Empirical”), the pre-smoothed empirical barycenter (“Pre-smoothed empirical”), the correctly specified MCB estimator (“MCB (correct)”), and the misspecified spline-based MCB estimator (“MCB (incorrect)”). The mean test is based on averaging the sequence features within each individual and applying Welch's two-sample t-test to the resulting individual-level averages. The second row presents enlarged views of the three null scenarios, with the y-axis restricted to 0–0.10. Black triangles in the zoomed-in null panels indicate rejection probabilities exceeding 0.10. The dashed horizontal line indicates the nominal significance level of 0.05.}
    \label{fig:sim2_power_low_differential}
\end{center}

\clearpage

\subsubsection{High differential}

\begin{center}
    \includegraphics[width=\textwidth,height=0.68\textheight,keepaspectratio]{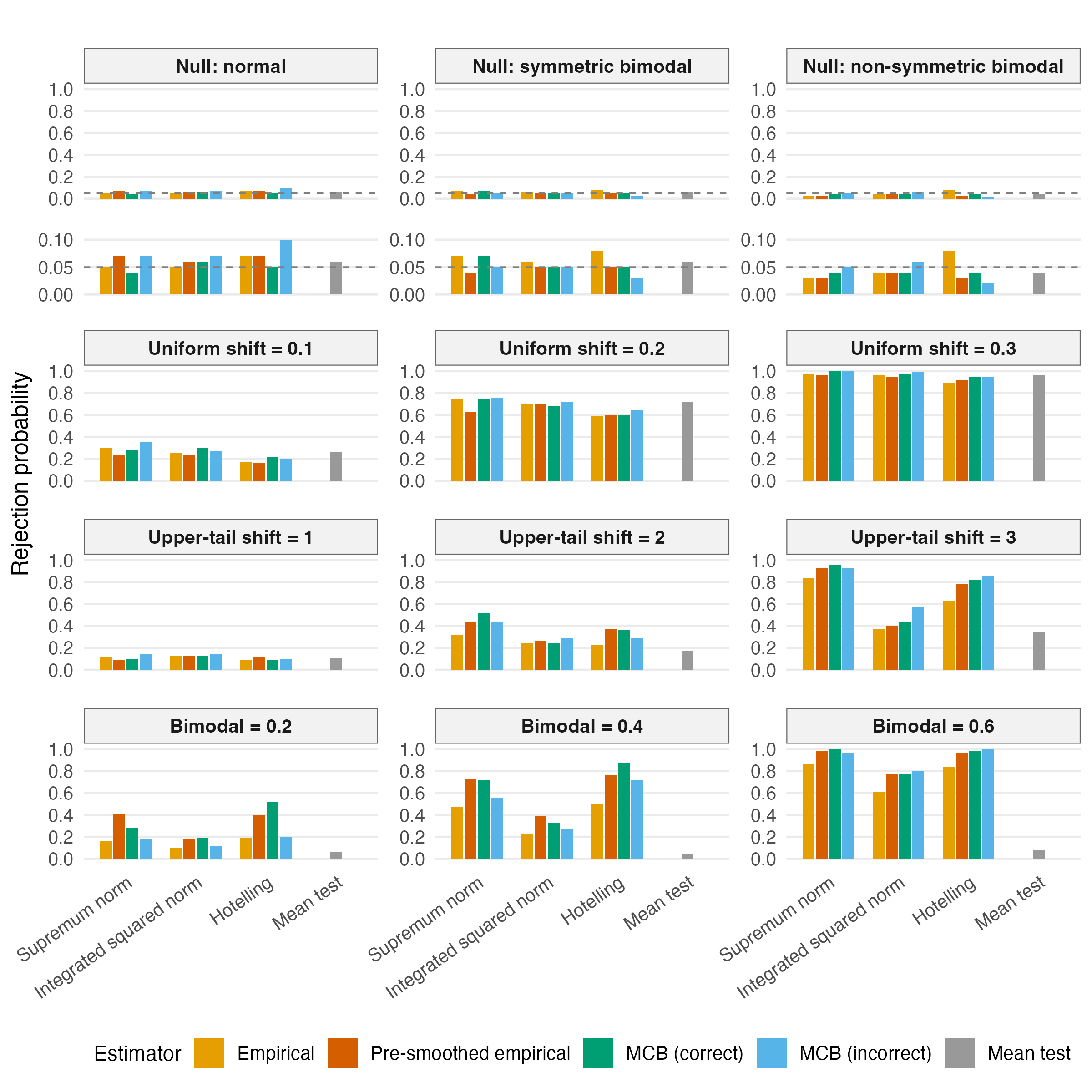}
    \captionof{figure}{Rejection probabilities across the three null scenarios and nine alternative scenarios for each estimator and test statistic, under high differential sequencing depth. The four barycenter estimators considered are the empirical barycenter (“Empirical”), the pre-smoothed empirical barycenter (“Pre-smoothed empirical”), the correctly specified MCB estimator (“MCB (correct)”), and the misspecified spline-based MCB estimator (“MCB (incorrect)”). The mean test is based on averaging the sequence features within each individual and applying Welch's two-sample t-test to the resulting individual-level averages. The second row presents enlarged views of the three null scenarios, with the y-axis restricted to 0–0.10. The dashed horizontal line indicates the nominal significance level of 0.05.}
    \label{fig:sim2_power_high_differential}
\end{center}

\end{document}